\documentclass{article}
\usepackage[T1]{fontenc}
\usepackage{xcolor}
\usepackage[margin=1in]{geometry}
\usepackage{amsfonts,amsmath,amsthm,amssymb,mathtools}  %

\mathtoolsset{centercolon}
\usepackage{xfrac,nicefrac}
\usepackage{mathdots}
\usepackage{mleftright}  %
\let\left\mleft
\let\right\mright
\usepackage{todonotes}

\usepackage{xspace}
\xspaceaddexceptions{]\}}  %
\usepackage{regexpatch}

\usepackage{bm,bbm,dsfont}  %
\usepackage{caption}
\usepackage[normalem]{ulem}
\usepackage{enumitem}

\usepackage{graphicx}
\usepackage{float}
\usepackage{subcaption}  %
\usepackage{tcolorbox}
\usepackage{tikz}
\usetikzlibrary{decorations.pathreplacing}
\usetikzlibrary{calc}
\usetikzlibrary{positioning}
\usetikzlibrary{arrows.meta}
\usetikzlibrary{fit}

\usepackage[linesnumbered,boxed,ruled,vlined]{algorithm2e}

\SetCommentSty{mycommfont}

\usepackage{thmtools,thm-restate}
\usepackage[colorlinks,citecolor=blue,linkcolor=blue,urlcolor=red]{hyperref}

\theoremstyle{plain}

\newtheorem{theorem}{Theorem}[section]  %
\newtheorem{lemma}[theorem]{Lemma}
\newtheorem{fact}[theorem]{Fact}

\newtheorem{corollary}[theorem]{Corollary}

\newtheorem{claim}[theorem]{Claim}

\theoremstyle{definition}  %
\newtheorem{definition}[theorem]{Definition}
\newtheorem{remark}[theorem]{Remark}

\usepackage[capitalise]{cleveref}
\crefname{algocf}{Algorithm}{Algorithms}
\Crefname{algocf}{Algorithm}{Algorithms}
\crefname{claim}{Claim}{Claims}
\Crefname{claim}{Claim}{Claims}
\Crefname{fact}{Fact}{Facts}

\usepackage{xurl}

\newfloat{Distribution}{htbp}{loa}
\crefname{Distribution}{Distribution}{Distributions}
\Crefname{Distribution}{Distribution}{Distributions}
\newfloat{Protocol}{htbp}{loa}
\crefname{Protocol}{Protocol}{Protocols}
\Crefname{Protocol}{Protocol}{Protocols}

\SetKwProg{Function}{Function}{:}{}

\SetKwProg{CodeBlock}{}{}{}
\SetKwProg{Repeat}{Repeat}{:}{}

\DeclarePairedDelimiter{\bk}{(}{)}
\DeclarePairedDelimiter{\Bk}{[}{]}
\DeclarePairedDelimiter{\BK}{\{}{\}}

\DeclarePairedDelimiterX\mysetbase[2]{\lbrace}{\rbrace}{#1\,\delimsize\vert\,#2}
\NewDocumentCommand{\myset}{sO{}m m}{%
  \IfBooleanTF{#1}%
    {\mysetbase*{#3}{#4}}%
    {\mysetbase[#2]{#3}{#4}}%
}

\DeclareMathOperator*{\E}{\mathbb{E}}

\let\Pr\PrAux
\DeclareMathOperator{\poly}{poly}

\renewcommand{\d}{\mathrm{d}}

\newcommand{\defeq}{\coloneqq}
\newcommand{\eps}{\varepsilon}

\renewcommand{\emptyset}{\varnothing}
\renewcommand{\epsilon}{\eps}

\newcommand{\defn}[1]{\emph{\boldmath\textbf{#1}}}

\usepackage{regexpatch}
\makeatletter
\xpatchcmd\thmt@restatable{%
\csname #2\@xa\endcsname\ifx\@nx#1\@nx\else[{#1}]\fi
}{%
\ifthmt@thisistheone
\csname #2\@xa\endcsname\ifx\@nx#1\@nx\else[{#1}]\fi
\else
\csname #2\@xa\endcsname[{Restated}]
\fi}{}{}
\makeatother

\usetikzlibrary{shapes.geometric}

\newcommand{\gap}{\mathrm{gap}}
\newcommand{\gapbar}{\overline{\mathrm{gap}}}
\newcommand{\trie}{\mathrm{trie}}
\newcommand{\enc}{\mathrm{enc}}

\title{Optimal Time-Space Tradeoff for Dynamic Difference-Encoded Dictionaries}
\author{
  Guy E. Blelloch\thanks{Carnegie Mellon University. \texttt{guyb@cs.cmu.edu}}
  \and
  Yang Hu\thanks{Tsinghua University. \texttt{y-hu22@mails.tsinghua.edu.cn}}
  \and
  William Kuszmaul\thanks{Carnegie Mellon University. \texttt{kuszmaul@cmu.edu}}
  \and
  Jingxun Liang\thanks{Carnegie Mellon University. \texttt{jingxunl@andrew.cmu.edu}}
  \and
  Renfei Zhou\thanks{Carnegie Mellon University. \texttt{renfeiz@andrew.cmu.edu}}
}
\date{}

\begin{document}

\maketitle

\begin{abstract}
The dynamic dictionary is a fundamental data structure that maintains a set $S\subset [U]$ of size $n$ (we assume $n=U^{1-\Theta(1)}$), supporting insertions, deletions and membership queries. Previous works mostly focused on constructing dictionaries that support operations in $O(1)$ time and use space as close to the \emph{information-theoretic bound} of $\log\binom{U}{n}$ bits as possible.

In this paper, we study \emph{difference-encoded} dictionaries, which are dictionaries that use space close to the gap entropy $\text{gap}(S)\coloneqq\sum_{i=2}^{|S|}\bigl(\lceil\log(x_i-x_{i-1}+1)\rceil+1\bigr)$ bits to store the set $S=\{x_1<\dots<x_n\}$. On many real-world datasets where the keys are clustered, we have $\text{gap}(S)\ll \log\binom{U}{n}$, making difference-encoded dictionaries more favorable than standard dictionaries in practice.

Prior to this work, the best dynamic difference-encoded dictionary is by Blandford and Blelloch [SODA'04], whose construction supports operations in $O(\log n)$ time and uses $O(\text{gap}(S))$ bits of space. In the static case, Gupta, Hon, Shah and Vitter [DCC'06] presented a dictionary that uses
$$
\text{gap}(S)+O(n\log\log U)
$$
bits of space and supports membership queries in $O(\log\log n)$ time.

In this work, we go beyond these bounds and fully settle the optimal time-space tradeoff for difference-encoded dictionaries. For an arbitrary parameter $0<\varepsilon<1/4$, we construct a dynamic dictionary that supports operations in $O(\log\varepsilon^{-1}/\log\log\varepsilon^{-1})$ expected amortized time and uses
$$
    \text{gap}(S)\cdot(1+O(\varepsilon))+O\left(n\log\frac{\text{gap}(S)}{n}\right)
$$
bits of space. We also prove a matching lower bound, showing that our time-space tradeoff is optimal even in the \emph{static} case.
\end{abstract}

\section{Introduction}

An \emph{unordered dictionary} is a data structure that supports insertions, deletions, and membership queries on some dynamically changing set $S \subseteq [U]$ of $n$ keys (where $n$ may change over time).
The classic solution to this problem is to use a hash table, which supports operations in $O(1)$ time while using $O(n \log U)$ bits of space. In this paper, as in most previous work on the subject \cite{fiat1988nonoblivious,bender2023iceberg,bender2022optimal,kuszmaul2022hash,li2023tight,li2024dynamic,DBLP:conf/stoc/HuLYZZ25,kuszmaul2025tight,DBLP:conf/soda/KuszmaulLZ26}, we will be most interested in the parameter regime $\log n = (1 - \Theta(1))\log U$.

Starting in the mid 1950s \cite{peterson1957addressing,brent1973reducing,celis1985robin,munro1986techniques}, there has been a long line of work on building hash tables that are as space-efficient as possible. Early works in this area focused on achieving space as close to 
\begin{equation}n \log U
    \label{eq:plainbound}
\end{equation}
bits as possible, while still supporting efficient operations (see, e.g., \cite{munro1986techniques} for an early survey, or \cite{bender2024tight,farach2024optimal,bender2025optimal} for more recent discussions). However, starting in the late-1990s \cite{pagh1999low,raman2003succinct,demaine2006dictionariis}, researchers began to target a stronger goal, aiming to achieve space close to the \emph{information-theoretic optimum} of
\begin{equation}\log \binom{U}{n} = n \log U - n \log n + O(n)
    \label{eq:succinctbound}
\end{equation}
bits. After a long line of work \cite{pagh1999low,raman2003succinct,demaine2006dictionariis,arbitman2010backyard,bender2023iceberg,bender2022optimal,li2023tight,li2024dynamic}, it is known how to construct hash tables that use space within a $(1 + o(1))$ factor of \eqref{eq:succinctbound}, while supporting $O(1)$-time operations, and it is even known what the optimal $1 + o(1)$ factor is \cite{li2023tight,li2024dynamic}. 

Of course, for real-world data sets, the space bound given by \eqref{eq:succinctbound} may still be far from optimal. In practice, data sets tend to exhibit clustering, where some parts of the key space $[U]$ are used much more frequently than others. For such data sets, one may be able to achieve a much smaller space bound by storing the \emph{differences} between keys, rather than the keys themselves. This raises a natural question: can one hope to build a hash table that uses space close to the \defn{gap entropy}
\begin{equation}
    \gap(S)\ \coloneqq\ \sum_{i=2}^{n} \bigl(\bigl\lceil \log(x_i-x_{i-1}+1)\bigr\rceil+1\bigr),
    \label{eq:deltaencodingbound}
\end{equation}
where $x_1 < x_2 < \ldots < x_{n}$ denotes the keys in $S$ in sorted order?

Intuitively, if one wants to achieve a bound close to the gap entropy, one should use an \emph{ordered dictionary}, such as a binary search tree, a van-Emde-Boas tree \cite{DBLP:journals/mst/BoasKZ77}, a y-fast trie \cite{DBLP:journals/ipl/Willard83}, etc. But these data structures require super-constant time per operation \cite{DBLP:journals/combinatorica/Ajtai88}. Hash tables, on the other hand, are designed to achieve close to $O(1)$-time operations, but do so by treating the data as fundamentally unordered. This barrier has led previous researchers \cite{gupta2007compressed} to conjecture that any data structure which measures its space in terms of $\gap(S)$ cannot hope to support $O(1)$-time operations. 

\paragraph{This Paper: Optimal difference-encoded hash tables. }In this paper, we show that it \emph{is}, in fact, possible to use space close to $\gap(S)$ in an $O(1)$-time data structure.

\begin{restatable}{theorem}{MainDict}
    \label{thm:main}
    Let $U$ be an integer. Let $\Omega(\log\log U/\log U)\le \varepsilon<1/4$ be a parameter\footnote{The constraint on $\varepsilon$ may seem strange. However, it turns out that $\Omega(\log\log U/\log U)\le \varepsilon<1/4$ is the only interesting regime. Indeed, the data structure constructed for $\varepsilon=\log\log U/\log U$ also satisfies the time and space constraints for any $\varepsilon<\log\log U/\log U$: For the time constraint, this is straightforward because $\log\varepsilon^{-1}/\log\log\varepsilon^{-1}$ is decreasing in $\varepsilon$ for $\varepsilon=o(1)$. For the space constraint, when $\varepsilon=O(\log\log U/\log U)$, we have that $|S|\log(\gap(S)/|S|)=\Omega(\varepsilon\cdot \gap(S))$, so the actual value of $\varepsilon$ does not affect the space bound \eqref{eq:actualbound}. To see this, note that $\gap(S)/|S|\le \log U$, which implies that $\log(\gap(S)/|S|)/\log\log(\gap(S)/|S|)=O(\log \log U/\log\log \log U)$ by monotonicity of $\log x/\log\log x$.}. There exists a data structure in the word RAM model with word size $w=\Theta(\log U)$ that maintains a key set $S\subseteq [U]$. In $O(\log\varepsilon^{-1}/\log\log\varepsilon^{-1})$ expected amortized time, the data structure supports the following operations:
    \begin{itemize}
        \item \textbf{Insertion}: Insert a key $x$ to $S$.
        \item \textbf{Deletion}: Delete a key $x$ from $S$.
        \item \textbf{Membership}: Given $x\in [U]$, output whether $x\in S$.
    \end{itemize} The data structure uses
    \begin{align}
        \label{eq:actualbound}
        \gap(S)+O(\varepsilon \cdot\gap(S))+O\bk*{|S|\log \frac{\gap(S)}{|S|}}+O\bk*{U^{\delta}}
    \end{align}
    bits of space, where $\delta>0$ is a small constant parameter of our choice.
\end{restatable}

Note that the $O(|S|\log (\gap(S)/|S|))$ term in \eqref{eq:actualbound} is, in general, information-theoretically necessary, so that the data structure can encode the \emph{number of bits} needed to encode each difference $x_i - x_{i - 1}$.

We also prove that, somewhat surprisingly, the space-time tradeoff above is optimal. We give a cell-probe lower bound showing that any hash table achieving the space bound in \eqref{eq:actualbound} must incur an expected query time of $\Omega(\log \epsilon^{-1}/\log \log \epsilon^{-1})$. This lower bound applies even in the static setting, where the set $S$ of keys is fixed in advance and never changes.  

In the process of proving our results, we also revisit an even more basic question: how to encode \emph{binary search trees} using space close to \eqref{eq:deltaencodingbound}. We show how to take any rotation-based binary search tree and encode it in space
\begin{equation}\gap(S) + O(n \log \log U)
    \label{eq:treebound}
\end{equation}
(not counting the auxiliary information) without changing the time per operation. Note that, as we discuss in \cref{sec:binary-tree}, a solution that uses space $ \approx \log |c - p|$ bits for each parent-child pair $(p, c)$ may be up to a factor of $\Omega(\log n)$ away from matching \eqref{eq:treebound} \cite{blandford2008compact,gupta2007compressed}.

As an immediate corollary, we get dynamic predecessor/successor and rank/select data structures that achieve optimal time bounds while using space given by \eqref{eq:actualbound}.

\begin{restatable}{corollary}{CorMainPredecessor}
    The data structure of \cref{thm:main} can be modified to support insertion/deletion and predecessor/successor queries in
    \begin{align*}
        O\bk*{\frac{\log\varepsilon^{-1}}{\log\log\varepsilon^{-1}}+\log\log U}
    \end{align*}
    time and membership queries in $O(\log\varepsilon^{-1}/\log\log\varepsilon^{-1})$ time. The space usage is the same as \eqref{eq:actualbound}.
\end{restatable}

\begin{restatable}{corollary}{CorMainRank}
    The data structure of \cref{thm:main} can be modified to support insertion/deletion and rank/select queries in
    \begin{align*}
        O\bk*{\frac{\log\varepsilon^{-1}}{\log\log\varepsilon^{-1}}+\frac{\log|S|}{\log\log U}}
    \end{align*}
    time, predecessor/successor queries in
    \begin{align*}
        O\bk*{\frac{\log\varepsilon^{-1}}{\log\log\varepsilon^{-1}}+\log\log U}
    \end{align*}
    time and membership queries in $O(\log\varepsilon^{-1}/\log\log\varepsilon^{-1})$ time. The space usage is the same as \eqref{eq:actualbound}.
    \label{cor:foo}
\end{restatable}

These data structures settle in the affirmative an open question previously posed by Gupta, Hon, Shah, and Vitter \cite{gupta2007compressed} as to whether such data structures exist.

\paragraph{Paper outline. }
In \cref{sec:warm-up}, we present a warm-up version of our data structure (\cref{thm:main-warmup}). This data structure consists of two parts: We partition the key set into chunks of $\poly\log U$ keys each. Then, we use a distributor (\cref{thm:distributor}) to map keys to their chunks, and design dictionaries that are efficient when the number of keys is only $\poly\log U$ (\cref{lem:polylog-dict-warmup}). These constructions are suboptimal: The distributor uses $O(|S|\log\log U)$ bits of space, and the small dictionaries take $O(\varepsilon^{-1})$ time to perform operations when they are allowed to use
\begin{align}
    \label{equ:warm-up-bound}
    \gap(S)+O(\varepsilon\cdot \gap(S))+O\bk*{|S|\log\frac{\gap(S)}{|S|}}
\end{align}bits of space.

In \cref{sec:binary-tree}, we improve our construction of small dictionaries by showing a difference-encoding scheme for binary search trees (\cref{lem:binary-tree}), reducing the query time to $O(\log\varepsilon^{-1})$. This is further improved in \cref{sec:trie-plus-btree}, where we use difference-encoded B-trees to maintain a trie structure on the keys (\cref{lem:trie-plus-btree}), giving us the optimal time-space tradeoff for small key sets (i.e., using $O(\log\varepsilon^{-1}/\log\log\varepsilon^{-1})$ query time and space equal to \eqref{equ:warm-up-bound}). In \cref{sec:better-distributor}, we optimize the distributor (\cref{thm:better-distributor}), improving its space usage from $O(|S|\log\log U)$ to
\begin{align*}
    O(\varepsilon\cdot \gap(S))+O\bk*{|S|\log\frac{\gap(S)}{|S|}}
\end{align*} bits (while slightly increasing the time cost). Finally, in \cref{sec:final}, we combine the improved data structures \cref{lem:trie-plus-btree} and \cref{thm:better-distributor}, and prove the main result \cref{thm:main} using the same framework as the warm-up.

In \cref{sec:lb}, we prove a matching lower bound, showing that our dynamic data structure is tight even in the static case.

\section{Related Works}

For dictionaries whose space is analyzed in terms of the gaps of the stored set, Blandford and Blelloch~\cite{blandford2004ordered,blandford2008compact} gave a general approach for replacing fixed-length representations by compact variable-length encodings. They partition the key set $S$ into contiguous subsets $\{S_i\}$, such that each subset has total gap entropy $\Theta(\log U)$ (so the number of subsets is $O(\gap(S)/\log U)$). Then, they maintain the set $S$ using $\Theta(\log U)$ bits of space per subset $S_i$. As a result, their framework yields dynamic dictionaries using
\[
  O(\gap(S))+O(U^{\alpha}\log U)
\]
bits of space ($0<\alpha<1$ is an arbitrary constant), while supporting rank/select queries in $O(\log n)$ time. Grossi and Sadakane~\cite{grossi2006squeezing} further reduced the space for static dictionaries to
\[
  \gap(S)+O\bk*{\frac{U\log\log U}{\log U}}
\]
bits (answering rank/select queries in constant time), by noting that when $n$ is very close to $U$, we can maintain each subset $S_i$ using $\gap(S_i)+o(\log U)$ bits of space. Gupta, Hon, Shah, and Vitter~\cite{gupta2007compressed} obtained a static dictionary on $n$ keys from $[U]$ using
\begin{align}
  \label{equ:gupta-space}
  \gap(S)+O\bk*{\frac{n\log(U/n)}{\log n}}+O\bk*{n\log\log(U/n)}
\end{align}
bits of space, supporting rank queries in \begin{align*}
    AT(U,n)=\min\BK*{\sqrt{\frac{\log n}{\log\log n}},\log\log n\cdot \frac{\log\log U}{\log\log\log U},\log\log n+\frac{\log n}{\log\log U}}
\end{align*}time and select queries in $O(\log\log n)$ time. In the standard parameter regime of $\log U = (1 + \Theta(1)) \log n$ (and for $\epsilon \ge 1/\log n$), these time bounds are the same as those in Corollary \ref{cor:foo}. 

Gupta, Hon, Shah, and Vitter~\cite{gupta2007compressed} also pose several open questions that we resolve. They conjecture that any solution using $O(\gap(S))$ space must incur $\omega(1)$-time membership queries (this turns out not to be true), and they leave open whether a dynamic data structure can hope to achieve good rank/select/update time bounds (a question that we answer in the affirmative with Corollary \ref{cor:foo}).

An interesting feature of our work is how, to avoid time bounds depending on $n$, we are able to combine techniques from a variety of settings in which one can interact with ordered data very efficiently, without relying on any settings in which one cannot. Indeed, our warmup result in Section \ref{sec:warm-up} is achieved in large part by strategically combining together the following three facts within a larger data structure: (1) that ordered dictionaries of $\poly(w)$ $w$-bit keys can support $O(1)$ time operations (via fusion nodes \cite{fredman1990transdichotomous,fredman1993surpassing,patrascu2014dynamic}); (2) that integer sorting of $n$ $\Theta(\log n)$-bit keys is linear time (via radix sort \cite{seward1954information}); and (3) that \emph{static} ordered data structures which answer rank queries on a specific set of $n$ keys can be implemented very space efficiently with $O(1)$-time queries (via monotone minimal perfect hashing \cite{belazzougui2009monotone}). We note that the last of these results is used only implicitly, in that we apply the high-level technique \cite{belazzougui2009monotone} (mapping each key $k$ to a prefix length $\ell$ such that the closest pivot $p$ to $k$ matches $k$ on those $\ell$ bits), but in the context of a dynamic data structure.

Finally, regarding the lower bound, our proof adapts the \emph{round elimination} technique from \cite{patrascu2010cellprobe}, which was originally used to prove space lower bounds for partial sum queries. This technique applies when there exists some correlation between answers to different queries. In the case of partial sum queries, the sum of the first $i$ items and the sum of the first $i-1$ items are highly correlated. Later works generalized the round elimination technique to prove lower bounds for other problems such as range minimum queries \cite{liu2020lower,liu2021nearly}, which exploited the correlation between queries \emph{from a subset}, and queries on interval graphs \cite{chakraborty2023interval}, which exploited the correlation between \emph{pairs of queries}. In comparison, our lower bound exploits the correlation between \emph{positive queries}: When $\gap(S)$ is small, knowing that one key $x$ belongs to $S$ gives us some information about other keys in $S$.

\section{Preliminaries}

\label{sec:prelim}

We define $\text{bin}(x)$ as the binary representation of $x$, which is a $\lceil \log U\rceil$-bit string where the most significant bits come first. Here, $U$ is the universe that $x$ lives in, which will always be clear from context. In particular, we always have $0\le x\le U-1$.

We use $[U]$ to denote the set $\{0,\dots,U-1\}$.
    
\subsection{The Gap Entropy}

The standard data-aware space measure used in this paper is called the \defn{gap entropy} (also called the gap measure in \cite{gupta2007compressed}), which measures the cost of difference-encoding a set of keys. Let $S=\{x_1<x_2<\dots<x_{|S|}\}$ be a non-empty set of integers. The gap entropy of $S$ is defined as
\begin{align*}
    \gap(S)\ \coloneqq\ \sum_{i=2}^{|S|} \bigl(\bigl\lceil \log(x_i-x_{i-1}+1)\bigr\rceil+1\bigr).
\end{align*}

For notational simplicity, we sometimes write $\gap(x_1,\dots,x_k)$ instead of $\gap(\{x_1,\dots,x_k\})$.

Our definition of the gap entropy is different from those in prior works, which gives it additional properties that make our lives much easier later:
\begin{itemize}
    \item Our gap entropy is additive, in the sense that for integers $x_1<\dots<x_n$ and some index $1\le m\le n$,
    \begin{align*}
        \gap(x_1,\dots,x_n)=\gap(x_1,\dots,x_m)+\gap(x_m,\dots,x_n).
    \end{align*}
    \item For any integers $x<x'$, $\gap(x,x')$ is at least $2$. This guarantees that $\log\gap(x,x')$ is always nonzero.
\end{itemize}

Note that the gap entropy is not strictly speaking an achievable space bound. This is because, in addition to encoding the differences $\{x_i-x_{i-1}\}$, we have to also encode the lengths of these differences $\{\log(x_i-x_{i-1})\}$ and the first key $x_1$, whose costs are not counted in the gap entropy. After accounting for these costs, an actually achievable space bound for encoding a non-empty subset of $[U]$ (denoted as $S=\{x_1<x_2<\dots<x_{|S|}\}\subset[U]$) is 
\begin{align*}
    &\gap(S)+\lceil\log U\rceil+O\bk*{\sum_{i=2}^{|S|}\log\gap(x_{i-1},x_i)} \\
    \le{}&\gap(S)+\lceil\log U\rceil+O\bk*{|S|\log\gapbar(S)}.
\end{align*}
The last step is due to Jensen's inequality and additivity of the gap entropy, and we define $\gapbar(S)\coloneqq \gap(S)/(|S|-1)$ for notational simplicity (with the convention $\gapbar(S)=1$ when $|S|=1$).

\subsection{Dynamic Retrieval Data Structures}

Retrieval data structures are similar to key-value dictionaries, except that they are only responsible for returning the correct value for actually present keys, and do not have to store the key set. This allows them to be much more space-efficient than dictionaries. It has been known for a long time that dynamic retrieval data structures can be constructed using $n\log V+O(n\log\log U)$ bits of space and constant query time \cite{demaine2006dictionariis}. In this paper, we use the following construction of \emph{resizable} dynamic retrieval data structures by Kuszmaul et al.~\cite{kuszmaul2026resizable}, which allows the set size $n$ to change over time.

\begin{theorem}[\cite{kuszmaul2026resizable}]
    \label{thm:retrieval}
    Let $U$ and $V$ be parameters with $\log V\le O(\log U)$. In the word RAM model with word size $w=\Omega(\log U+\log V)$, there is a fully dynamic retrieval data structure maintaining a set $S\subseteq [U]$ for which $\poly\log U\le |S|\le U/2$, where each key is associated with a value in $[V]$. In $O(1)$ time, the data structure supports insertion, deletion and value queries. In particular, when answering a value query on key $x\in [U]$, the data structure returns the value associated with $x$ if $x\in S$, and an arbitrary value otherwise. At any time, the data structure uses
    \begin{align*}
        |S|\cdot \log V+O(|S|\log\log(U/|S|))
    \end{align*}
    bits of space. The data structure assumes access to lookup tables and hash functions of total size at most $U^{\delta}$ bits, where $\delta>0$ is an arbitrary parameter. At any given moment, the space and time bounds hold with high probability in $|S|$.\footnote{Upon close inspection, when given an arbitrarily long sequence of operations, the data structure of \cref{thm:retrieval} is \emph{always} guaranteed to return the correct answer for each query. Furthermore, the worst-case time and space bounds are $O(|S|)$ and $O(|S|\log U)$, respectively.}
\end{theorem}

\begin{remark}    
    \label{rem:retrieval-variable-length}
    The same paper~\cite{kuszmaul2026resizable} also gives an extension in which the value associated with each key $x\in S$ may have an arbitrary length $v_x\in[O(\log U)]$ specified at insertion time, using $\sum_{x\in S}v_x+O(|S|\log\log(U/|S|))$ bits of space and the same $O(1)$-time operations. This extension is later used to implement the variable-length word RAM model.
\end{remark}

\subsection{The Variable-Length Word RAM Model}

While our main result is stated in the word RAM model, we will mostly work in a slightly different model that we call the \defn{variable-length word RAM model}, in which each word stores variable-length information. Note that this model does not add new abilities to the word RAM model; it is only a level of indirection that sweeps the use of retrieval data structures under the rug.

\begin{definition}[Variable-Length Word RAM Model]
    \label{def:vl_word_RAM}
    In the variable-length word RAM model with \defn{word size} $w$, the memory is a collection of $2^w$ words. At any point in time, each word is either \defn{idle} or \defn{active}. When a word is active, it can store as its \defn{content} a bit string of length at most $w$. In unit time, an algorithm in the variable-length word RAM model can perform the following operations:
    \begin{enumerate}
        \item Activate a word and set its content to be empty.
        \item Deactivate a word.
        \item Write to/Read the content of an active word. When we read an active word, the memory returns the length of its content as well as the content itself.
        \item Perform arithmetic or bitwise operations on two $w$-bit integers.
    \end{enumerate}
    Note that the memory is not responsible for remembering which words are active. If the queried word of some type 1 (resp.~2, 3) operation is currently active (resp.~idle), the behaviour of the memory is undefined. We say that a data structure in the variable-length word RAM model is \defn{legitimate} if it never causes undefined behavior.

    The \defn{space usage} of the memory in the variable-length word RAM model with word size $w$ is defined to be
    \begin{align*}
        \text{(total length of word contents)}+C_{retr}\cdot \log w\cdot \text{(\# of active words)},
    \end{align*}
    where $C_{retr}$ is a sufficiently large constant. We say that a data structure has \defn{address limit} $M$, if it only ever accesses the first $M$ words, regardless of its input.

\end{definition}

The variable-length word RAM model possesses many nice properties: It is easy to implement the variable-length word RAM model in the word RAM model (using a dynamic retrieval data structure), and vice versa. It is also easy to combine multiple data structures in the variable-length word RAM model (just assign disjoint sets of addresses to each data structure).

Another implicit benefit of this model is that we no longer need to store pointers to nodes when implementing binary search trees (and other similar data structures). This is because the address space is sufficiently large, allowing us to index each node by its value.

We defer the formal discussions regarding these properties to \cref{app:vl_word_RAM}.

\section{Warm-Up: A Dictionary Using \texorpdfstring{$O(\varepsilon^{-1})$}{1/eps} Time}

\label{sec:warm-up}

In this section, we prove a weakened version of our main result, stated below as a data structure in the variable-length word RAM model.

\begin{theorem}
    \label{thm:main-warmup}
    Let $U$ be an integer. Let $\Omega(\log\log U/\log U)\le \varepsilon<1/4$ be a parameter of our choice. There exists a legitimate data structure in the variable-length word RAM model with word size $w=100\log U$ that maintains a key set $S\subseteq [U]$. In $O(\varepsilon^{-1})$ expected amortized time, the data structure supports the following operations:
    \begin{itemize}
        \item \textbf{Insertion}: Insert a key $x$ to $S$.
        \item \textbf{Deletion}: Delete a key $x$ from $S$.
        \item \textbf{Membership}: Given $x\in [U]$, output whether $x\in S$.
    \end{itemize} The data structure uses
    \begin{align*}
        \gap(S)+O(\varepsilon \cdot\gap(S))+O(|S|\log \log U)+O\bk*{U^\delta}
    \end{align*}
    bits of space, where $\delta>0$ is an arbitrary constant parameter.
\end{theorem}

Note that, in comparison to Theorem \ref{thm:main}, the guarantees of Theorem \ref{thm:main-warmup} are weaker in two ways: (1) it offers $O(\epsilon^{-1})$-time operations, instead of $O(\log \epsilon^{-1}/\log\log \epsilon^{-1})$; and (2) it includes an additive $O(|S| \log \log U)$ term in the space bound, instead of the $O(|S| \log \gapbar(S))$ term from Theorem \ref{thm:main}. As we will see in later sections, each of these issues will require significant additional technical ideas to resolve. Nonetheless, Theorem \ref{thm:main-warmup} serves as a natural warmup for the paper, allowing us to present several of the basic building blocks that will be used throughout the paper in their simplest forms.

\subsection{Reducing to Chunks of Size \texorpdfstring{$\poly\log U$}{poly log U}}
\label{sec:outer-reduction}

To present the proof of \cref{thm:main-warmup}, we begin by assuming black-box access to the following lemma, which handles the case where the set $S$ has size $|S| \le \poly\log U$. (We will then prove Lemma \ref{lem:polylog-dict-warmup} in the next subsection.) This allows us to focus on the basic task of reducing from an arbitrary set $S$ to the case where $|S|$ is small.

\begin{lemma}
    \label{lem:polylog-dict-warmup}
    Let $U,\varepsilon, S$ be as defined in \cref{thm:main-warmup}. If it is guaranteed that $\log U\le |S|=\poly\log U$, then there exists a legitimate data structure in the variable-length word RAM model with word size $w=100\log U$ that maintains $S$. In $O(\varepsilon^{-1})$ expected amortized time, the data structure supports the following operations:
    \begin{itemize}
        \item \textbf{Insertion}: Insert a key $x$ to $S$.
        \item \textbf{Deletion}: Delete a key $x$ from $S$.
        \item \textbf{Predecessor/Successor}: Given $x\in [U]$, output the predecessor or successor of $x$ in $S$ (or $\bot$ if it does not exist).
    \end{itemize}
    The data structure also supports an operation called $\textbf{deallocate}$, that deactivates every active word in $O(|S|)$ time. The data structure has address limit $U^{90}$. It uses
    \begin{align*}
        \gap(S)+O(\varepsilon \cdot\gap(S))+O\bk*{|S|\log \gapbar(S)}
    \end{align*}
    bits of space. The data structure assumes access to lookup tables and hash functions of total size at most $U^{\delta}$ bits, where $\delta>0$ is an arbitrary parameter.
\end{lemma}

We remark that \cref{lem:polylog-dict-warmup} implicitly supports membership queries: $x\in S$ if and only if the predecessor of $x+1$ in $S$ is $x$. Also note that the last term $O(|S|\log \gapbar(S))$ of the space cost is better than the $O(|S|\log\log U)$ term of \cref{thm:main-warmup}. This improvement is not necessary for the warm up, but will be useful to us later on.
    
At a high level, we prove \cref{thm:main-warmup} by partitioning the key set $S$ into subsets of $\poly\log U$ keys, and building a small dictionary for each subset using \cref{lem:polylog-dict-warmup}. We then show how to build a small auxiliary data structure that, given a key, locates the subset that it belongs to. And, finally, we show how to implement insertions/deletions while still incurring only $O(\epsilon^{-1})$ expected amortized time.

\subsubsection{Partitioning into Chunks}  Let $C_{\text{chunk}}$ be a sufficiently large constant. We partition the universe $[U]$ into disjoint intervals called \defn{chunks}, such that each chunk contains $\Theta(\log^{C_{\text{chunk}}}U)$ keys\footnote{Note that we can assume $|S|\ge \poly\log U$ for any $\poly\log U$ by adding dummy keys, which only increase the space cost by $U^{o(1)}$ and is acceptable.}. This size invariant is maintained by dynamically splitting and merging chunks as updates occur. Whenever a chunk grows too large, it is split into two smaller chunks; whenever it becomes too small, it is merged with a neighboring chunk. By scheduling these splits and merges at the appropriate times, one can straightforwardly ensure the following property.
\begin{claim}
    \label{clm:dynamic_chunks}
    An amortized number of $O(1/\log^{C_{\text{chunk}}}U)$ chunks are created and destroyed per update. Each chunk creation/destruction involves moving $O(\log^{C_{\text{chunk}}}U)$ keys from one chunk to another.
\end{claim}

The smallest values of the chunks are called \defn{pivots}. These pivots play a crucial role in our data structure.

\subsubsection{A Distributor That Maps Keys to Chunks}

We now describe a small auxiliary data structure that, given a key, returns the chunk it belongs to. A similar (but static) data structure appears in the construction of monotone minimal perfect hash functions (see e.g., \cite{belazzougui2009monotone,boldi2015minimal}). Following the naming of \cite{boldi2015minimal}, we call this data structure a distributor. Note that, whereas the distributor in this section takes roughly $O(|S| \log \log U)$ bits of space, in Section \ref{sec:better-distributor} we will see a more sophisticated construction that brings the space down to roughly $O(|S| \log \gapbar(S))$.

\begin{theorem}[Distributor]
    \label{thm:distributor}
    Let $U$ be an integer. There is a data structure in the word RAM model with word size $w=\Theta(\log U)$ that maintains two sets $P,S\subset [U]$, where $0\in P$ always holds. For each key $x\in [U]$, its \defn{pivot} is defined as the largest key in $P$ smaller than or equal to $x$ (denoted as $p_x$). Our data structure supports the following operations in expected amortized time:
    \begin{itemize}
        \item \textbf{Insertion/Deletion on $P$}: Insert a key to or delete a key from $P$ in $O(\log U)$ time.

        It is guaranteed that this operation would \emph{not} change the pivot of any key in $S$, and that for each insertion (resp.~deletion), the given key did not exist (resp.~did exist) in $P$.
        \item \textbf{Insertion on $S$ with advice}: Given a key $x\notin S$ \emph{and its pivot} $p_x\in P$, insert $x$ to $S$ in $O(1)$ time.
        \item \textbf{Deletion on $S$}: Given a key $x\in S$, delete $x$ from $S$ in $O(1)$ time.
        \item \textbf{Pivot}: Given a key $x\in S$, return its pivot $p_x$ in $O(1)$ time. If the given key is not in $S$, the distributor may return anything.
    \end{itemize}
    The data structure uses
    \begin{align*}
        O(|P|\cdot \log^2U+|S|\cdot \log\log U)
    \end{align*}
    bits of space with high probability in $|S|$ when $|S|\ge (\log U)^c$ for a sufficiently large constant $c$. The data structure assumes access to lookup tables and hash functions of total size at most $U^{\delta}$ bits, where $\delta>0$ is an arbitrary parameter.
\end{theorem}
When using \cref{thm:distributor} in the dictionary problem, we will take $S$ to be the set of all keys and $P$ to be the set of pivots.

Note that while \cref{thm:distributor} is stated in the word RAM model, it can also be implemented in the variable-length word RAM model while keeping its (asymptotic) time and space guarantees.

Also note that, when insertions are performed, the distributor defined above requires \emph{advice} (the exact pivot $p_x$ of $x$) in order for a new key $x$ to be inserted. The question of how to obtain this advice will be addressed in the next subsection.

\paragraph{Intuition.} Before presenting our construction, we first analyze two na\"ive constructions of distributors to gain some intuition.
\begin{itemize}
    \item Perhaps the simplest idea is to build a predecessor data structure on $P$, which can answer pivot queries for any key in $[U]$. Such a data structure can be constructed using $O(\log U)$ bits of space per key in $P$. However, as shown by \cite{beame1999predecessor}, dynamic predecessor data structures require at least $\Omega(\log\log U/\log\log\log U)$ time to answer queries (when $|P|=U^{\Omega(1)}$), no matter the space usage.
    
    Thus, in order to answer queries in $O(1)$ time, we crucially need to use the fact that distributors are only required to correctly answer pivot queries for keys in $S$.
    
    \item In order to achieve $O(1)$ query time, we can use a retrieval data structure to directly store the mapping $x\mapsto p_x$ for each $x\in S$. However, it necessarily costs $O(\log U)$ bits per key to store the $p_x$'s, which is unacceptable.
    
    To improve this space cost, we note that, even though we cannot afford $\Theta(\log U)$ bits per key $x \in S$, we \emph{can} afford $\Theta(\log^2 U)$ bits of space per \emph{pivot} $p_x$. This fact will be crucial to our construction.
\end{itemize}

\paragraph{The distributor construction.} Instead of directly storing the pivots $p_x$ for each $x\in S$ in the retrieval data structure, we store $\log\log U$ bits of information per $x$ that help us find $p_x$.

For each $x\in S$, let $v_x$ denote the length of the longest common prefix of $\text{bin}(x)$ and $\text{bin}(p_x)$ (recall from the preliminaries that $\text{bin}(x)$ is the binary representation of $x$, in which the most significant bits come first).

Our construction consists of two parts: 

\begin{itemize}
    \item We use a dynamic retrieval data structure to store the mapping $x\mapsto v_x$, which uses $O(|S|\log\log U)$ bits of space and answers queries in $O(1)$ time. This retrieval data structure satisfies the requirement of \cref{thm:retrieval} because the number of keys in $S$ is always at least $\poly\log U$.
    
    After querying this retrieval data structure and knowing the value of $v_x$, we can recover the following information about $p_x$:
    \begin{itemize}
        \item If $v_x=\lceil\log U\rceil$, then $p_x=x$.
        \item Otherwise, by definition of $v_x$, we know that $\text{bin}(p_x)_{1\dots v_x}=\text{bin}(x)_{1\dots v_x}$ and $\text{bin}(p_x)_{v_x+1}\ne \text{bin}(x)_{v_x+1}$. Since $x\ge p_x$, we must have $\text{bin}(p_x)_{v_x+1}=0$ and $\text{bin}(x)_{v_x+1}=1$. Moreover, since $p_x$ is the largest key in $P$ that is smaller than $x$, it must also be the largest key in $P$ whose first $v_x+1$ bits are equal to $\text{bin}(x)_{1\dots v_x}\circ \text{``0''}$.
    \end{itemize}
    
    \item Based on the aforementioned information about $p_x$, we build a key-value dictionary (see \cite{bender2022optimal}) that is used to find $p_x$. The keys in this dictionary correspond to non-empty bit strings of length at most $\lceil\log U\rceil$, where a key $s\in \{0,1\}^{\le \lceil\log U\rceil}$ is present in the dictionary if and only if $s$ is a prefix of $\text{bin}(p)$ for some $p\in P$, and the value associated with key $s$ is the largest such $p$. As a convention, if key $s$ is not present in the dictionary, we say its value is $-1$.

    The number of keys in this dictionary is at most $O(|P|\cdot \log U)$ (because each pivot has $O(\log U)$ prefixes), so it uses $O(|P|\log^2U)$ bits overall. Moreover, after querying the retrieval data structure on $x$, we know that $p_x$ is the largest pivot that has $\text{bin}(x)_{1\dots v_x}\circ \text{``0''}$ as a prefix, and can therefore find $p_x$ by directly querying the dictionary at $\text{bin}(x)_{1\dots v_x}\circ \text{``0''}$, using $O(1)$ time.
\end{itemize}

\paragraph{The analysis.} The distributor uses $O(|P|\log^2U+|S|\log\log U)$ bits of space overall. Pivot queries are answered by first querying the mapping $x\mapsto v_x$, then querying the dictionary on $\text{bin}(x)_{1\dots v_x}\circ \text{``0''}$, using $O(1)$ time. 

Insertions and deletions to $S$ are performed by directly modifying the mapping $x\mapsto v_x$ (this is easy since the algorithm is given $p_x$ as advice during insertions, and knows the value of $p_x$ during deletions by definition of the distributor). Insertions and deletions to $P$ are done as follows. Let $p\in P$ be the key that is being inserted or deleted. Observe that by definition of the key-value dictionary, only the values of the keys $\text{bin}(p)_{1\dots i}$ for $1\le i\le \lceil\log U\rceil$ may change. Also note that, for any $s\in \{0,1\}^{\le\lceil\log U\rceil}$, the value of key $s$ is always equal to the maximum of the values of $s\circ \text{``0''}$ and $s\circ \text{``1''}$. We can thus maintain the dictionary as follows:
    \begin{itemize}
        \item First, we set the value of $\text{bin}(p)$ to be $p$ if we are creating $p$, and remove this key from the dictionary otherwise.
        \item Next, we enumerate $i=\lceil\log U\rceil-1\dots 1$ (note that we enumerate the indices in decreasing order). For each $i$, we set the value of $\text{bin}(p)_{i}$ in the dictionary to be its correct value.
        
        When enumerating $i$, we first query the dictionary for the values of $\text{bin}(p)_{1\dots  i}\circ \text{``0''}$ and $\text{bin}(p)_{1\dots  i}\circ \text{``1''}$ (these values are already guaranteed to be correct), and set the value of $\text{bin}(p)_{1\dots i}$ to be their maximum. If $\text{bin}(p)_{1\dots i}$ has value $-1$, then we remove it from the dictionary.
    \end{itemize}
This process costs $O(\log U)$ time per key insertion/deletion to $P$. We thus conclude the construction of the distributor.

\subsubsection{Exploiting Sort Efficiency for Insertions}

So far, our dictionary consists of the small dictionaries of \cref{lem:polylog-dict-warmup} and a distributor. A query on key $x$ first probes the distributor to locate the chunk of $x$, then accesses the corresponding small dictionary.

Next, we address the issue of performing insertions efficiently. When the dictionary inserts a key $x$ into $S$, we need to provide \cref{thm:distributor} with the pivot $p_x \in P$ for $x$, which intuitively would seem to require a predecessor query on the set of pivots. However, predecessor queries take $\Omega(\log \log |S|)$ time, which would be unacceptable.

To resolve this issue, we observe that predecessor queries can be very efficient if they are performed in batches. That is, answering $m$ predecessor queries simultaneously on a set of size $k$ only takes $O(m+k)$ time. This is done by first sorting the query keys using radix sort (taking $O(m)$ time when $m=U^{\Omega(1)}$), then answering the queries in increasing order of key, whose time cost is the same as merging the set of query keys with the pivots, i.e., $O(m+k)$.

Inspired by this, we use a dynamic dictionary to buffer the insertions. We only insert keys to the distributor and the small dictionaries when the number of buffered insertions reaches $\max(U^{\delta/2},|S|/\log^2 U)$ ($\delta$ is the constant parameter in the statement of \cref{thm:main-warmup}). The buffer dictionary uses $o(|S|)+O(U^{\delta})$ bits of space which is negligible to \cref{thm:main-warmup}, and performing predecessor queries on $\max(U^{\delta/2},|S|/\log^2 U)$ keys simultaneously takes
\begin{align*}
    O(\max(U^{\delta/2},|S|/\log^2 U)+|S|/\log^{C_{\text{chunk}}}U)
\end{align*}time (the latter term accounts for the number of pivots), which is only $O(1)$ time per inserted key. This resolves the issue of insertions.

We only need to buffer the insertions and not the deletions, because the distributor does not require advice for deletions. When we need to delete a key, we simply have to check whether it exists in the buffer dictionary or in one of the small dictionaries of a chunk, and remove it from there.

\subsubsection{Putting Everything Together} 
We now formally describe the entire construction. Our dictionary consists of the following parts: A buffer dictionary, a distributor, a dynamic dictionary that maintains the set of pivots, and a set of small dictionaries, one for each chunk, constructed using \cref{lem:polylog-dict-warmup}. These data structures are combined using \cref{clm:virtual_address_combine}. More specifically, we instantiate \cref{clm:virtual_address_combine} with $U+3$ logical data structures where the first $U$ logical data structures are used to store the small dictionaries of the chunks, and the last three logical data structures are used to store the buffer dictionary, the distributor and the pivot dictionary. The small dictionary of a chunk whose pivot is $p$ is stored in the $p$-th logical data structure. This instantiation of \cref{clm:virtual_address_combine} is valid because each small dictionary has address limit $U^{90}$, and the other three data structures all have address limit $O(U)$ by \cref{clm:word_RAM_to_vl}.

Also note that the space guarantee of \cref{thm:distributor} only holds with high probability in $|S|$, so we need to perform a rebuild whenever the distributor takes too much space. A rebuild can be performed in $O(|S|\cdot \varepsilon^{-1})$ time and is negligible because $\varepsilon^{-1}=O(\log U/\log\log U)=o(|S|)$ and rebuilds only happen with probability $1/\poly |S|$ at each time step.

\paragraph{Performing operations.} To answer a membership query on key $x$, we first access the buffer dictionary to see if it contains $x$. If not, we then query the distributor to obtain the pivot $p_x$ (if $x$ does exist in $S$), and query the small dictionary corresponding to $p_x$ to see if $x$ exists. Note that if $x$ does not exist in $S$, the distributor may return anything, but this is acceptable: If the returned value is not a pivot, we can detect it by checking the pivot dictionary. If the returned value is a pivot but does not correspond to the chunk containing $x$, the corresponding small dictionary will report that $x$ does not exist.

To perform an insertion, we simply add the new key to the buffer dictionary. Then, if the buffer dictionary is too large (i.e., at least $\max(U^{\delta/2},|S|/\log^2U)$), we perform insertions in batches, and clear the buffer dictionary. We first sort the keys in the buffer dictionary, then enumerate the pivots in increasing order (by accessing the pivot dictionary) and check the predecessor of each buffered key. After knowing the predecessors, we insert the keys into the distributor and the small dictionaries, and clear the buffer dictionary.

Deletions are performed more straightforwardly. When deleting key $x$, if $x$ is in the buffer dictionary, then we can easily delete it. Otherwise, we find the correct pivot $p_x$ by accessing the distributor, then delete $x$ from both the distributor and the small dictionary corresponding to $p_x$.

We also need to maintain the chunk structure, i.e., create and destroy chunks. This is straightforward given the operations supported by the distributor. Note that when creating or destroying a chunk, the distributor requires that the chunk currently contains no keys (see the operation ``insertion/deletion to $P$'' of \cref{thm:distributor}). This is guaranteed by our algorithm at the cost of moving $O(\log^{C_{\text{chunk}}}U)$ keys from one chunk to another for each chunk creation/deletion.

\paragraph{Time efficiency.} When answering a query on key $x$, we need to access the buffer dictionary, then possibly the distributor and one small dictionary. These operations cost $O(\varepsilon^{-1})$ expected amortized time in total. After this, if the number of buffered insertions grows too large, we insert them into the distributor and the small dictionaries, and clear the buffer dictionary. This also costs $O(\varepsilon^{-1})$ expected amortized time.

As for the cost of creating and destroying chunks: By \cref{clm:dynamic_chunks}, the amortized number of chunk creations/destructions is $O(1/\log^{C_{\text{chunk}}}U)$, and each such operation involves updating the set of pivots in the distributor (costing $O(\log U)$ time) and moving $O(\log^{C_{\text{chunk}}}U)$ keys between chunks (each costing $O(\varepsilon^{-1})$ time to access the small dictionaries and $O(1)$ time to access the distributor). 
When destroying a chunk, we also need to run the deallocate operation on a small dictionary, which costs $O(\log^{C_{\text{chunk}}}U)$ time. All the above time costs are expected and amortized. Therefore, each chunk creation/destruction costs $O(\log^{C_{\text{chunk}}}U\cdot \varepsilon^{-1})$  expected amortized time.

Overall, the data structure uses $O(\varepsilon^{-1})$ expected amortized time per operation.

\paragraph{Space efficiency.} The buffer dictionary and the pivot dictionary both use $o(|S|)+O(U^{\delta})$ bits of space. The distributor uses
\begin{align*}
    O(|S|\cdot \log U/\log^{C_{\text{chunk}}}U+|S|\cdot \log\log U)={}O(|S|\cdot \log\log U)
\end{align*}bits of space, and each small dictionary uses
\begin{align*}
    \gap(S_i)+O(\varepsilon\cdot \gap(S_i))+O\bk*{|S_i|\cdot \log\gapbar(S_i)}
\end{align*}bits of space (letting $S_i$ denote the set of keys in the $i$-th chunk). In total, the space cost is
\begin{align*}
    &o(|S|)+O(U^{\delta})+O(|S|\cdot \log\log U) \\
    &+\sum_{i=1}^{O(|S|/\log^{C_{\text{chunk}}}U)}\bk*{\gap(S_i)+O(\varepsilon \gap(S_i))+O\bk*{|S_i|\cdot \log\gapbar(S_i)}} \\
    ={}&O(|S|\cdot \log\log U)+\sum_{i=1}^{O(|S|/\log^{C_{\text{chunk}}}U)}\bk*{\gap(S_i)+O(\varepsilon \gap(S_i))+O(|S_i|\cdot\log \log U)}+O(U^{\delta}) \\
    ={}&\gap(S)+O(\varepsilon \cdot\gap(S))+O(|S|\cdot\log\log U)+O(U^{\delta}).
\end{align*}

This concludes the proof of \cref{thm:main-warmup} (assuming \cref{lem:polylog-dict-warmup}).

\subsection{Solving a Chunk in \texorpdfstring{$O(\varepsilon^{-1})$}{O(eps \^-1)} Time}

\label{sec:chunk}

Finally, to complete the warmup construction, we prove \cref{lem:polylog-dict-warmup}. Note that, whereas the construction in this section incurs $O(\epsilon^{-1})$ time per operation, this is not optimal. With significant additional data structural techniques, we will be able to bring this down to $O(\log \epsilon^{-1})$ (in Section \ref{sec:binary-tree}) and then to $O(\log \epsilon^{-1} / \log \log \epsilon^{-1})$ (in Section \ref{sec:trie-plus-btree}), which we will show to be optimal (even just for membership queries) in Section \ref{sec:lb}.

\subsubsection{A Special Parameter Regime}

We first solve the special (and arguably most natural) regime of \cref{lem:polylog-dict-warmup} where $\varepsilon=1/|S|$ and $\gap(S)=\Theta(|S|\log U)$, then generalize to other regimes. Note that, in later sections (\cref{sec:binary-tree,sec:trie-plus-btree}) when we improve the dependency on $\epsilon^{-1}$, it will suffice to focus on just this parameter regime, as our reductions for other parameter regimes will carry over to the new solutions.

\begin{lemma}
    \label{lem:difference-encode-array}
    Let $U$ be an integer. There exists a legitimate data structure in the variable-length word RAM model with word size $w=100\log U$ that maintains a key set $S\subseteq [U]$, where $|S|\ge \log U$ always holds. The data structure supports the following operations in $O(|S|)$ time:
    \begin{itemize}
        \item \textbf{Insertion}: Insert a key $x$ to $S$.
        \item \textbf{Deletion}: Delete a key $x$ from $S$.
        \item \textbf{Predecessor/Successor}: Given $x\in [U]$, output the predecessor or successor of $x$ in $S$ (or $\bot$ if it does not exist).
    \end{itemize}
    The data structure also supports an operation called $\textbf{deallocate}$, that deactivates every active word in $O(|S|)$ time. The data structure has address limit $U$, and uses
    \begin{align*}
        \gap(S)+O(|S|\log\log U)
    \end{align*}
    bits of space.
\end{lemma}

The proof of \cref{lem:difference-encode-array} is extremely simple:

\begin{proof}
    Since we are allowed $O(|S|)$ time per operation, we simply have to use $|S|$ words to store the difference encoding of $S$. This uses $\gap(S)+\lceil \log U\rceil+O(|S|\log\log U)$ bits of space, where the last term also accounts for the cost of using $|S|$ words in the variable-length word RAM model. When performing any operation, we will decode $S$, perform the operation (possibly updating $S$ in the process), then write the new encoding of $S$ to the memory.
\end{proof}

In the following, we show that \cref{lem:difference-encode-array} implies \cref{lem:polylog-dict-warmup}.

\subsubsection{Step 1: Generalizing to Any \texorpdfstring{$\varepsilon$}{eps}}

To better describe our data structure, it is convenient to define the following intermediate goal, which solves \cref{lem:polylog-dict-warmup} for any $\varepsilon$ provided that $\gap(S)=\Theta(|S|\log U)$.

\begin{lemma}
    \label{lem:array-intermediate}
    Let $0<\varepsilon<1/4$ be a parameter of our choice. There exists a legitimate data structure having address limit \textcolor{blue}{$U^{10}$} that supports the same operations as \cref{lem:difference-encode-array} while using \textcolor{blue}{$O(\varepsilon^{-1}+\log_{w}|S|)$ amortized} time and
    \begin{align*}
        \gap(S)+\textcolor{blue}{O(\varepsilon\cdot |S|\log U)}+O(|S|\log\log U)
    \end{align*}
    bits of space.
\end{lemma}

The differences between \cref{lem:array-intermediate} and \cref{lem:difference-encode-array} are marked in blue.

Compared to \cref{lem:difference-encode-array}, \cref{lem:array-intermediate} allows us to trade space for time. This is achieved by partitioning the set into groups, which are subsets of size $O(\varepsilon^{-1})$, and building a dictionary for each group using \cref{lem:difference-encode-array} (so that it only costs $O(\varepsilon^{-1})$ time to access each group). We then use a fusion tree to help locate the relevant group of each query, which uses $O(\log U)$ bits of extra space per group.

\paragraph{Partitioning into groups.} We partition the universe $[U]$ into disjoint intervals called \defn{groups} (not to be confused with the \emph{chunks} from earlier), where the number of keys in each group is $\Theta(\varepsilon^{-1})$. Here, we assume that $\varepsilon\ge 1/|S|$ since otherwise we simply have to have one group and apply \cref{lem:difference-encode-array} on it. This partitioning is maintained in the same way as the chunks in \cref{sec:outer-reduction}, so it also satisfies \cref{clm:dynamic_chunks} (but with $O(\log^{C_{\text{chunk}}}U)$ replaced by $O(\varepsilon^{-1})$). Similarly to \cref{sec:outer-reduction}, with a slight abuse of notation, we also define the \defn{pivots} to be the smallest values of the groups.

Note that the number of groups is $O(\varepsilon\cdot |S|)$, so we are allowed to use $O(\log U)$ bits of extra space per group in the data structure.

\paragraph{The data structure.} Our data structure consists of two parts:
\begin{itemize}

    \item We build a fusion tree (e.g., \cite{patrascu2014dynamic}) on the pivots. This fusion tree uses $O(\log U)$ bits of space per group (which is $O(\varepsilon\cdot |S|\log U)$ bits of space overall), and can answer predecessor/successor queries on the set of pivots in $O(\log_w |S|)$ time.
    
    \item We build a dictionary for each group using \cref{lem:difference-encode-array}, so that each dictionary can be accessed in $O(\varepsilon^{-1})$ time. These dictionaries use
    \begin{align*}
        \gap(S)+O(\varepsilon\cdot |S|\log U)+O(|S|\log\log U)
    \end{align*}
    bits of space in total.
\end{itemize}

We combine the data structures using \cref{clm:virtual_address_combine}. Since each of the $U$ dictionaries has address limit $U$, the address limit of the overall data structure is at most $U^{10}$. The query algorithms are standard (first query the fusion tree, then go to the specific dictionary) and omitted.

\paragraph{Time efficiency.} Since the partitioning of groups satisfies a property similar to \cref{clm:dynamic_chunks}, the dictionaries and the fusion tree can be maintained in $O(\varepsilon^{-1})$ and $O(\log_w|S|)$ amortized time per update to $S$, respectively. We omit the details. In total, the data structure answers queries in $O(\varepsilon^{-1}+\log_w|S|)$ amortized time.

\paragraph{Space Efficiency.} In total, the data structure uses
\begin{align*}
    \gap(S)+O(\varepsilon\cdot |S|\log U)+O(|S|\log\log U)
\end{align*}
bits of space, as desired. 

\subsubsection{Step 2: Generalizing to Any \texorpdfstring{$\gap(S)$}{gap(S)}}

Finally, we prove \cref{lem:polylog-dict-warmup} using \cref{lem:array-intermediate}. Let us compare the differences between \cref{lem:array-intermediate} and \cref{lem:polylog-dict-warmup}: If we directly apply the data structure of \cref{lem:array-intermediate} to solve \cref{lem:polylog-dict-warmup}, its time cost would be acceptable to \cref{lem:polylog-dict-warmup} since $\log_w|S|=O(1)$ when $|S|=\poly\log U$. However, this data structure would use too much space: \cref{lem:polylog-dict-warmup} only allows
\begin{align*}
    \gap(S)+O(\varepsilon \cdot\gap(S))+O\bk*{|S|\log\gapbar(S)}
\end{align*}
bits of space, whereas \cref{lem:array-intermediate} uses
\begin{align*}
    \gap(S)+O(\varepsilon\cdot |S|\log U)+O(|S|\log\log U)
\end{align*}
bits.

Note that in the typical case where $\gap(S)=\Theta(|S|\log U)$, these two bounds are indeed equivalent to each other. In the following, we reduce the general case to the typical case by partitioning the set of keys into subsets called blocks, such that each block has $\Theta(\log U)$ total gap entropy. We then show that when viewing each block as a whole, the general case of \cref{lem:polylog-dict-warmup} can be solved using \cref{lem:array-intermediate}.

For simplicity, we assume that $0$ is always stored in the key set (this only increases the total space by $O(\log U)$, and is negligible since $|S|\ge \log U$).

\paragraph{Partitioning into blocks.} Denote the keys in $S$ as $x_1<\dots<x_{|S|}$. We partition the set $S$ into consecutive subsets called \defn{blocks} (not to be confused with either the chunks or groups from earlier), where each block has total gap entropy $\Theta(\log U)$. Formally, if a block contains keys $x_l,\dots,x_r$, then $\gap(x_l,\dots,x_{r+1})=\Theta(\log U)$. Let $t$ denote the current number of blocks and let $a_i(1\le i\le t)$ denote the index of the smallest key of the $i$-th block (for notational simplicity, let $a_{t+1}\defeq |S|+1$ and $x_{|S|+1}\defeq U$).

Note that, contrary to the groups and chunks, we are expected to create or destroy a constant number of blocks per update. This is acceptable because, unlike the chunks and groups, the information of each block is so small that it is possible to rebuild the data structure for an entire block in $O(1)$ time.

\paragraph{Difference encoding for a small number of keys.} We encode each block of $S$ using the following lemma.
    
\begin{lemma}
    \label{lem:diff-encode-larger-interval}
    Let $k$ be a positive integer and fix two integers $x_0<x_{k+1}$. For any set of keys $x_1<\dots<x_k\in (x_0,x_{k+1})$, there exists an encoding $\enc_{x_0,x_{k+1}}(\{x_1,\dots,x_k\})$ of $\{x_1,\dots,x_k\}$ using
    \begin{align*}
        |\enc_{x_0,x_{k+1}}(\{x_1,\dots,x_k\})|&\le\gap(x_0,\dots,x_{k+1})-\gap(x_0,x_{k+1})+O\bk*{\sum_{i=1}^{k+1}\log\gap (x_{i-1},x_i)} \\
        &\le O(\gap(x_0,\dots,x_{k+1}))
    \end{align*}
    bits. The encoding also encodes the number of keys $k$, which can be learned by reading a prefix of length $O(\log (k+1))$.
\end{lemma}

Compared to the na\"ive difference encoding which uses
\begin{align*}
    \gap(x_0,\dots,x_{k})+O\bk*{\sum_{i=1}^{k+1}\log\gap (x_{i-1},x_i)}
\end{align*}
bits of space, the improvement of \cref{lem:diff-encode-larger-interval} seems negligible. However, this improvement turns out to be crucial in many of our applications.

\begin{proof}
    The idea behind \cref{lem:diff-encode-larger-interval} is simple: Find the index $1\le i^*\le k+1$ that maximizes $\gap(x_{i^*-1},x_{i^*})$. Our encoding stores $k,i^*$ (using $O(\log (k+1))$ bits) and all differences except one $\{x_i-x_{i-1}\}_{i\ne i^*}$. Given these differences, we can recover all the keys as we know $x_0$ and $x_{k+1}$. It takes at most
    \begin{align*}
        \gap(x_0,\dots,x_{k+1})-\gap(x_{i^*-1},x_{i^*})+O\bk*{\sum_{i=1}^{k+1}\log\gap (x_{i-1},x_i)}
    \end{align*}
    bits to store the differences. Since $x_{i^*}-x_{i^*-1}\ge (x_{k+1}-x_0)/(k+1)$ by definition, we have that
    \begin{align*}
        \gap(x_{i^*-1},x_{i^*})\ge \gap(x_0,x_{k+1})-O(\log(k+1)).
    \end{align*}
    Finally, the lemma is proven by noting that
    \begin{align*}
        \log(k+1)\le k+1\le\bk*{\sum_{i=1}^{k+1}\log\gap (x_{i-1},x_i)}.&\qedhere
    \end{align*}
\end{proof}

In order to use \cref{lem:diff-encode-larger-interval} in our data structure, we need to also support efficient queries on the encodings. This is achieved by the following lemma, where we assume access to a small lookup table.

\begin{lemma}
    \label{lem:query-encodings}
    Let $\delta>0$ be a constant. An encoding is \defn{short} if it uses no more than $(\delta/3)\cdot \log U$ bits, or if $k=1$. In the word RAM model with word size $w=\Theta(\log U)$, we support the following operations on the encodings (in the following, all the keys are in the range $\{-1,\dots,U\}$):
    \begin{itemize}
        \item \textbf{Partitioning}: Given keys $x_0<x_{k+1}$ and $\enc_{x_0,x_{k+1}}(\{x_1,\dots,x_k\})$, output an integer 
        $m$ that is at most $O(1+\gap(x_0,\dots,x_{k+1})/\log U)$, followed by $m+1$ indices $0=i_0<\dots<i_m=k+1$ and $m$ short encodings
        \begin{align*}
            \{\enc_{x_{i_{j-1}},x_{i_{j}}}(\{x_{i_{j-1}+1},\dots,x_{i_{j}-1}\})\}_{1\le j\le m}.
        \end{align*}
        This operation runs in $O(1+\gap(x_0,\dots,x_{k+1})/w)$ time.
        \item \textbf{Merging}: Given an integer $m$, $m+1$ keys $y_0<\dots<y_m$ and $m$ short encodings $\{\enc_{y_{j-1},y_j}(S_j)\}_{1\le j\le m}$, output the merged encoding
        \begin{align*}
            \enc_{y_0,y_{m}}(S_1\cup\{y_1\}\cup\dots\cup \{y_{m-1}\}\cup S_m)
        \end{align*}
        in $O(m)$ time.
        \item \textbf{Insertion/Deletion}: Given $y_0<y_1$, a short encoding $\enc_{y_0,y_1}(S)$ and a key $x\in (y_0,y_1)$, return the encoding $\enc_{y_0,y_1}(S)\oplus \{x\}$ in $O(1)$ time.
        \item \textbf{Special Deletion}: Given $y_0<y_1<y_2$ and short encodings $\enc_{y_0,y_1}(S_1)$, $\enc_{y_1,y_2}(S_2)$, return the encoding $\enc_{y_0,y_2}(S_1\cup S_2)$ in $O(1)$ time.
        \item \textbf{Predecessor/Successor}: Given $y_0<y_1$ and a short encoding $\enc_{y_0,y_1}(S)$, answer a predecessor or successor query on $S$ in $O(1)$ time.
    \end{itemize}
    The data structure assumes access to lookup tables and hash functions of total size at most $U^{\delta}$ bits, where $\delta>0$ is an arbitrary parameter.
\end{lemma}

The lookup table required by \cref{lem:query-encodings} can be constructed using standard techniques. We omit the details.

\paragraph{The data structure: Storing the set one block at a time.} To store $S$, we first apply \cref{lem:array-intermediate} to maintain a set of $t$ keys, where the $i$-th key is $x_{a_i}$ ($a_i(1\le i\le t)$ denotes the index of the smallest key of the $i$-th block). Next, we reserve a set of $O(U)$ words in the memory and use it to store the encodings $\enc_{x_{a_i},x_{a_{i+1}}}(\{x_{a_i+1},\dots,x_{a_{i+1}-1}\})$. Formally, the set of $O(U)$ words is partitioned into $U$ parts of $O(1)$ words each, and for every $1\le i\le t$, we store $\enc_{x_{a_i},x_{a_{i+1}}}(\{x_{a_i+1},\dots,x_{a_{i+1}-1}\})$ in the $x_{a_i}$-th part. The words that do not correspond to some $x_{a_i}$ are idle.

\paragraph{Correctness.} The information stored in the data structure is sufficient for recovering the set $S$. Indeed, we know the keys $\{x_{a_i}\}$ as they are stored explicitly. Given these keys, we can retrieve the encodings $\enc_{x_{a_i},x_{a_{i+1}}}(\{x_{a_i+1},\dots,x_{a_{i+1}-1}\})$ and decode them to obtain the entire key set.

\paragraph{Time efficiency.} For each operation, we first query the data structure of \cref{lem:array-intermediate} to reach the right block in $O(\varepsilon^{-1})$ time (say it's the $i$-th block), then apply \cref{lem:query-encodings} on the encoding of that block to partition it into short encodings. Formally, the partitioning gives us an integer $m$, keys $x_{a_{i}}=y_0<\dots<y_m=x_{a_{i+1}}$ and short encodings $\{\enc_{y_{j-1},y_j}(S_j)\}_{1\le j\le m}$. Since a block only has $\Theta(\log U)$ total gap entropy, we have $m=O(1)$.

\cref{lem:polylog-dict-warmup} supports insertion, deletion and predecessor/successor queries. These operations can all be answered by first locating the correct short encoding containing the key in question, then using the operations supported by \cref{lem:query-encodings}.

We need to handle some special cases for insertion and deletion: After inserting a key, the resulting encoding grows longer and may not be a short encoding anymore. In this case, we need to partition it again. Also, it may be the case that the key we are trying to delete does not lie in one of the short encodings, but rather equals some $y_i$. In this case, we can use the special deletion supported by \cref{lem:query-encodings}. If the key we are deleting is $x_{a_i}$, then we partition the $i-1$-th block (we must have $i>1$ in this case since we assumed that $0$ is always present in $S$), and run the special deletion operation to merge the last short encoding of block $i-1$ and the first short encoding of block $i$. Therefore, all the operations can be done in $O(1)$ time (aside from the time it takes to locate the correct block).

After performing the queries, we merge the short encodings back together into a constant number of blocks, in a way that maintains the total gap entropy of each block. We then update the dictionary of \cref{lem:array-intermediate}, which again takes $O(\varepsilon^{-1})$ time.

\paragraph{Space efficiency.} Our data structure for  Lemma \ref{lem:polylog-dict-warmup} consists of two parts: A dictionary given by \cref{lem:array-intermediate} that stores the keys $\{x_{a_i}\}$, and the encodings $\enc_{x_{a_i},x_{a_{i+1}}}(\{x_{a_i+1},\dots,x_{a_{i+1}-1}\})$.

The space cost of the dictionary is
\begin{align*}
    \gap(x_{a_1},\dots,x_{a_t})+O(\varepsilon\cdot t\log U)+O(t\log\log U)
\end{align*}
bits by \cref{lem:array-intermediate}.

As for the encodings, by \cref{lem:diff-encode-larger-interval}, the $i$-th encoding has length
\begin{align*}
    &|\enc_{x_{a_i},x_{a_{i+1}}}(\{x_{a_i+1},\dots,x_{a_{i+1}-1}\})|\\
    ={}&\bk*{\sum_{j=a_i+1}^{a_{i+1}}\gap(x_{j-1},x_j)}-\gap(x_{a_i},x_{a_{i+1}})+O\bk*{\sum_{j=a_i+1}^{a_{i+1}}\log \gap(x_{j-1},x_j)}.
\end{align*}
Each encoding occupies $O(1)$ words in the memory, so we need to pay $O(\log\log U)$ bits of extra space by definition of the variable-length word RAM model. The total space usage of the encodings is thus
\begin{align*}
    &\sum_{i=1}^{t}\bk*{\sum_{j=a_i+1}^{a_{i+1}}\gap(x_{j-1},x_j)}-\gap(x_{a_i},x_{a_{i+1}})+O\bk*{\sum_{j=a_i+1}^{a_{i+1}}\log \gap(x_{j-1},x_j)}+O(\log\log U) \\
    =&\gap(x_1,\dots,x_{|S|+1})-\gap(x_{a_1},\dots,x_{a_{t+1}})+O\bk*{\sum_{j=2}^{|S|+1}\log \gap(x_{j-1},x_j)}+O(t\log\log U).
\end{align*}

The overall space usage is
\begin{align}
    \label{equ:reduction-space-cost}
    &\gap(x_1,\dots,x_{|S|+1})-\gap(x_{a_1},\dots,x_{a_{t+1}})+O\bk*{\sum_{j=2}^{|S|+1}\log \gap(x_{j-1},x_j)}+O(t\log\log U) \nonumber\\
    &+\gap(x_{a_1},\dots,x_{a_t})+O(\varepsilon\cdot t\log U)+O(t\log\log U) \nonumber\\
    ={}&\gap(S)+O\bk*{\sum_{j=2}^{|S|}\log \gap(x_{j-1},x_j)}+O(\varepsilon\cdot t\log U)+O(t\log\log U).
\end{align}
Here we slightly changed the index ranges of the summations, which changes the value by $O(\log U)$ and is allowed since $|S|\ge \log U$.

To see that \eqref{equ:reduction-space-cost} satisfies the space requirement of \cref{lem:polylog-dict-warmup}, it remains to show that
\begin{align*}
    O\bk*{\sum_{j=2}^{|S|}\log \gap(x_{j-1},x_j)}+O(\varepsilon\cdot t\log U)+O(t\log\log U)=O(\varepsilon\cdot \gap(S))+O\bk*{|S|\log\gapbar(S)}.
\end{align*}
For the first term, we have by Jensen's inequality that
\begin{align*}
    \sum_{j=2}^{|S|}\log \gap(x_{j-1},x_j)&\le (|S|-1)\cdot\log\bk*{\sum_{j=2}^{|S|}\gap(x_{j-1},x_j)/(|S|-1)} \\
    &=O\bk*{|S|\log\gapbar(S)}.
\end{align*}
For the second term, we have
\begin{align*}
    O(\varepsilon\cdot t\log U)=O(\varepsilon\cdot \gap(S))
\end{align*}
By definition of $t$. For the third term, we have
\begin{align*}
    &O(t\log\log U) \\
    ={}&O(\gap(S)\log\log U/\log U) \\
    ={}&O\bk*{\gap(S)\log\gapbar(S)/\gapbar(S)} \\
    ={}&O\bk*{|S|\log\gapbar(S)}.
\end{align*}
To see that the second equality holds, note that $\gapbar(S)=O(\log U)$ because $\gapbar(S)=\gap(S)/(|S|-1)$, and that $\log x/x$ is decreasing when $x$ is large.

This concludes the proof of \cref{lem:polylog-dict-warmup}.

\section{Difference-Encoded Binary Search Trees}
\label{sec:binary-tree}

Our next major endeavor (in both this section and the next) is to improve the time cost of the warm-up construction from $O(\epsilon^{-1})$ to $O(\log \epsilon^{-1} / \log \log \epsilon^{-1})$. To do this, we first revisit an even more basic problem: how to encode a binary search tree using space close to $\gap(S)$. This leads us to the following lemma, which we believe to be of independent interest.

\begin{lemma}[Difference-Encoded Binary Search Tree]
    \label{lem:binary-tree}
    Let $U$ be an integer, and let $S\subset [U]$. There exists a legitimate data structure in the variable-length word RAM model with word size $w=100\log U$ that maintains the set $S$. In $O(\log |S|)$ worst-case time, the data structure supports the following operations:
    \begin{itemize}
        \item \textbf{Insertion}: Insert a key $x$ to $S$.
        \item \textbf{Deletion}: Delete a key $x$ from $S$.
        \item \textbf{Predecessor/Successor}: Given $x\in [U]$, output the predecessor or successor of $x$ in $S$ (or $\bot$ if it does not exist).
    \end{itemize}
    The data structure also supports an operation called $\textbf{deallocate}$, that deactivates every active word in $O(|S|)$ time. The data structure has address limit $U$. It uses
    \begin{align*}
        \gap(S)+\lceil \log U\rceil+O(|S|(\log\log U+\log|S|))
    \end{align*}bits of space in the worst case.
\end{lemma}

One straightforward idea for difference-encoding a binary search tree is to encode the value of each node relative to its parent. This encoding scheme is time-efficient because we always access the parent of a node before accessing the node itself. However, it is very space-inefficient when measured by $\gap(S)$: Imagine a degenerate tree (i.e., a tree in the form of a path) that stores the keys $S=\{1,\dots,n\}$, such that the root has key $1$, its only child has key $n$, followed by $2,n-1,3,n-2,\dots$. In this case, the above encoding scheme requires $O(n\log n)$ bits of space to store the keys, which is no better than the trivial encoding. In comparison, the gap entropy of $S$ is only $\gap(S)=O(n)$.

To improve this scheme, we show that if you store the key of a node relative to not just its parent, but to \emph{all} of its ancestors (storing the difference between the key $k$ and its \emph{closest} ancestor), then you do obtain a space-efficient encoding scheme regardless of the tree structure. Crucially, this encoding scheme is also time-efficient under updates, despite the fact that decoding a node may require knowing the values of all its ancestors.

Moreover, while we use AVL trees as an example in the proof of \cref{lem:binary-tree}, our encoding scheme actually works for any binary search tree that maintains its tree structure using rotations. These include most binary search trees such as splay trees, treaps, red-black trees and so on.

\begin{corollary}[informal]
    Given any binary search tree algorithm for storing a set of keys $S\subset [U]$ that satisfies the following conditions:
    \begin{itemize}
        \item Prior to accessing a node, the algorithm always accesses its parent.
        \item The algorithm maintains its tree structure using rotations only.
    \end{itemize}
    In the word RAM model with word size $w=\Theta(\log U)$, we can encode the tree using $\gap(S)+\lceil \log U\rceil+O(|S|(\log\log U+\log|S|))$ bits of space (not counting the cost of storing auxiliary information such as the size and depth of each node), while preserving the (asymptotic) time cost of any operation on the tree.
\end{corollary}

\subsection{Basic Facts about Binary Search Trees}

Consider using a binary search tree to maintain the set $S$. For concreteness, we use AVL trees (e.g., \cite{adelsonvelskii1962algorithm}) in the following discussion, but the same holds for other trees. Standard implementations of AVL trees answer all queries in $O(\log |S|)$ worst-case time (which is already good enough for us), but require $O(|S|\log U)$ bits of space. In the following, we show that we can improve the space usage of AVL trees to meet the space requirement of \cref{lem:binary-tree}. We assume basic knowledge of AVL trees and binary search trees in general.

Before diving into our algorithm, we first list some basic facts about AVL trees. We are not going to dive into the details of AVL trees; the facts listed here are only those that are relevant to our improvement.

Logically, the AVL tree consists of $|S|$ nodes, where the \defn{content} of each node contains
\begin{itemize}
    \item the key $x\in S$ corresponding to the node (in the following, we refer to this node as $N_x$),
    \item the depth of the subtree rooted at $N_x$ (which is at most $O(\log|S|)$ bits), and
    \item pointers to the children of $N_x$ (if they exist).
\end{itemize}

When processing an update or a query operation on an AVL tree, we would access $O(\log |S|)$ nodes, where if a node is accessed, then so are all its ancestors.

When processing an insertion/deletion operation, the AVL tree maintains its invariants by performing at most $O(\log |S|)$ \defn{rotations}, which are standard subroutines for maintaining binary search trees. See \cref{fig:avl-left-rotate} for an illustration.

\subsection{Physical Implementation of Binary Search Trees in the Variable-Length Word RAM Model (without Difference Encoding)}

Before discussing how to achieve space close to $\gap(S)$, we first describe how to encode a binary search tree in the variable-length word RAM model, without paying any significant space cost for pointers.

In the variable-length word RAM model, we can store an AVL tree as follows:
\begin{itemize}
    \item Each node occupies a word in the memory.
    \item The address of a node in the memory consists of two parts:
    \begin{itemize}
        \item the key of the node's parent ($0$ if it's the root), and
        \item one bit indicating whether the node is a left child ($0$ if it's the root).
    \end{itemize}
    \item In the word of a node $N_x$, we store the content of $N_x$ as defined above, with one twist: Instead of storing pointers to the children of $N_x$, we only need to store $2$ bits denoting whether $N_x$ has a left or right child. This is sufficient since the addresses of the children of $N_x$ are determined by the value of $x$, which is already known.
\end{itemize}

This costs $\lceil \log U\rceil+O(\log\log U+\log |S|)$ bits of space per node, where the first term is the cost of storing the key $x$.

There is one small detail that we need to handle during updates: We need to make sure that the nodes have the right names. This is straightforward since updates to the AVL tree are performed by doing rotations, and we only have to change the names of a constant number of nodes per rotation.

This implementation supports all the operations listed in \cref{lem:binary-tree}. The AVL tree only uses one word per node (which is enough because the word size is $100\log U$), so it has address limit $U$. Thus, we now have a data structure that satisfies all the requirements of \cref{lem:binary-tree} except the space usage, and have eliminated the space needed to store pointers to each node's children. The only remaining problem is that the keys are not difference-encoded. In the following, we improve upon this straightforward implementation.

\subsection{A Tool for Difference Encoding}

We use the following lemma to difference-encode the keys.

\begin{lemma}
    \label{lem:diff-encode-interval}
    Let $L<R$ be integers. For any integer $x\in (L,R)$, there exists an encoding $\enc_{L,R}(x)$ of $x$ using
    \begin{align}
        \label{equ:diff-encode-interval}
        |\enc_{L,R}(x)|=\gap(L,x)+\gap(x,R)-\gap(L,R)+O(\log\gap(L,R))
    \end{align}
    bits. Moreover, in the word RAM model with $w=\Omega(\log (R-L))$, $\enc_{L,R}(x)$ can be encoded and decoded in $O(1)$ time.
\end{lemma}

\cref{lem:diff-encode-interval} is a special case of \cref{lem:diff-encode-larger-interval}, but we restate it here so that this section is self-contained.

\begin{proof}
    In the encoding $\enc_{L,R}(x)$, we store either $x-L$ or $R-x$, depending on which one is smaller. Formally, we first use one bit to indicate whether $x-L<R-x$, then use
    \begin{align*}
        \min\{\gap(L,x),\gap(x,R)\}+O(\log \gap(L,R))
    \end{align*}
    bits to store $\min(x-L,R-x)$.

    It is clear that one can fully recover $x$ from $L,R$ and $\enc_{L,R}(x)$. To compare the space usage to \eqref{equ:diff-encode-interval}, note that
    \begin{align*}
        &\min\{\gap(L,x),\gap(x,R)\} \\
        ={}&\gap(L,x)+\gap(x,R)-\max\{\gap(L,x),\gap(x,R)\} \\
        \le {}&\gap(L,x)+\gap(x,R)-\gap(L,R)+1,
    \end{align*}
    where the last line is due to the fact that $\max\{x-L,R-x\}\ge (R-L)/2$.
\end{proof}

\subsection{Using Difference Encoding in the Binary Search Tree}

\paragraph{The new encoding scheme.} For every node $N_x$ in the AVL tree, define $L(x)$ to be the largest key smaller than $x$ among the ancestors of $N_x$ (if there is no such key, let $L(x)=-1$). Similarly define $R(x)$ (if no such key, $R(x)=U$). Our new encoding scheme is as follows: Instead of storing $x$ explicitly in the description of the node $N_x$, we store the encoding $\enc_{L(x),R(x)}(x)$ obtained by applying \cref{lem:diff-encode-interval}.

\paragraph{Space efficiency.} We now show that the new encoding scheme is space-efficient. Using \cref{lem:diff-encode-interval}, the space cost of encoding key $x$ is
\begin{align}
    \label{equ:single-node-encode}
    \gap(x,R(x))+\gap(L(x),x)-\gap(L(x),R(x))+O(\log\log U+\log |S|).
\end{align}Note that in \cref{lem:binary-tree}, we can afford to use $O(\log\log U+\log |S|)$ bits of extra space per key. Thus, we can ignore the final term for now, and focus on the total cost of the first three terms in \eqref{equ:single-node-encode}. This is done by the following lemma.

\begin{lemma}
    \label{lem:induction-binary-tree}
    Fix any binary search tree, and let $S=\{x_1<\dots<x_m\}\subseteq [U]$ be its set of keys. Define $x_0\defeq -1,x_{m+1}\defeq U$ for notational simplicity. We have
    \begin{align*}
        &\sum_{x\in S}\gap(x,R(x))+\gap(L(x),x)-\gap(L(x),R(x)) \\
        =& \gap(x_0,\dots,x_{m+1})-\gap(x_0,x_{m+1}).
    \end{align*}
\end{lemma}

\begin{proof}
    Proof by induction. When $m=1$, there is only one key $x_1$ in the tree, and we have $L({x_1})=-1,R({x_1})=U$, so the equality holds.

    When $m>1$, let $N_{x_l}$ be a leaf node in the tree. Consider the binary search tree obtained by removing node $N_{x_l}$. On this smaller binary search tree, the values $L(x)$ and $R(x)$ are the same as the original binary search tree for any $x\ne x_l$. Applying the induction hypothesis on the smaller binary search tree, we have
    \begin{align*}
        &\sum_{x\in S,\textcolor{blue}{x\ne x_l}}\gap(x,R(x))+\gap(L(x),x)-\gap(L(x),R(x)) \\
        =& \gap(x_0,\dots,x_{l-1},x_{l+1},\dots,x_{m+1})-\gap(x_0,x_{m+1})\\
        =& \gap(x_0,\dots,x_{m+1})-\gap(x_0,x_{m+1})-\textcolor{blue}{\bk*{\gap(x_l,x_{l+1})+\gap(x_{l-1},x_l)-\gap(x_{l-1},x_{l+1})}}.
    \end{align*}

    Note that the space cost of the smaller binary tree only differs from the space bound that we need to prove for the original binary search tree in the last three terms (marked in blue).  To prove the lemma for the original binary search tree, it remains to show that
    \begin{align*}
        \gap(x_l,R(x_l))+\gap(L(x_l),x_l)-\gap(L(x_l),R(x_l))=\gap(x_l,x_{l+1})+\gap(x_{l-1},x_l)-\gap(x_{l-1},x_{l+1}).
    \end{align*}
    That is, we need to show that $L(x_l)=x_{l-1}$ and $R(x_l)=x_{l+1}$ where $N_{x_l}$ is a leaf node.
    
    In the following, we prove $L(x_l)=x_{l-1}$, and the other equality is symmetric. If $l=1$, then $x_l$ is the smallest key in the binary search tree, and $L(x_l)$ must be $-1=x_{0}$. If $l>1$, then $L(x_l)=x_{l-1}$ if and only if $N_{x_{l-1}}$ is an ancestor of $N_{x_l}$. Now suppose that this is not the case. Let $N_{x'}$ denote the least common ancestor of $N_{x_{l-1}}$ and $N_{x_l}$. Since $N_{x_{l-1}}$ is not an ancestor of $N_{x_l}$, we have $x'\notin \{x_{l-1},x_l\}$, and since we are working on a binary search tree, $x'$ must take value between $x_{l-1}$ and $x_l$, which is impossible.
\end{proof}

Given \cref{lem:diff-encode-interval}, our total space usage is at most
\begin{align*}
    &\gap(\{-1\}\cup S\cup \{U\})-\gap(-1,U)+O(|S|(\log\log U+\log |S|)) \\
    ={}&\gap(S)+O(|S|(\log\log U+\log |S|))+\gap(x_0,x_1)+\gap(x_m,x_{m+1})-\gap(x_0,x_{m+1}).
\end{align*}
Since $\gap(a,b)\le \log U+O(1)$ for any $a<b\in \{-1,\dots,U\}$ and $\gap(x_0,x_{m+1})=\gap(-1,U)\ge\log(U)$, it follows that
\begin{align*}
    \gap(x_0,x_1)+\gap(x_m,x_{m+1})-\gap(x_0,x_{m+1})\le \log U+O(1).
\end{align*}

\paragraph{Time efficiency.} We show that our new encoding scheme can be encoded and decoded efficiently. We only need to show that when accessing a node $N_x$ in the AVL tree, we can efficiently compute the values $L(x),R(x)$, which are needed to decode $\enc_{L(x),R(x)}(x)$.

To compute $L(x),R(x)$, recall that if node $N_x$ is accessed while performing some operation on the AVL tree, then its parent must also be accessed. Therefore, if we can compute the $L,R$ values of a node in constant time given the $L,R$ values of the node's parent, then it is possible to efficiently encode and decode the value of $x$ whenever we access node $N_x$. This is shown in the following claim.

\begin{claim}
    Let $N_x$ be a non-root node in the AVL tree, with $N_y$ being its parent. Given $L(y),R(y)$, we can compute $L(x),R(x)$ in constant time.
\end{claim}

\begin{proof}
    Recall that binary search trees require that any node in the left (resp.~right) subtree of some node $N_x$ has smaller (resp.~larger) key than $x$. Stated differently, any ancestor $N_z$ of $N_y$ for which $z<y$ must also satisfy $z<x$. 

    Therefore, if $x$ is the right child of $y$, then any ancestor of $N_y$ whose key is smaller than $x$ also has key smaller than $x$, which implies that $L(x)=y$. If $x$ is the left child of $y$, then the set of ancestors of $N_y$ having key smaller than $x$ is the same as that of $y$, so $L(x)=L(y)$. Similarly for $R(x)$.
\end{proof}

In addition to efficient encoding/decoding, we also show that (perhaps surprisingly) our new encoding scheme can be efficiently maintained \emph{under updates}. Recall that during each update, the structure of the tree is modified by performing rotations, as described in \cref{fig:avl-left-rotate}.

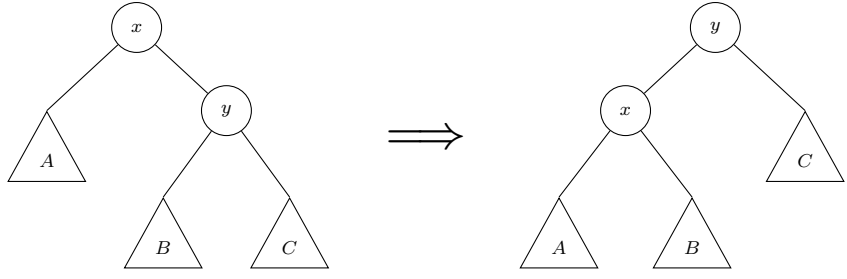
\begin{figure}[H]
    \centering
    \begin{tikzpicture}[
        scale=0.88,
        every node/.append style={transform shape},
        font=\footnotesize,
        >=Stealth,
        nd/.style={circle, draw, minimum size=7.5mm, inner sep=0pt},
        tri/.style={
            isosceles triangle,
            isosceles triangle apex angle=58,
            draw,
            minimum width=11.5mm,
            minimum height=10.5mm,
            inner sep=0pt,
            shape border rotate=90,
            anchor=north,
        },
    ]
        \begin{scope}[shift={(-4.35,0)}]
            \node[nd] (xb) at (0,0) {$x$};
            \node[tri] (Ab) at (-1.35,-1.25) {$A$};
            \node[nd] (yb) at (1.35,-1.25) {$y$};
            \node[tri] (Bb) at (0.4,-2.55) {$B$};
            \node[tri] (Cb) at (2.3,-2.55) {$C$};
            \draw (xb) -- (Ab.north);
            \draw (xb) -- (yb);
            \draw (yb) -- (Bb.north);
            \draw (yb) -- (Cb.north);
        \end{scope}
        \node[scale=2.5] at (0,-1.68) {$\Longrightarrow$};
        \begin{scope}[shift={(4.35,0)}]
            \node[nd] (ya) at (0,0) {$y$};
            \node[nd] (xa) at (-1.35,-1.25) {$x$};
            \node[tri] (Aa) at (-2.35,-2.55) {$A$};
            \node[tri] (Ba) at (-0.35,-2.55) {$B$};
            \node[tri] (Ca) at (1.35,-1.25) {$C$};
            \draw (ya) -- (xa);
            \draw (ya) -- (Ca.north);
            \draw (xa) -- (Aa.north);
            \draw (xa) -- (Ba.north);
        \end{scope}
    \end{tikzpicture}
    \caption{Left rotation in a binary search tree: $y$ becomes the new root of the subtree, the middle subtree $B$ moves from the left child of $y$ to the right child of $x$. Subtrees $A$, $B$, and $C$ keep their internal structure.}
    \label{fig:avl-left-rotate}
\end{figure}

A rotation can change the set of ancestors of many nodes. However, for all nodes but $x,y$, these changes do not affect their $L,R$ values. Indeed, after performing the left rotation, nodes in $A$ gain $y$ as a new ancestor, but $y>x$ and $x$ was already an ancestor larger than every key in~$A$, so the new ancestor $y$ cannot change the $R$ value of any node in $A$. Similarly for $C$. Finally, nodes in $B$ still have both $x$ and $y$ as ancestors, so their $L$ and $R$ values are unchanged. Thus, only $x$ and $y$ themselves need their encodings recomputed after the rotation. Since each rotation only causes $O(1)$ nodes to change their $L,R$ values, the total time needed to perform an update is asymptotically the same as in a classical AVL tree (i.e., $O(\log |S|)$). This concludes the proof of \cref{lem:binary-tree}.

\section{Difference-Encoded B-trees on Tries}

\label{sec:trie-plus-btree}

If we replace \cref{lem:difference-encode-array} with the previous binary tree solution \cref{lem:binary-tree}, we can improve the time cost of \cref{thm:main-warmup} to $O(\log \varepsilon^{-1})$, which is much better than $O(\varepsilon^{-1})$ but still not optimal. In this section, we further improve \cref{lem:binary-tree} and achieve optimal query time for small key sets (at the cost of only achieving an expected space bound). The main result of this section is as follows.

\begin{lemma}
    \label{lem:trie-plus-btree}    
    Let $U$ be an integer and let $\Omega(\log\log U/\log U)\le \varepsilon<1/4$ be a parameter. Let $S\subset [U]$ be a set whose size is always at most $\varepsilon^{-1}$. There exists a legitimate data structure in the variable-length word RAM model with word size $w=100\log U$ that maintains the set $S$. In $O(\log \varepsilon^{-1}/\log\log\varepsilon^{-1})$ worst-case time, the data structure supports the following operations:
    \begin{itemize}
        \item \textbf{Insertion}: Insert a key $x$ to $S$.
        \item \textbf{Deletion}: Delete a key $x$ from $S$.
        \item \textbf{Predecessor/Successor}: Given $x\in [U]$, output the predecessor or successor of $x$ in $S$ (or $\bot$ if it does not exist).
    \end{itemize}
    The data structure also supports an operation called $\textbf{deallocate}$, that deactivates every active word in $O(|S|)$ time. The data structure has address limit $U^3$. It uses
    \begin{align*}
        \gap(S)+\lceil\log U\rceil+O(|S|\log\log U)
    \end{align*}bits of space in expectation, and uses at most $O(|S|\log U)$ bits of space in the worst case. The data structure assumes access to lookup tables and hash functions of total size at most $U^{\delta}$ bits, where $\delta>0$ is an arbitrary constant parameter.
\end{lemma}

\subsection{The Intuition}

If the keys in $S$ are guaranteed to be small (e.g., at most $U^{1/B}$ for some $B$), then we can improve the time cost for maintaining $S$ by replacing the binary tree with a B-tree, where each node stores $B$ keys. This B-tree can be difference-encoded in much the same way as the binary tree in \cref{sec:binary-tree}, giving us a data structure with access time $O(\log\varepsilon^{-1}/\log B)$ (since $|S|$ is always at most $\varepsilon^{-1}$).

This assumption is, of course, not true in general. Surprisingly, the ability to process small keys is useful for solving the general case. In the actual data structure, we set $B=O(\log\varepsilon^{-1}/\log\log\varepsilon^{-1})$ and partition the binary representation of every key evenly into $B$ parts, then use a collection of B-trees to maintain the information of each part. More precisely, for each $1\le i\le B$ and every possible prefix $s\in \{0,1\}^{(\log U/B)\cdot (i-1)}$, we build a B-tree to maintain the $i$-th parts of all the keys in $S$ whose first $i-1$ parts agree with $s$.

A membership query on key $x$ is answered by first querying if any key in $S$ agrees with $x$ on the first part. If true, we then query if any key in $S$ agrees with $x$ on both the first and second parts, and so on. Each of these $B$ queries is answered by querying one of the B-trees.

\paragraph{Time efficiency.} Let $P_i$ denote the number of distinct prefixes $s\in \{0,1\}^{(\log U/B)\cdot (i-1)}$ that appear in $S$. We have $0<P_1\le\dots\le P_{B+1}=|S|$. The time cost of the $i$-th query is, on average over keys in $S$, at most $O(\log_B(P_{i+1}/P_i)+1)$ ($P_{i+1}$ is the total number of keys, and $P_i$ is the number of B-trees). This term telescopes over $i$, giving a total average time cost of $O(\log_B |S|+B)=O(\log\varepsilon^{-1}/\log\log\varepsilon^{-1})$. This is not a worst-case bound because some B-trees might contain more keys than others. To handle this, we use biased B-trees where it is faster to access keys that are ``more important''.

\paragraph{Space efficiency.} To see that it is space-efficient to store this collection of B-trees, we relate each B-tree with a subtree on the trie corresponding to $S$. The space cost of each B-tree is then bounded in terms of the number of edges in the subtree it corresponds to. This allows us to compare the total space usage with the total number of edges in the trie, which, as we will show, is close to $\gap(S)$ (once you randomly shift the keys).

\subsection{Difference-Encoded Biased B-Tree}

We first introduce the biased B-tree construction for maintaining a set of short keys.

\begin{lemma}
    \label{lem:difference-encoded-btree}
    Let $U$ and $1\le B\le \log U$ be integers. Let $U'=2^{\lceil(\log U)/B\rceil}$. There exists a legitimate data structure in the variable-length word RAM model with word size $w=100\log U$ that maintains a set $S\subset [U']$, where each key $x\in S$ is associated with a weight $1\le w_x\le \poly\log U$. The data structure supports the following operations (let $W=\sum_{x\in S}w_x$ denote the total weight of $S$ prior to the operation):
    \begin{itemize}
        \item \textbf{Insertion}: Insert a key $x$ to $S$, whose weight is $w_x$. This operation takes $O(1+\log_B((W+w_x)/w_x))$ time.
        \item \textbf{Deletion}: Delete a key $x$ from $S$. This operation takes $O(1+\log_B(W/w_x))$ time.
        \item \textbf{Predecessor/Successor}: Given $x$, output the predecessor or successor of $x$ in $S$ (or $\bot$ if it does not exist). If the return value is $\bot$, this operation takes $O(1)$ time. Otherwise, letting $y$ denote the return value, this operation takes $O(1+\log_B(W/w_y))$ time.
    \end{itemize}
    The data structure also supports an operation called $\textbf{deallocate}$, that deactivates every active word in $O(|S|)$ time. The data structure has address limit $U$. It uses
    \begin{align}
        \label{equ:btree}
        \gap(S)+\lceil \log U'\rceil+O(|S|(\log\log U+\log|S|))
    \end{align}bits of space. The data structure assumes access to lookup tables and hash functions of total size at most $U^{\delta}$ bits, where $\delta>0$ is an arbitrary parameter.
\end{lemma}

\begin{proof}
    This proof is very similar to the proof of \cref{lem:binary-tree}, so we assume that the reader is familiar with \cref{sec:binary-tree}. Similarly to \cref{lem:binary-tree}, we first describe a solution to \cref{lem:difference-encoded-btree} that is time-efficient but not space-efficient, then apply the difference-encoding scheme of \cref{lem:diff-encode-larger-interval} to improve its space usage.

    \paragraph{Biased B-tree in the variable-length word RAM model.} If we loosen the space requirement to $O(|S|\log U)$ bits, then the biased $a,b$-tree of \cite{feigenbaum1983biased} satisfies our time requirement. We formally describe a modified version of the construction in \cite{feigenbaum1983biased} in the following.

    In a biased B-tree, each node $N$ contains as its \defn{content} the following information:
    \begin{itemize}
        \item A set of $m=O(B)$ keys called \defn{separators}, denoted as $y_1<\dots<y_m$, along with their weights. It is guaranteed that $y_1$ (resp.~$y_m$) is the smallest (resp.~largest) key in the subtree of $N$.
        \item The node $N$ has $m-1$ children, where the $i$-th is responsible for maintaining the set of keys in the interval $(y_{i-1},y_i)$. The content stores the total weight of each child. There is no need to store the pointers to the children, as we will explain below.
        
        It is possible that some $(y_{i-1},y_i)$ does not contain any key in $S$. In this case, we will build an \emph{empty} node for this interval, whose content simply states that the node is empty.
    \end{itemize}
    When storing the contents in the variable-length word RAM model with word size $w=100\log U$, we use one word for each node, where the address of node $N$ is defined as the separator in the parent of $N$ that immediately precedes the interval governed by $N$ (if $N$ is the root, then its address is $-1$). Similarly to the binary tree construction, this definition of address is the reason that we do not need to store pointers to the children of a node.

    It is guaranteed that the total weight of a child of $N$ is at most $O(1/B)$ times the total weight of $N$ itself. This means that a key $x$ must be stored in a node whose depth is at most $O(1+\log_B(W/w_x))$, which implies the desired time bounds on the operations. Also, the minimum and maximum values in a node are always directly stored in the node itself, which explains why it only takes $O(1)$ time to report that a predecessor or successor does not exist.

    Also, the content of each node is a bit string of length $O(\log U)$, so we can perform predecessor/successor queries and other operations on the set of separators in $O(1)$ time when given access to a lookup table of size $U^{\delta}$. The details are omitted, as they are similar to that of a standard B-tree.

    Finally, the space cost consists of the cost of encoding the separators, the cost of storing the weights of the children and separators, and the overhead for using the variable-length word RAM model. Note that both the number of nodes and the total number of separators are $O(|S|)$. Therefore, apart from the cost of encoding the separators, the cost of storing the remaining information is at most $O(|S|(\log\log U+\log|S|))$ bits.

    \paragraph{Applying Difference Encoding.} For each node $N$, let $L(N),R(N)$ denote the separators in the parent of $N$ that immediately precedes and succeeds the interval governed by $N$ (if $N$ is the root, we define $L(N)=-1$ and $R(N)=U$). Instead of directly storing the separators $y_1,\dots,y_m$ of $N$, we store the encoding
    \begin{align*}
        \enc_{L(N),R(N)}(y_1,\dots,y_m)
    \end{align*}
    using \cref{lem:diff-encode-larger-interval}. The operations retain their asymptotic time cost under the new encoding, because we can decode the set of separators of a node in $O(1)$ time whenever we access the node.

    \paragraph{Space efficiency.} Similarly to the proof of \cref{lem:binary-tree}, we can show by induction that when using the above encoding, our B-tree uses
    \begin{align*}
        \gap(S)+\lceil \log U'\rceil+O(|S|(\log\log U+\log|S|))
    \end{align*}
    bits of space in total. We omit the details, as they are the same as \cref{lem:induction-binary-tree}.
\end{proof}

\subsection{The Trie Entropy}

In the proof of \cref{lem:trie-plus-btree}, we use the difference-encoded biased B-trees to maintain the trie corresponding to the set $S$. We first introduce some tools for analyzing the trie.

We define the \defn{trie entropy} (also called the trie measure in \cite{gupta2007compressed}) of a set of keys, which is closely related to the gap entropy. Given a non-empty subset of $[U]$ denoted as $S=\{x_1<x_2<\dots<x_{|S|}\}\subset[U]$, the trie entropy of $S$ is defined as the number of edges in the binary trie corresponding to $S$.

Formally, the trie entropy of $S$ is defined as
\begin{align}
    \label{equ:trie_entropy}
    \trie_U(S)\defeq\lceil\log U\rceil+\sum_{i=2}^{|S|}\lceil\log((x_i\oplus x_{i-1})+1)\rceil.
\end{align}
Here $\oplus$ denotes the bitwise exclusive or. The definition of $\trie$ depends on the universe $U$, and we omit the subscript when it is clear from context.

To see that \cref{equ:trie_entropy} is indeed the number of edges in the trie: When we only have one key, the number of edges in the trie is $\lceil\log U\rceil$. Then, when a key $x_i$ is added to the trie, we find the most significant bit on which $\text{bin}(x_i)$ and $\text{bin}(x_{i-1})$ differ, and create a path starting from that bit to add $x_i$ as a leaf of the trie. The length of this path is thus exactly $\lceil\log((x_i\oplus x_{i-1})+1)\rceil$ bits.

\paragraph{Relation between the gap and trie entropies.} We show that the gap and trie entropies are (in a certain sense) almost the same size.

For each~$i$, one always has
\begin{align}
    \label{equ:gap_atmost_trie_single_i}
    \lceil\log((x_i-x_{i-1})+1)\rceil\le \lceil\log((x_i\oplus x_{i-1})+1)\rceil.
\end{align}
Intuitively, this is because when $x_i-x_{i-1}\ge 2^k$ for some integer $k$, we must have that $\lfloor x_i/2^k\rfloor\ne \lfloor x_{i-1}/2^k\rfloor$, or equivalently, $\lfloor x_i/2^k\rfloor\oplus \lfloor x_{i-1}/2^k\rfloor\ge 1$, meaning that $x_i\oplus x_{i-1}\ge 2^k$.

Summing \eqref{equ:gap_atmost_trie_single_i} for every~$i$ gives
\begin{align}
    \label{equ:gap_atmost_trie}
    \gap(S)+\lceil\log U\rceil\le \trie_U(S)+|S|,
\end{align}so the gap entropy is smaller than the trie entropy (up to an $|S|$ term).

However, the other direction of this inequality does not hold in general: For instance, when $x_{i-1}=2^k-1$ and $x_i=2^k$, we have $ \lceil\log((x_i-x_{i-1})+1)\rceil=1$ but $\lceil\log((x_i\oplus x_{i-1})+1)\rceil=k+1$. That is, the trie entropy may be much larger than the gap entropy.

Fortunately, if we apply a random shift to the keys, then the resulting trie entropy would be much closer to the gap entropy. Formally, given a non-empty subset of $[U]$ denoted as $S=\{x_1<x_2<\dots<x_{|S|}\}\subset[U]$, and letting $\Delta$ denote a uniformly random integer drawn from $[2^{\lceil \log U\rceil}]$, we have that
\begin{align}
    \label{equ:shifted_trie_entropy}
    \E[\trie_{3U}(S+\Delta)]\le \gap(S)+\lceil\log U\rceil+O(|S|),
\end{align}
where $S+\Delta=\{x+\Delta:x\in S\}$ (i.e., the keys in $S+\Delta$ are at most $3U$). Intuitively, after applying a random shift, the aforementioned bad case of $x_{i-1}+\Delta=2^k-1,x_i+\Delta=2^k$ is very unlikely to happen. We refer to Section~2.3 of \cite{gupta2007compressed} for a formal proof of this result. The term $\E[\trie_{3U}(S+\Delta)]$ is later on referred to as the \defn{shifted trie entropy} of $S$.

\subsection{Proof of Lemma \ref{lem:trie-plus-btree}}

\subsubsection{Notations for the Trie}

Let $T$ denote the trie built from the keys in $S$. In the following, we first show a data structure that maintains $S$ using
\begin{align*}
    \trie_{U}(S)+O(|S|(\log\log U+\log|S|))
\end{align*}
bits of space. This is slightly worse than \cref{lem:trie-plus-btree} because the trie entropy could be much larger than the gap entropy. In the end, we will improve this space bound by randomly shifting the keys in $S$.

Let $B=O(\log \varepsilon^{-1}/\log\log \varepsilon^{-1})$ be an integer parameter and assume for the sake of simplicity that $\log U$ is an integer, and that $B$ divides $\log U$. Since $\varepsilon=\Omega(\log\log U/\log U)$, this assumption can be achieved by increasing $\log U$ by at most $O(\log\log U)$, which does not affect the time and space guarantees. 

Let $U'=U^{1/B}$. We say that a node in $T$ is \defn{critical}, if its depth is a multiple of $\log(U')$. In particular, the root is critical because its depth is $0$, and every leaf is critical because their depths are $\log U$.

A \defn{component} of $T$ is defined as a depth-$\log(U')$ subtree rooted at some non-leaf critical node. Observe that every edge belongs to exactly one component, and that the leaves of a component are all critical nodes. The \defn{reduced key} corresponding to a leaf in a component is the bit string corresponding to the path from the root of the component to the said leaf.

The \defn{size} of a component refers to its number of leaves (which must be nonzero by our definition). A component is \defn{degenerate} if its size is $1$, and \defn{non-degenerate} otherwise. For each critical node $x$, its \defn{trailing path} $tp_x$ is a bit string defined as follows: If the component rooted at $x$ is non-degenerate, then $tp_x=\emptyset$. Otherwise, letting $y$ denote the (unique) leaf of the component rooted at $x$, we have $tp_x=v\circ tp_y$. In other words, the trailing path is the concatenation of the degenerate components that follow $x$. We say that a critical node is \defn{non-skippable}, if it is a leaf of a non-degenerate component. In addition, the root is also non-skippable.

The \defn{weight} of a node in $T$ is defined as the number of leaves in its subtree.

\begin{figure}[H]
    \centering
    \begin{tikzpicture}[
        font=\footnotesize,
        x=1cm,
        y=1cm,
        >=Stealth,
        nd/.style={circle, draw, minimum size=7pt, inner sep=0pt},
        compnd/.style={draw=black, double, line width=0.35pt, double distance=0.9pt, line join=round},
        compdeg/.style={ellipse, draw=black, line width=0.7pt, inner xsep=4pt, inner ysep=1pt},
        bitlabel/.style={midway, fill=white, inner sep=1pt, font=\footnotesize\bfseries},
        levellabel/.style={font=\scriptsize, anchor=east},
        legend/.style={font=\scriptsize, anchor=west},
    ]
        \foreach \y/\d in {0/0,-1.7/2,-3.4/4,-5.1/6} {
            \draw[gray!35, double, line width=0.35pt, double distance=0.9pt] (-4.3,\y) -- (4.3,\y);
            \node[levellabel] at (-4.45,\y) {depth $\d$};
        }

        \draw[compnd] (-3.45,-2.15) -- (3.45,-2.15)
            .. controls (2.1,1.35) and (-2.1,1.35) .. cycle;
        \draw[compnd] (-4.0,-3.9) -- (-1.4,-3.9)
            .. controls (-1.9,-0.35) and (-3.5,-0.35) .. cycle;
        \node[compdeg, fit={(-0.1,-1.6) (0.1,-3.5)}] {};
        \node[compdeg, fit={(2.6,-1.6) (2.8,-3.5)}] {};
        \node[compdeg, fit={(-3.5,-3.3) (-3.3,-5.2)}] {};
        \node[compdeg, fit={(-2.1,-3.3) (-1.9,-5.2)}] {};
        \node[compdeg, fit={(-0.1,-3.3) (0.1,-5.2)}] {};
        \draw[compnd] (1.45,-5.5) -- (3.95,-5.5)
            .. controls (3.45,-2.25) and (1.95,-2.25) .. cycle;

        \node[nd] (r) at (0,0) {};

        \node[nd] (n0) at (-1.35,-0.85) {};
        \node[nd] (n1) at (1.35,-0.85) {};

        \node[nd] (n00) at (-2.7,-1.7) {};
        \node[nd] (n01) at (0,-1.7) {};
        \node[nd] (n10) at (2.7,-1.7) {};

        \node[nd] (n000) at (-2.7,-2.55) {};
        \node[nd] (n010) at (0,-2.55) {};
        \node[nd] (n101) at (2.7,-2.55) {};

        \node[nd] (n0000) at (-3.4,-3.4) {};
        \node[nd] (n0001) at (-2.0,-3.4) {};
        \node[nd] (n0101) at (0,-3.4) {};
        \node[nd] (n1010) at (2.7,-3.4) {};

        \node[nd] (n00000) at (-3.4,-4.25) {};
        \node[nd] (n00011) at (-2.0,-4.25) {};
        \node[nd] (n01010) at (0,-4.25) {};
        \node[nd] (n10100) at (2.7,-4.25) {};

        \node[nd] (k000000) at (-3.4,-5.1) {};
        \node[nd] (k000111) at (-2.0,-5.1) {};
        \node[nd] (k010101) at (0,-5.1) {};
        \node[nd] (k101000) at (2.0,-5.1) {};
        \node[nd] (k101001) at (3.4,-5.1) {};

        \draw (r) -- (n0) node[bitlabel, above] {\textbf{0}};
        \draw (r) -- (n1) node[bitlabel, above] {\textbf{1}};
        \draw (n0) -- (n00) node[bitlabel, above] {\textbf{0}};
        \draw (n0) -- (n01) node[bitlabel, above] {\textbf{1}};
        \draw (n1) -- (n10) node[bitlabel, above] {\textbf{0}};

        \draw (n00) -- (n000) node[bitlabel, above] {\textbf{0}};
        \draw (n000) -- (n0000) node[bitlabel, above] {\textbf{0}};
        \draw (n000) -- (n0001) node[bitlabel, above] {\textbf{1}};
        \draw (n01) -- (n010) node[bitlabel, above] {\textbf{0}};
        \draw (n010) -- (n0101) node[bitlabel, above] {\textbf{1}};
        \draw (n10) -- (n101) node[bitlabel, above] {\textbf{1}};
        \draw (n101) -- (n1010) node[bitlabel, above] {\textbf{0}};

        \draw (n0000) -- (n00000) node[bitlabel, above] {\textbf{0}};
        \draw (n00000) -- (k000000) node[bitlabel, above] {\textbf{0}};
        \draw (n0001) -- (n00011) node[bitlabel, above] {\textbf{1}};
        \draw (n00011) -- (k000111) node[bitlabel, above] {\textbf{1}};
        \draw (n0101) -- (n01010) node[bitlabel, above] {\textbf{0}};
        \draw (n01010) -- (k010101) node[bitlabel, above] {\textbf{1}};
        \draw (n1010) -- (n10100) node[bitlabel, above] {\textbf{0}};
        \draw (n10100) -- (k101000) node[bitlabel, above] {\textbf{0}};
        \draw (n10100) -- (k101001) node[bitlabel, above] {\textbf{1}};

        \node[compnd, ellipse, minimum width=5mm, minimum height=3mm] at (-1.0,-6.0) {};
        \node[legend] at (-0.65,-6.0) {non-degenerate};
        \node[compdeg, minimum width=5mm, minimum height=3mm] at (2.0,-6.0) {};
        \node[legend] at (2.35,-6.0) {degenerate};

    \end{tikzpicture}
    \caption{Cutting a trie into components when $\log U=6$ and $B=3$.}
    \label{fig:trie-cut-example}
\end{figure}
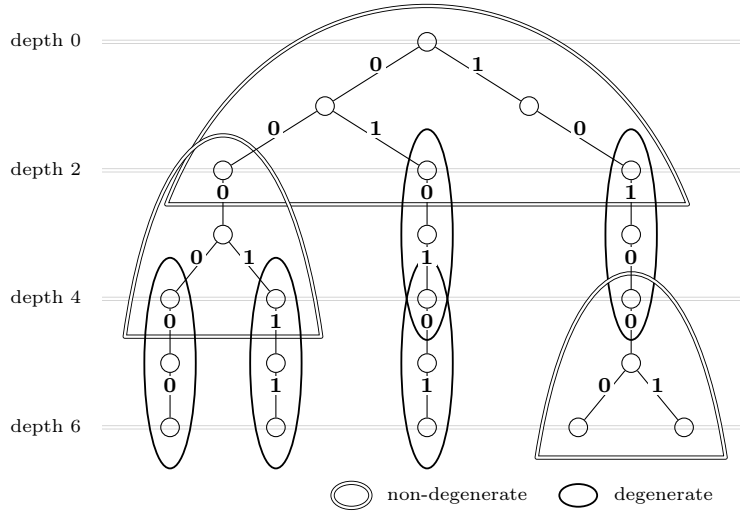

\subsubsection{The Data Structure}

We use biased B-trees to store the information of the components. Formally, for each \emph{non-degenerate} component, we apply \cref{lem:difference-encoded-btree} to maintain its set of reduced keys, where the weight of each reduced key is simply the weight of its corresponding leaf (as a node in $T$). As an example, the biased B-tree of the top component in \cref{fig:trie-cut-example} stores the keys $00,01,10$, and their weights are $2, 1, 2$, respectively. These B-trees are combined using \cref{clm:virtual_address_combine}.

In addition to the B-trees, we use $O(U)$ separate words to store the trailing paths of the \emph{non-skippable} nodes, where we pre-assign each word to a critical node and store the trailing path of the node in that word (if this node is non-skippable). If the corresponding critical node does not exist or is not non-skippable, then the word is idle. Note that, even when the trailing path of a non-skippable node is empty, we still store $O(\log\log U)$ bits of information in its word (instead of letting the word be idle).

\paragraph{Answering queries.} We answer the queries by traversing the trie $T$. Take predecessor queries as an example. We start with the component $\mathcal{C}$ containing the root, and look for the predecessor of $\text{bin}(x)_{1\dots \log(U')}$ among the set of reduced keys. If $\mathcal{C}$ is a degenerate component, then we simply have to compare the unique reduced key with $\text{bin}(x)_{1\dots \log(U')}$. Otherwise, we query the biased B-tree for the predecessor of $\text{bin}(x)_{1\dots \log(U')}$. After finding the predecessor (say it is $y$), we recurse to the component rooted at $y$, and so on.

One thing to keep in mind is that, the information of a degenerate component is stored in the (unique) non-skippable node whose trailing path contains the component. This node is always easy to find, as we perform the query in a top-down fashion.

\paragraph{Handling updates.} Take insertions as an example. When inserting key $x$, we update all the components on the path from the root to $x$ (i.e., the components that are rooted at $\text{bin}(x)_{1\dots \log(U')\cdot i}$ for some $i$). If the component $\mathcal{C}$ rooted at $\text{bin}(x)_{1\dots \log(U')\cdot i}$ did not exist, then we create it, and update the trailing path of its first non-skippable ancestor. Otherwise, if $\mathcal{C}$ did not contain $\text{bin}(x)_{\log(U')\cdot i+1\dots \log(U')\cdot (i+1)}$ as a reduced key, then:
\begin{itemize}
    \item If $\mathcal{C}$ was previously degenerate, we need to update the trailing path of its first non-skippable ancestor, and build a B-tree for $\mathcal{C}$.
    \item Otherwise, we simply add $\text{bin}(x)_{\log(U')\cdot i+1\dots \log(U')\cdot (i+1)}$ into the B-tree of $\mathcal{C}$.
\end{itemize}

Finally, if $\mathcal{C}$ was non-degenerate and \emph{did} contain $\text{bin}(x)_{\log(U')\cdot i+1\dots \log(U')\cdot (i+1)}$, then we need to update the weight of $\text{bin}(x)_{\log(U')\cdot i+1\dots \log(U')\cdot (i+1)}$ in the B-tree. We achieve this by deleting $\text{bin}(x)_{\log(U')\cdot i+1\dots \log(U')\cdot (i+1)}$ from the B-tree, then re-inserting it with the correct weight.

\subsubsection{Performance Analysis}

\paragraph{Time efficiency.} When performing each operation, we will access $B$ components. These components all lie on the path from the root to some node $y$: For insertion and deletion, $y=x$. For predecessor and successor queries, $y$ is the return value, which we assume is not $\bot$ as otherwise the time cost is trivially $O(1)$.

For each $1\le i\le B$, let $W_i$ denote the total weight of the $i$-th accessed component, excluding $w_y$. In particular, let $W_{B+1}=0$. We have that $0=W_{B+1}\le \dots\le W_1\le O(\varepsilon^{-1})$. For each $1\le i\le B$, the time cost for accessing the $i$-th component can be bounded as follows:
\begin{itemize}
    \item If the $i$-th component is degenerate or did not exist prior to the operation, the time cost is $O(1)$.
    \item If the current operation is a predecessor or successor query, then the time cost is $O(1+\log_B((W_i+w_y)/(W_{i+1}+w_y)))$.
    \item Otherwise, if the current operation is an insertion, then we need $O(1+\log_B((W_i+w_y)/(W_{i+1}+w_y)))$ time to insert the key $\text{bin}(y)_{\log(U')\cdot i+1\dots \log(U')\cdot (i+1)}$ into the B-tree. Furthermore, if $W_{i+1}\ne 0$, then the $\text{bin}(y)_{\log(U')\cdot i+1\dots \log(U')\cdot (i+1)}$ did exist in the B-tree but had a different weight, so in order to update the weight, we need to first delete the old occurrence of $\text{bin}(y)_{\log(U')\cdot i+1\dots \log(U')\cdot (i+1)}$, which takes $O(1+\log_B(W_i/W_{i+1}))$ time. Deletions are similar.
\end{itemize}
To summarize, the time cost for accessing the $i$-th component is at most
\begin{align*}
    O\bk*{1+\log_B\bk*{\frac{W_i+w_y}{W_{i+1}+w_y}}}+[W_{i+1}\ne 0]\cdot O\bk*{1+\log_B\bk*{\frac{W_i}{W_{i+1}}}}.
\end{align*}
Let $i^*$ denote the smallest $i$ for which $W_i=0$. If $i^*=1$, then the total time is $O(B)$, so we assume that $i^*>1$. In this case, the total time is at most
\begin{align*}
    &O\bk*{\sum_{i=i^*-1}^{B}1+\log_B\bk*{\frac{W_i+w_y}{W_{i+1}+w_y}}}+O\bk*{\sum_{i=1}^{i^*-2}1+\log_B\bk*{\frac{W_i}{W_{i+1}}}} \\
    \le{}&O(B)+O\bk*{\sum_{i=i^*-1}^{B}\log_B\bk*{\frac{W_i+w_y}{W_{i+1}+w_y}}+\sum_{i=1}^{i^*-2}\log_B\bk*{\frac{W_i}{W_{i+1}}}} \\
    \le{}&O(B)+O\bk*{\log_B\bk*{\frac{W_{i^*-1}+w_y}{W_{B+1}+w_y}}+\log_B\bk*{\frac{W_{1}}{W_{i^*-1}}}} \\
    \le{}&O(B)+O(\log_B(\varepsilon^{-1})).
\end{align*}
Since $B=\Theta(\log\varepsilon^{-1}/\log\log\varepsilon^{-1})$, the total time cost is $O(\log\varepsilon^{-1}/\log\log\varepsilon^{-1})$.

\paragraph{Space efficiency.} The space cost consists of the cost of applying \cref{lem:difference-encoded-btree} to store the non-degenerate components (whose cost is given by \eqref{equ:btree}), and the cost of directly storing the trailing paths for the non-skippable nodes. We now analyze these costs one by one.

\begin{itemize}
    \item The cost of using a B-tree to store a non-degenerate component (whose set of reduced keys is $S'$) can be written as
    \begin{align*}
        \trie_{U'}(S')+O(|S'|\log\log U),
    \end{align*}
    because $\gap(S')+\lceil \log U'\rceil\le \trie_{U'}(S')+O(|S'|)$. Note that $\trie_{U'}(S')$ is the number of edges in the non-degenerate component, and that the sum of $|S'|$ over all non-degenerate components is exactly the number of non-skippable nodes. Therefore, the total cost of storing the B-trees equals the total number of edges in the non-degenerate components, plus $O(\log\log U)$ times the number of non-skippable nodes.
    \item Each degenerate component contributes to the trailing path of exactly one non-skippable node, so the cost of storing the trailing paths is the total number of edges that belong to degenerate components, plus $O(\log\log U)$ times the number of non-skippable nodes.
\end{itemize}

In total, the space cost is equal to $\trie_{U}(S)$ (the total number of edges in the entire trie) plus $O(\log\log U)$ times the number of non-skippable nodes. The latter term is bounded by the following claim.

\begin{claim}
    The number of non-skippable nodes in $T$ is at most $O(|S|)$.
\end{claim}

\begin{proof}
    Recall that $|S|$ is equal to the number of leaves in $T$.
    
    Build a tree $T'$ on the critical nodes in $T$, where the parent of a node is its nearest ancestor in $T$. Let $V$ denote the set of critical nodes, and let $d_x,x\in V$ denote the number of children of $x$ in $T'$. The component rooted at critical node $x$ is non-degenerate if $d_x>1$, and in this case, it contributes $d_x$ to the number of non-skippable nodes. Therefore, the number of non-skippable nodes is (we add an extra $1$ to account for the root)

    \begin{align*}
        &1+\bk*{\sum_{x\in V}[d_x>1]\cdot d_x}\le 1+2\bk*{\sum_{x\in V}d_x-[d_x>0]} \\
        \le{}& O(|V|-1-|V|+\sum_{x\in V}[d_x=0]) \\
        ={}&O(\text{\# of leaves in }T')=O(\text{\# of leaves in }T).
    \end{align*}
    Here we used the fact that $\sum_{x\in V}d_x$ counts the number of edges in the tree, and is equal to $|V|-1$.
\end{proof}

Therefore, our space usage is at most

\begin{align*}
    \trie_{U}(S)+O(|S|\log\log U).
\end{align*}

\paragraph{Applying a random shift.} Let $\Delta$ be sampled uniformly at random in $[2^{\lceil\log U\rceil}]$, and recall that $S+\Delta=\{x+\Delta:x\in S\}$. In the real algorithm, instead of storing $S$, we actually store $S+\Delta$ (as a subset of $[3U]$), so our space cost is

\begin{align*}
    \trie_{3U}(S+\Delta)+O(|S|\log\log U).
\end{align*}

By \cref{equ:shifted_trie_entropy}, its expectation is at most

\begin{align*}
    \gap(S)+\lceil \log U\rceil+O(|S|\log\log U).
\end{align*}

Furthermore, by definition of the trie entropy, the worst case space bound is at most $O(|S|\log U)$, as desired.

\section{Further Improvements Using a Better Distributor}

\label{sec:better-distributor}

Recall that the warm-up data structure \cref{thm:main-warmup} partitions the key set into small chunks, then solves each chunk with a small dictionary and uses a distributor to map the keys into their corresponding chunk. Previously, we optimized the time-space tradeoff of the small dictionaries in \cref{sec:binary-tree,sec:trie-plus-btree}. In this section, we instead shift our focus to optimizing the distributor.

Recall from the preliminaries that a set $S$ of keys can be encoded using
\begin{align*}
    \gap(S)+\lceil\log U\rceil+O((|S|-1)\log\gapbar(S))
\end{align*}
bits of space. The warm-up data structure for $\poly\log U$ many keys, \cref{lem:polylog-dict-warmup}, achieves this space bound (up to an $O(\varepsilon\cdot \gap(S))$ term), but the main result \cref{thm:main-warmup} still contains the suboptimal $O(|S|\log\log U)$ term, which directly comes from applying the distributor. In particular, when $\gap(S)=\Theta(|S|)$ and $\varepsilon=O(1)$, \cref{lem:polylog-dict-warmup} only uses $O(|S|)$ bits of space, while \cref{thm:main-warmup} requires a much larger space of $O(|S|\log\log U)$ bits. In the following, we improve \cref{thm:main-warmup} by presenting a more space-efficient distributor.

\begin{theorem}[Improved Distributor]
    \label{thm:better-distributor}
    Let $U$ be an integer. Let $\Omega(\log\log U/\log U)\le \varepsilon<1/4$ be a parameter. There is a data structure in the variable-length word RAM model with word size $w=100\log U$ that maintains two sets $P,S\subset [U]$, where $0\in P$ always holds, and any two keys $p_1,p_2\in P$ must have $|p_2-p_1|\ge \log U$. For each key $x\in [U]$, its \defn{pivot} is defined as the largest key in $P$ smaller than or equal to $x$ (denoted as $p_x$). Our data structure supports the following operations in amortized expected time:
    \begin{itemize}
        \item \textbf{Insertion/Deletion on $P$}: Insert a key to or delete a key from $P$ in $O(\log U)$ time.

        It is guaranteed that this operation would \emph{not} change the pivot of any key in $S$, and that for each insertion (resp.~deletion), the given key did not exist (resp.~did exist) in $P$.
        \item \textbf{Insertion on $S$ with advice}: Given a key $x\notin S$, \emph{its pivot} $p_x\in P$ and the predecessor and successor of $x$ in $S$, insert $x$ to $S$ in $O(\log\log\varepsilon^{-1})$ time.
        \item \textbf{Deletion on $S$ with advice}: Given a key $x\in S$ and the predecessor and successor of $x$ in $S$, delete $x$ from $S$ in $O(\log\log\varepsilon^{-1})$ time.
        \item \textbf{Pivot}: Given a key $x\in S$, return its pivot $p_x$ in $O(\log\log\varepsilon^{-1})$ time. If the given key is not in $S$, the distributor may return anything.
    \end{itemize}
    The data structure has address limit $U^{10}$. It uses
    \begin{align*}
        O(|P|\cdot \poly\log U+\varepsilon \cdot \gap(S)+|S|\log \gapbar(S))
    \end{align*}
    bits of space in expectation, and uses $O(|P|\cdot \poly\log U+|S|\log U)$ bits of space in the worst case. The data structure assumes access to lookup tables and hash functions of total size at most $U^{\delta}$ bits, where $\delta>0$ is an arbitrary parameter.
\end{theorem}

\paragraph{Differences between the statements of \cref{thm:better-distributor} and \cref{thm:distributor}.} \cref{thm:better-distributor} achieves the correct space bound, at the cost of adding more constraints compared to \cref{thm:distributor}. However, these new constraints are not substantial: Despite all these differences, \cref{thm:better-distributor} can still be used in the same way as \cref{thm:distributor}. We briefly discuss the differences in the following.

\begin{itemize}
    \item The time and space costs of maintaining the pivots increases. This is fine because later in the main result, we will partition the key set into chunks of $\Theta(\log^{C_{\text{chunk}}}U)$ keys for a sufficiently large $C_{\text{chunk}}$. As a result, we can allow an extra time and space cost of $\poly\log U$ per pivot.
    \item The time and space bounds of \cref{thm:better-distributor} depend on the parameter $\varepsilon$, whereas \cref{thm:distributor} does not. When proving the main result, we will take the parameter $\varepsilon$ to be the same as the parameter in the overall data structure. Under this instantiation, the time and space bounds of \cref{thm:better-distributor} are acceptable.
    \item \cref{thm:better-distributor} only achieves an expected space bound. We later overcome this by first building a dictionary similar to the warm-up that has expected space guarantee, then in the actual algorithm partition the key set into $\poly\log U$ subsets (maintained in the same way as the chunks), and maintain each subset using the expected-space dictionary. By a simple Chernoff bound, the overall space bound will hold with high probability.
    \item When inserting or deleting a key in $S$, \cref{thm:better-distributor} requires more advice than \cref{thm:distributor}. In particular, we need to provide the distributor with the predecessor and successor of $x$ in $S$. Note that when accessing the distributor, we always know which chunk contains $x$ (this is discussed in more detail later in \cref{sec:final}).
    
    Since the small dictionaries (such as \cref{lem:trie-plus-btree}) support predecessor and successor queries, and we know which chunk contains $x$, the advice can be computed by querying the small dictionaries that correspond to the chunk containing $x$ and its neighboring chunks.
\end{itemize}

\subsection{A Warm-Up to Proving Theorem \ref{thm:better-distributor}}

Recall that the $O(|S|\log\log U)$ term in \cref{thm:distributor} is the cost of using a retrieval data structure to store the values $v_x$ for every $x\in S$, where $v_x$ denotes the length of the largest common prefix of the binary representations of $x$ and $p_x$, and is at most $\lceil\log U\rceil$. To improve this bound, we observe that our goal of $O(|S|\log \gapbar(S))$ is only better than $O(|S|\log\log U)$ when $\gapbar(S)$ is small, so what we need is a way to store $v_x$ more efficiently for small $\gapbar(S)$. In this section, as a warm-up, we present a distributor that uses
\begin{align*}
    O(|P|\cdot \poly\log U)+o(|S|\log\log U)
\end{align*}bits of space in expectation when $\gapbar(S)=\log^{o(1)}U$ and $\gapbar(S)$ is fixed over time (up to a constant factor).

\paragraph{Intuition.} To exploit the fact that $\gapbar(S)$ is small, we set a threshold value $T$ (where $\gapbar(S)\le T\ll \log U$) and consider two cases:

\begin{itemize}
    \item If $x_i$ and its pivot $p_{x_i}$ only differ in the last $T$ bits, then $v_x$ only has $T$ possibilities, so we only have to spend $O(\log T)$ bits of space to store $v_{x_i}$, instead of $\log\log U$ bits.

    One thing to note is that, in order to store the values using $O(\log T)$ bits of space per key, we need to reduce the \emph{universe size} (i.e., the maximum number of keys) of the retrieval data structure as well as the value size (see the definition of \cref{thm:retrieval}). In the algorithm, we will store the $v_x$'s using a collection of retrieval data structures, each of which has a universe size of $O(2^T)$.
    \item If $x_i$ and its pivot $p_{x_i}$ differ in more than the last $T$ bits, we know that $v_{x_i}$ is the same as $v_{x_i-(x_i\bmod 2^T)}$ (because $x_i$ and $x_i-(x_i\bmod 2^T)$ only differ in the last $T$ bits), so we only have to store the value of $v_{x_i-(x_i\bmod 2^T)}$ for every $x_i$. That is, we pay
    \begin{align*}
        O(|\{x-(x\bmod 2^T):x\in S\}|\cdot \log\log U)
    \end{align*}bits of space instead of $O(|S|\log\log U)$ bits.

    When $T$ is much larger than $\gapbar(S)$, for most adjacent pairs of keys $x_{i-1},x_i$, we would have
    \begin{align*}
        x_{i-1}-(x_{i-1}\bmod 2^T)=x_i-(x_i\bmod 2^T)
    \end{align*}
    (if we randomly shift the keys in $S$). This means that the number of distinct keys in
    \begin{align*}
        \{x-(x\bmod 2^T):x\in S\}
    \end{align*}should be much smaller than $|S|$, allowing us to save space.
\end{itemize}

\paragraph{The construction.} We assume for simplicity that $U$ is a power of $2$. We also assume that the sets $P,S$ have been randomly shifted. That is, we randomly sample an integer $\Delta\in [U]$ at the start of the algorithm, and add $\Delta$ to every key. As a result, our actual universe size becomes $2U$.

For any key $x\in [2U]$, we define its \defn{pivot} $p_x$ to be the largest key in $P$ smaller than or equal to $x$, and let $v_x$ denote the length of the largest common prefix of the binary representations of $x$ and $p_x$. As a convention, we always assume that $0\in P$ (after shifting), so that $p_x$ is well-defined.

Let $1\le T\le \log U$ be an integer parameter, whose actual value is to be determined later. A \defn{sub-universe} is an interval of the form $[i\cdot 2^T,(i+1)\cdot 2^T)\subset [2U]$ for some integer $i$. A sub-universe is \defn{nontrivial} if it contains a key in $P$.

Our data structure consists of the following parts:

\begin{itemize}
    \item A dictionary maintaining $P$, supporting predecessor and successor queries on $P$ in $O(\log U)$ time.
    \item \textbf{The sub-universe dictionary}: A key-value dictionary that maintains the nontrivial sub-universes. For each nontrivial sub-universe $W$, we store as its value the smallest key in $P\cap W$.
    \item \textbf{The level-0 retrieval data structure}: Define
    \begin{align*}
        S_T^{(0)}\defeq \{x-(x\bmod 2^T):x\in S\land (x\text{ and }p_x\text{ belong to different sub-universes})\}.
    \end{align*}
    We use a retrieval data structure to maintain the mapping $x\mapsto v_x$, where $x\in S_T^{(0)}$. Note that, as discussed in the intuitions, when $x$ and $p_x$ belong to different sub-universes, $v_x=v_{x-(x\bmod 2^T)}$.
    \item \textbf{The level-1 retrieval data structures}: Define 
    \begin{align*}
        S_T^{(1)}\defeq \{x:x\in S\land (x\text{ and }p_{x}\text{ belong to the same sub-universe})\}.
    \end{align*}
    For each nontrivial sub-universe $W$, we use a retrieval data structure to maintain the mapping $x\mapsto v_x-(\log (2U)-T)$, where $x\in S_T^{(1)}\cap W$. In this case, since $x$ and $p_x$ belong to the same sub-universe, $v_x$ must be at least $\log(2U)-T$.

    For any retrieval data structure that we use, we add $\poly\log U$ dummy keys to satisfy the requirement of \cref{thm:retrieval}.

    \item \textbf{The prefix-to-pivot dictionary}: The above data structures are responsible for finding $v_x$, not $p_x$. Finally, we give a data structure that finds $p_x$ given $x$ and $v_x$.
    
    Formally, the keys in this dictionary correspond to non-empty bit strings $s\in \{0,1\}^{\le \log U+1}$. A bit string $s$ exists in the dictionary if there exists a key in $P$ whose binary representation contains $s$ as a prefix, and the value associated with $s$ is the largest such key. This dictionary is also present in \cref{thm:distributor}.
\end{itemize}

These parts are combined using \cref{clm:virtual_address_combine}. Each part is a data structure in the word RAM model that uses at most $U^2$ words, so the overall data structure has address limit $U^{10}$.

Intuitively, compared with \cref{thm:distributor}, we use multiple retrieval data structures instead of one to store the mapping $x\mapsto v_x$, and use the sub-universe dictionary to figure out which retrieval data structure we should query. This finer partitioning of retrieval data structures is the essence of our space improvement.

\paragraph{Answering queries.} Given a key $x\in S$ in sub-universe $W$, we first query the sub-universe dictionary to see whether $W$ is nontrivial, and to find the smallest key in $P\cap W$ if one exists (denoted as $p$). Given this, we can tell whether $x$ and $p_x$ belong to different sub-universes: If $x$ and $p_x$ belong to the same sub-universe, then in particular $p_x$ must be in $W$, so $W$ is nontrivial. Furthermore, since $p_x\le x$ and $p$ is the smallest key in $P\cap W$, we must have $x\ge p$. Conversely, if $W$ is trivial or $x<p$, then $x$ and $p_x$ must be in different sub-universes.

If $x$ and $p_x$ belong to different sub-universes, then we can learn the value of $v_x$ by querying the level-0 retrieval data structure at $x-(x\bmod 2^T)$. Otherwise, we know that $v_x\ge \log(2U)-T$, and can learn the value of $v_x$ by querying the level-1 retrieval data structure corresponding to $W$.

Next, if $v_x=\log U+1$, then we have $x=p_x$. Otherwise, $p_x$ must be equal to the largest key in $P$ whose prefix equals $\text{bin}(x)_{1\dots v_x}\circ \text{``0''}$, as argued in \cref{thm:distributor}, so we can find $p_x$ by querying the prefix-to-pivot dictionary.

In total, answering a pivot query takes $O(1)$ time.

\paragraph{Updating $P$.} Let $p$ denote the queried key, and let $W$ denote its sub-universe.

\begin{itemize}
    \item  Modifying the dictionary for $P$ takes $O(1)$ time.
    
    \item The sub-universe dictionary can be answered by querying the predecessor and successor of $p$ in $P$, which takes $O(\log U)$ time.

    \item The level-0 retrieval data structure is unchanged, because it is guaranteed that an update to $P$ does not change the pivot of any key in $S$. This implies that the set $S_T^{(0)}$ is unchanged. Furthermore, the value $v_x$ is also unchanged for any $x\in S_T^{(0)}$ because it is equal to some $v_{x'}$ where $x'\in S$, and $v_{x'}$ is unchanged because of the guarantee.

    \item We may need to create or destroy a level-1 retrieval data structure, because the sub-universe $W$ containing $p$ may change its state. Note that, if $W$ changes its state during the current operation, then the level-1 retrieval data structure corresponding to $W$ must be empty (save for the $\poly\log U$ dummy keys). This is because, if $S_T^{(1)}\cap W$ contains some key $x$, then $p_x$ must also be in $W$ and must be distinct from the key $p$ that is being updated (otherwise the pivot of $x$ changes during the update, which contradicts our guarantee). As a result, creating or destroying a level-1 retrieval data structure only involves taking care of the dummy keys, which takes $\poly\log U$ time.
    
    \item Finally, updating the prefix-to-pivot dictionary takes $O(\log U)$ time, which has been described in \cref{thm:distributor} and is omitted.
\end{itemize}
    
In total, an update to $P$ takes $\poly\log U$ time.

\paragraph{Updating $S$.} Let $x$ denote the queried key, and let $W$ denote its sub-universe. When updating $S$, we only have to modify the retrieval data structures.

We first check whether $x$ contributes to $S_T^{(0)}$ and $S_T^{(1)}$. This only depends on the values of $x$ and $p_x$, and we know $p_x$ (when inserting $x$, we are given $p_x$ as advice; when deleting $x$, we can query the distributor to get $p_x$). Therefore, we can decide this at no cost.

If $x$ contributes to $S_T^{(1)}$, then we need to update the level-1 retrieval data structure corresponding to $W$. This retrieval data structure must exist, because in this case, $x$ and $p_x$ belong to the same sub-universe, which in particular means that $W$ is nontrivial.

As for the level-0 retrieval data structure, in addition to checking whether $x$ contributes to $S_T^{(0)}$, we also need to check whether the current update changes $S_T^{(0)}$. That is, whether there exists another key $x'\in S\setminus\{x\}$, such that $x'$ and $p_{x'}$ belong to different sub-universes, and $x'-(x'\bmod 2^T)=x-(x\bmod 2^T)$ (which is equivalent to saying that $x'\in W$). Letting $p$ denote the smallest key in $P\cap W$ (which can be learned by querying the sub-universe dictionary), our task is reduced to asking: Does there exist a key in $S\cap [x-(x\bmod 2^T),p)$ other than $x$? This question is easy to answer because we are given the predecessor and successor of $x$ in $S$ as advice, as stated in \cref{thm:better-distributor}. This concludes the algorithm for updating $S$, which takes $O(1)$ time.

\paragraph{Space efficiency.} 

\begin{itemize}
    \item The dictionary maintaining $P$ uses $O(|P|\cdot \log U)$ bits of space.
    \item The number of nontrivial sub-universes is at most the number of keys in $P$, so the sub-universe dictionary also uses $O(|P|\cdot \log U)$ bits of space.
    \item The space cost of the level-0 retrieval data structure is $O(|S_T^{(0)}|\cdot \log\log U)$, where we bound the size of $S_T^{(0)}$ later. Note that the retrieval data structure of \cref{thm:retrieval} satisfies its space bound with high probability, which also means that the space bound holds in expectation. This is good enough for us.
    \item For each level-1 retrieval data structure, its universe size is at most $2^T+\poly\log U$ (the sub-universe has size $2^T$, and we add $\poly\log U$ dummy keys), and the values are also at most $O(T)$. Therefore, by \cref{thm:retrieval}, the level-1 retrieval data structures only cost $\log T+\log\log\log U$ bits of space per key, which sums up to
    \begin{align*}
        O(|P|\cdot \poly\log U+|S|(\log T+\log\log\log U))
    \end{align*}
    bits (the first term accounts for the dummy keys).
    \item The prefix-to-pivot dictionary costs $O(|P|\cdot \poly\log U)$ bits of space, as in \cref{thm:distributor}.
\end{itemize}

Finally, each dynamic dictionary and retrieval data structure assumes access to a hash function of size $U^{\delta}$ bits. Since we are happy with an expected space bound, we only have to sample one such hash function, and use it on every data structure.

The total space cost is
\begin{align*}
    O(|P|\cdot \poly\log U+|S_T^{(0)}|\cdot \log\log U+|S|(\log T+\log\log\log U)).
\end{align*}

We now bound the expected size of $S_T^{(0)}$ as a function of $|S|$ and $\gapbar(S)$. Actually, we will bound the size of a slightly larger set
\begin{align*}
    S'=\{x-(x\bmod 2^T):x\in S\}.
\end{align*}

To bound $|S'|$, we count the number of pairs of adjacent keys in $S$ for which the values of $\lfloor x/2^T\rfloor$ differ. Recall that $S$ has been randomly shifted by $\Delta$, which is uniformly random from $[U]$, and that $U$ is a power of $2$. Letting $x_1<\dots<x_{|S|}$ denote the (unshifted) keys in $S$, the size of $|S'|$ can be expressed as
\begin{align}
    \label{equ:S'}
    |S'|=1+\sum_{i=2}^{|S|}[\lfloor (x_i+\Delta)/2^T\rfloor\ne \lfloor (x_{i-1}+\Delta)/2^T\rfloor].
\end{align}

For each $i$, we have that
\begin{align*}
    &\Pr[\lfloor (x_i+\Delta)/2^T\rfloor\ne \lfloor (x_{i-1}+\Delta)/2^T\rfloor] \\
    ={}&\Pr[(x_i+\Delta)\oplus (x_{i-1}+\Delta)\ge 2^T] \\
    \le{}&\min(1,(x_i-x_{i-1})/2^T).
\end{align*}
Intuitively, as $x_i-x_{i-1}$ increases, it becomes more likely for $x_i+\Delta$ and $x_{i-1}+\Delta$ to differ in more than the last $T$ bits. Plugging this back into \eqref{equ:S'}, we have that the expected size of $S'$ is
\begin{align*}
    \E[|S'|]&=1+\sum_{i=2}^{|S|}\min(1,(x_i-x_{i-1})/2^T) \\
    &\le 1+\sum_{i=2}^{|S|}\min \bk*{1, 2^{\gap(x_{i-1},x_i)-T}}\tag{$\gap(x,x')=\log(x'-x+1)+O(1)$}
\end{align*}

We want to bound this expression as a function of $|S|$ and $\gapbar(S)$, where $\gapbar(S)$ is the average gap entropy of $S$. That is, we need to maximize $\E[|S'|]$ when $|S|$ and $\gapbar(S)$ are fixed. Observe that it is never beneficial to let $\gap(x_{i-1},x_i)>T$ for any $i$, and that when $\gap(x_{i-1},x_i)$ is always in $[2,T]$, each term $\min \bk*{1,2^{\gap(x_{i-1},x_i)-T}}$ is convex in $\gap(x_{i-1},x_i)$. This means that, we can maximize $\E[|S'|]$ by setting $\gap(x_{i-1},x_i)=T$ for as many $i$ as possible, and setting $\gap(x_{i-1},x_i)=2$ for the others. The maximum number of indices $i$ for which $\gap(x_{i-1},x_i)=T$ is $\gap(S)/T=|S|\cdot \gapbar(S)/T$. Therefore, $\E[|S'|]$ is at most
\begin{align}
    \label{equ:E_S'}
    \min\bk*{|S|,O\bk*{\frac{|S|\cdot \gapbar(S)}{T}}}.
\end{align}

\paragraph{Setting the threshold $T$.} Recall the goal of this warm-up: When $\gapbar(S)=\log^{o(1)}U$ and is fixed over time (up to a constant factor), we want to achieve $O(|P|\cdot \poly\log U)+o(|S|\log\log U)$ bits of space. As we just established, the total space cost is
\begin{align*}
    O\bk*{|P|\cdot \poly\log U+\min\bk*{|S|,O\bk*{\frac{|S|\cdot \gapbar(S)}{T}}}\cdot \log\log U+|S|(\log T+\log\log\log U)}.
\end{align*}
We can set $T=\gapbar(S)\cdot \log\log U$, so that when $\gapbar(S)=\log^{o(1)}U$, the total space cost is at most
\begin{align*}
    &O(|P|\cdot \poly\log U)+O(|S|)+O\bk*{|S|\cdot (\log\gapbar(S)+\log\log \log U)} \\
    ={}&O\bk*{|P|\cdot \poly\log U+|S|\cdot (\log\gapbar(S)+\log\log \log U)} \\
    ={}&O\bk*{|P|\cdot \poly\log U}+o(|S|\cdot \log\log U).
\end{align*}

\subsection{The Proof of Theorem \ref{thm:better-distributor}}

Upon closer inspection, the warm-up already achieves the correct space bound of
\begin{align*}
    O\bk*{|P|\cdot \poly\log U+|S|\cdot \log\gapbar(S)}
\end{align*}
when $\gapbar(S)\ge (\log\log U)^{\Omega(1)}$ and $\gapbar(S)$ is fixed over time. In the following, we remove these two constraints.

To remove the requirement of $\gapbar(S)\ge (\log\log U)^{\Omega(1)}$, we note that any two keys in $P$ are at distance at least $\log U$ from each other. This is explicitly stated in \cref{thm:better-distributor}, and is naturally satisfied by our algorithm because each key in $P$ is the minimal value of a chunk, and each chunk contains $\poly\log U$ keys in $S$. Given this property, we can solve the case of $\gapbar(S)= (\log\log U)^{o(1)}$ by letting $T=(1/2)\log\log U$. Under this parametrization, each sub-universe has size $<\log U$, so it contains at most one pivot. As a result, the level-$1$ retrieval data structures are not needed, so we can save space.

To remove the requirement that $\gapbar(S)$ is fixed over time, we use \emph{multiple thresholds} simultaneously in our data structure. Intuitively, the warm-up can be seen as partitioning the problem into subproblems of universe size $2^T$, then using a retrieval data structure to solve each problem. We can further reduce the problem size by replacing each level-$1$ retrieval data structure with a distributor (using threshold $T'$), and so on and so forth. When the thresholds are set appropriately, we can guarantee that for any value of $\gapbar(S)$, the distributor behaves as if it only has one threshold $T=\gapbar(S)\cdot \log\log U$, which is just like the warm-up. The details are as follows.

\paragraph{Notations.} Define the thresholds as follows: Let $T_1=\min(\lceil \varepsilon^{-1}\cdot (\log\log U)^2\rceil,\log(2U))$, $T_2=\lfloor\sqrt{T_1}\rfloor,T_3=\lfloor\sqrt{T_2}\rfloor,\dots$, all the way until $T_k=\lfloor\sqrt{T_{k-1}}\rfloor\in [\sqrt{\log\log U},\log\log U)$. Note that the number of thresholds is $k=O(\log\log\varepsilon^{-1})$. We also define $T_0=\log(2U)$ and $T_{k+1}=0$ for notational simplicity.

For each $0\le j\le k$, the \defn{level-$j$ sub-universes} are intervals of the form $[i\cdot 2^{T_j},(i+1)\cdot 2^{T_j})\subset [2U]$, where $i$ is some integer. A sub-universe is \defn{nontrivial} if it contains a key in $P$. 

For each $x\in S$, its \defn{critical level} $l_x$ is defined as the largest level $0\le j\le k$ for which $x$ and $p_x$ belong to the same level-$j$ sub-universe. For each $0\le j<k$, define
\begin{align*}
    S_T^{(j)}\defeq\{x-(x\bmod 2^{T_{j+1}}):x\in S\land l_x=j\}.
\end{align*}
Observe that, if we only have one level (i.e., $k=1$), then this definition is consistent with the definitions of $S_T^{(0)}$ and $S_T^{(1)}$ in the warm-up. 

Similarly to the warm-up, we have the following fact.

\begin{fact}
    \label{fact:distributor-level}
    Let $0\le j<k$. For any key $x\in S$ such that $l_x=j$, we have that $v_x=v_{x-(x\bmod 2^{T_{j+1}})}$. Furthermore, $v_x\ge \log(2U)-T_j$.
\end{fact}

\begin{proof}
    If $l_x=j$, then $x$ and $p_x$ belong to different level-$(j+1)$ sub-universes. This means that $x'=x-(x\bmod 2^{T_{j+1}})$ is in the range $(p_x,x]$, so its pivot must also be equal to $p_x$. Then, since $x$ and $p_x$ differ somewhere above the last $T_{j+1}$ bits, and $x$ and $x'$ agree everywhere above the last $T_{j+1}$ bits, we must have $v_x=v_{x'}$.

    Moreover, since $x$ and $p_x$ belong to the same level-$j$ sub-universe, they agree everywhere above the last $T_j$ bits, so $v_x\ge \log(2U)-T_j$.
\end{proof}

\paragraph{The construction.} Our data structure is very similar to the warm-up, with only the following changes:

\begin{itemize}
    \item \textbf{The sub-universe dictionary}: We now have $k$ sub-universe dictionaries, corresponding to the levels $1\dots k$. The level-$j$ sub-universe dictionary maintains the set of nontrivial level-$j$ sub-universes, where for each sub-universe $W$, we store as its value the smallest key in $P\cap W$.
    \item \textbf{The retrieval data structures}: For each $0\le j< k$, and for each nontrivial level-$j$ sub-universe $W$, we build a level-$j$ retrieval data structure. This retrieval data structure stores the mapping $x\mapsto v_x-(\log(2U)-T_j)$ for every $x\in S_T^{(j)}\cap W$. Note that these values are always nonnegative by \cref{fact:distributor-level}.

    Similarly to the warm-up, we add $\poly\log U$ dummy keys to each retrieval data structure.
\end{itemize}

Note a very important difference between this algorithm and the warm-up: We do not have level-$k$ retrieval data structures. This is because $T_k<\log\log U$, and \cref{thm:better-distributor} guarantees that the distance between any two distinct keys in $P$ is at least $\log U$, which means that a level-$k$ sub-universe can have at most one key in $P$. Therefore, if the critical level of some key $x\in S$ is $l_x=k$, then we can directly know $p_x$ by querying the sub-universe dictionary, without having to build retrieval data structures for such keys.

\paragraph{Answering queries.} Given a key $x\in S$, let $W^{(0)}\supset\dots\supset W^{(k)}$ denote the sub-universes of different levels that contain $x$. 

We first find the critical level $l_x$. For each $0\le j\le k$, we query the sub-universe dictionary to see if $W^{(j)}$ is trivial. If not, then the sub-universe dictionary also tells us $p^{(j)}$, the smallest key in $P\cap W^{(j)}$. Similarly to the warm-up, we have the following condition: $x$ and $p_x$ belong to the same level-$j$ sub-universe, if and only if $W^{(j)}$ is nontrivial and $p^{(j)}\le x$. Using this, we can find the critical level $l_x$.

If $l_x=k$, then we must have $p_x=p^{(k)}$, because the distance between any two keys in $P$ is at least $\log U$, and the size of a level-$k$ sub-universe is $2^{T_k}<\log U$.

Otherwise, we query the level-$l_x$ retrieval data structure corresponding to $W^{(l_x)}$, which gives us the value of $v_{x-(x\bmod 2^{T_{l_x+1}})}$, which is equal to $v_x$ by \cref{fact:distributor-level}.

\paragraph{Handling updates.} The updates are almost the same as the warm-up (except that we have to do the same thing $k$ times), and are omitted.

\paragraph{Time efficiency.} Each level of the distributor can be maintained as in the warm-up. Since we now have $O(\log\log\varepsilon^{-1})$ levels, the time for performing each query blows up by an $O(\log\log\varepsilon^{-1})=o(\log U)$ factor, giving us the final time bounds in \cref{thm:better-distributor}.

\paragraph{Space efficiency.} Everything other than the retrieval data structures uses $O(|P|\cdot \poly\log U)$ bits of space in total. Similarly for the dummy keys in the retrieval data structures.

This only leaves the cost of storing the real keys in the retrieval data structures. A real key in a level-$j$ retrieval data structure costs $O(\log T_j)$ bits of space. In total, they cost at most
\begin{align*}
    O\bk*{\sum_{j=0}^{k-1}|S_T^{(j)}|\cdot \log T_j}
\end{align*}
bits of space. For each $S_T^{(j)}$, the expected space cost can be bounded similarly to \eqref{equ:E_S'}, so the total expected space is at most
\begin{align*}
    &O\bk*{\sum_{j=0}^{k-1}\min\bk*{|S|,\frac{|S|\cdot\gapbar(S)}{T_{j+1}}}\cdot \log T_j} \\
    ={}&|S|\cdot O\bk*{\sum_{j=0}^{k-1}\min\bk*{1,\frac{\gapbar(S)}{T_{j+1}}}\cdot \log T_j}.
\end{align*}
bits (here we ignore the $\log\log\log U$ term that was present in the warm-up, because $T_i\ge \log\log^{\Omega(1)}U$ for any $i$). Comparing this with the space requirement of \cref{thm:better-distributor}, it remains to show that, for any value of $\gapbar(S)$, we have
\begin{align}
    \label{equ:multiple-thresholds-goal}
    \sum_{j=0}^{k-1}\min\bk*{1,\frac{\gapbar(S)}{T_{j+1}}}\cdot \log T_j=O\bk*{\varepsilon\cdot \gapbar(S)+\log \gapbar(S)}.
\end{align}

We consider three different regimes of $\gapbar(S)$.

\begin{itemize}
    \item If $\gapbar(S)\ge \min(\varepsilon^{-1}\cdot \log\log U,\sqrt{\log U})$, we have
        \begin{align*}
            \varepsilon\cdot \gapbar(S)+\log \gapbar(S)\ge \Omega(\log\log U).
        \end{align*}
        In this case, \eqref{equ:multiple-thresholds-goal} holds as long as the LHS is at most $O(\log\log U)$. This is indeed true, because
        \begin{align*}
            &\sum_{j=0}^{k-1}\min\bk*{1,\frac{\gapbar(S)}{T_{j+1}}}\cdot \log T_j\le \sum_{j=0}^{k-1} \log T_j \\
            \le{}& O\bk*{\sum_{j=0}^{k-1}(\log\log U)/2^{j}} \tag{definition of $T_j$} \\
            ={}&O(\log\log U).
        \end{align*}
    \item If $\gapbar(S)\le \min(\varepsilon^{-1}\cdot \log\log U,\sqrt{\log U})$ but $\gapbar(S)\ge (\log\log U)^{1/4}$: Let $j^*$ be the largest level $1\le j\le k$ such that $T_j\ge \gapbar(S)\cdot \log\log U$. Such a level must exist, because $T_1=\min(\varepsilon^{-1}\cdot(\log\log U)^2,\log (2U))$ by our definition, which is at least $\gapbar(S)\cdot \log\log U$. 
    
    Next, we break the LHS of \eqref{equ:multiple-thresholds-goal} into two parts.
    \begin{itemize}
        \item For the levels before $j^*$, we have
        \begin{align*}
            &\sum_{j=0}^{j^*-1}\min\bk*{1,\frac{\gapbar(S)}{T_{j+1}}}\cdot \log T_j\le \sum_{j=0}^{j^*-1}\min\bk*{1,\frac{\gapbar(S)}{T_{j+1}}}\cdot \log \log U \\
            \le{}&\sum_{j=0}^{j^*-1}\frac{\gapbar(S)}{T_{j^*}}\cdot \frac{T_{j^*}}{T_{j+1}}\cdot \log \log U \\
            \le{}&\sum_{j=0}^{j^*-1}\frac{T_{j^*}}{T_{j+1}}\tag{$T_{j^*}\ge \gapbar(S)\cdot\log\log U$}\\
            \le{}&\sum_{j=0}^{j^*-1}2^{(j+1)-j^*}=O(1)=O(\log\gapbar(S)).\tag{$T_j=T_{j+1}^2\ge 2T_{j+1}$ for any $j$}
        \end{align*}
        \item For the levels including and after $j^*$, note that since
        \begin{align*}
            T_{j^*+1}=\sqrt{T_{j^*}}<\gapbar(S)\cdot \log\log U
        \end{align*}
        and $\gapbar(S)\ge(\log\log U)^{1/4}$, we must have $T_{j^*}\le (\gapbar(S)\cdot \log\log U)^2\le \gapbar(S)^{10}$, so $\log T_{j^*}=O(\log \gapbar(S))$. Therefore, we have
        \begin{align*}
            &\sum_{j=j^*}^{k-1}\min\bk*{1,\frac{\gapbar(S)}{T_{j+1}}}\cdot \log T_j \\
            \le{}&\sum_{j=j^*}^{k-1}\log T_j=\sum_{j=j^*}^{k-1}\log T_{j^*}/2^{j-j^*}\tag{definition of $T_j$} \\
            ={}&O(\log T_{j^*})=O(\log\gapbar(S)).
        \end{align*}
    \end{itemize}
    
    \item If $\gapbar(S)\le (\log\log U)^{1/4}$: Recall that $T_k\in (\sqrt{\log\log U},\log\log U]$, so $T_{k-2}\ge T_k^4\ge \gapbar(S)\cdot \log\log U$. Therefore, if we define $j^*$ similarly as in the previous case, then $j^*\ge k-2$, and the levels before $j^*$ contribute at most $O(1)$ to the total sum, as before. It remains to bound the last two terms $j=k-2,k-1$. When $j=k-2$ or $j=k-1$, we have 
    \begin{align*}
        &\min\bk*{1,\frac{\gapbar(S)}{T_{j+1}}}\cdot \log T_j\le \frac{\gapbar(S)}{T_{j+1}}\cdot \log T_j \\
        \le{}&\frac{(\log\log U)^{1/4}}{T_{j+1}}\cdot \log T_j \\
        \le{}&\frac{(\log\log U)^{1/4}}{\sqrt{\log\log U}}\cdot \log T_j\tag{$T_{j+1}\ge T_k>\sqrt{\log\log U}$} \\
        \le{}&\frac{(\log\log U)^{1/4}}{\sqrt{\log\log U}}\cdot 16\log \log\log U \tag{$T_j\le T_{k-2}\le (((\log\log U)+1)^2+1)^2$}\\
        ={}&O(1).
    \end{align*}
\end{itemize}

Therefore, \cref{equ:multiple-thresholds-goal} holds and the total expected space cost of our new distributor is at most
\begin{align*}
    O(|P|\cdot \poly\log U+ \varepsilon\cdot|S|\cdot\gapbar(S)+|S|\log\gapbar(S)),
\end{align*}
as desired.

\section{Proof of the Main Result}

\label{sec:final}

Finally, we prove the main result of this paper.

\MainDict*

Recall the structure of the warmup: We partition the key set into chunks of $\poly\log U$ keys. We then use a distributor (\cref{thm:distributor}) to map keys to chunks, and use a small dictionary (\cref{lem:polylog-dict-warmup}) to solve each chunk. The proof of \cref{thm:main} is similar to the warmup, except that we replace the distributor and small dictionaries with the improved constructions (\cref{thm:better-distributor} and \cref{lem:trie-plus-btree}).

\subsection{Improved Dictionary for Chunks}

Similarly to the warmup, \cref{lem:trie-plus-btree} can be seen as solving the dictionary problem for small key sets in the special regime where $\varepsilon=1/|S|$ and $\gap(S)=\Theta(|S|\log U)$. Therefore, using the same reductions as in \cref{sec:chunk}, we can generalize \cref{lem:trie-plus-btree} as follows.

\begin{lemma}
    \label{lem:polylog-dict}
    Let $U,\varepsilon, S$ be as defined in \cref{thm:main}. If it is guaranteed that $\log U\le |S|=\poly\log U$, then there exists a legitimate data structure in the variable-length word RAM model with word size $w=100\log U$ that maintains $S$. In $O(\log\varepsilon^{-1}/\log\log\varepsilon^{-1})$ expected amortized time, the data structure supports the following operations:
    \begin{itemize}
        \item \textbf{Insertion}: Insert a key $x$ to $S$.
        \item \textbf{Deletion}: Delete a key $x$ from $S$.
        \item \textbf{Predecessor/Successor}: Given $x\in [U]$, output the predecessor or successor of $x$ in $S$ (or $\bot$ if it does not exist).
    \end{itemize}
    The data structure also supports an operation called $\textbf{deallocate}$, that deactivates every active word in $O(|S|)$ time. The data structure has address limit $U^{90}$. It uses
    \begin{align*}
        \gap(S)+O(\varepsilon \cdot\gap(S))+O(|S|\log \gapbar(S))
    \end{align*}
    bits of space in expectation, and uses at most $O(|S|\log U)$ bits of space in the worst case. The data structure assumes access to lookup tables and hash functions of total size at most $U^{\delta}$ bits, where $\delta>0$ is an arbitrary parameter.
\end{lemma}

The only difference between \cref{lem:polylog-dict} and its warmup version \cref{lem:polylog-dict-warmup} is that, because we use \cref{lem:trie-plus-btree} instead of the warmup version \cref{lem:difference-encode-array}, \cref{lem:polylog-dict} achieves a better time bound and a slightly worse space bound (expected instead of worst case). The formal proof of \cref{lem:polylog-dict} is omitted, as it is the same as in \cref{sec:chunk}.

\subsection{The Data Structure}

Similarly to the warmup, our data structure consists of the following parts: A buffer dictionary, a distributor, a dictionary maintaining the set of pivots, and a set of small dictionaries.

Updates are first handled using the buffer dictionary. Then, when the buffer dictionary is too large (at least $|S|/\poly\log U$) or when we need to delete a key that is not in the buffer dictionary, we access the rest of the data structure. We first query the distributor for the chunk that contains the queried key, then access the corresponding small dictionary for that chunk.

We now discuss the differences between the new construction and the warmup. Comparing the interfaces of our new constructions \cref{thm:better-distributor} and \cref{lem:polylog-dict} with the warmup constructions \cref{thm:distributor} and \cref{lem:polylog-dict-warmup}, we can see that the interface for the small dictionaries is exactly the same as before. As for the distributor, the new construction requires more advice when inserting or deleting a key $x$ in $S$: Namely, we need to provide the predecessor and successor of $x$ in $S$\footnote{The new distributor also requires that any two pivots differ by at least $\log U$ which was not required in the warmup, but this is naturally satisfied by our construction because each chunk contains $\ge \log U$ keys.}.

In general, the predecessor and successor of a key in $S$ cannot be computed time-efficiently. However, we have additional information that can help us. Whenever we insert or delete a key $x$ in $S$, we know which chunk $x$ is in: For insertions, they only happen when we clear the buffer dictionary, in which case we perform a predecessor query for $x$ on the set of pivots. For deletions, we can know the chunk containing $x$ by querying the distributor.

Given this information, the predecessor and successor of $x$ in $S$ are computed as follows: We first perform predecessor and successor queries in the chunk containing $x$, using $O(\log\varepsilon^{-1}/\log\log\varepsilon^{-1})$ time. If the predecessor or successor does not exist, then we go to the previous chunk or the next chunk (which can be looked up in the dictionary for pivots), and perform a predecessor or successor query there. Note that the dictionary maintaining the pivots also maintains the predecessor and successor of each pivot, which is updated in $O(\log U)$ time whenever we update the set of pivots, and supports $O(1)$ time queries. In total, computing the advice needed for the distributor takes $O(\log\varepsilon^{-1}/\log\log\varepsilon^{-1})$ time.

\subsection{Performance Analysis}

\paragraph{Time efficiency.} Each access to the distributor takes $O(\log\log\varepsilon^{-1})$ time, which is slower than before. However, this is hidden behind the cost of accessing the small dictionaries, so the expected amortized query time is simply $O(\log\varepsilon^{-1}/\log\log\varepsilon^{-1})$.

\paragraph{Space efficiency.} The buffer dictionary and the dictionary maintaining the set of pivots only use $o(|S|)$ bits of space. It remains to analyze the distributor and the small dictionaries.

Letting $\{S_i\}$ denote the collection of chunks, the overall \emph{expected} space cost is

\begin{align*}
    &\bk*{\sum_{i=1}^{O(|S|/\log^{C_{\text{chunk}}}U)}\gap(S_i)+O(\varepsilon\cdot\gap(S_i))+O(|S_i|\log\gapbar(S_i))}+O(\varepsilon\cdot\gap(S)+|S|\log\gapbar(S))+O(U^{\delta}) \\
    \le{}&\gap(S)+O(\varepsilon\cdot\gap(S))+O(|S|\log\gapbar(S))+O(U^{\delta}).
\end{align*}
Here, the $O(|S_i|\log\gapbar(S_i))$ terms sum up to at most $O(|S|\log\gapbar(S))$ because of Jensen's inequality.

\subsection{Achieving a Worst-Case Space Bound}

Note that both the distributor and the small dictionaries only achieve an expected space bound, mainly because we need to randomly shift the keys. We obtain a worst-case space bound on the entire data structure by partitioning the key set into $\poly\log U$ subsets of size $\Omega(|S|/\poly\log U)$ each, and using the expected-space dictionary described above. This partitioning is maintained similarly to the chunks. We then build a dictionary for each subset, and use a \emph{fusion tree} to maintain the smallest keys of each subset. It is important that the number of subsets is only $\poly\log U$, because if it were to be larger, the fusion tree would not be able to support $O(1)$ time queries. Furthermore, $\poly\log U$ subsets is sufficient for us to apply a Chernoff bound\footnote{Note that for both the distributor and the small dictionaries, it is guaranteed that they use $O(|S|\log U)$ bits of space in the worst case, which allows us to apply the Chernoff bound.} on the total space, which means that the space bound of \cref{thm:main} holds for the entire data structure not just in expectation, but with high probability in $U$. Then, in the unlikely event that the total space becomes too large, we simply rebuild the entire data structure in $O(U^2)$ time. Finally, we apply \cref{clm:word_RAM_to_vl} to simulate the variable-length data structure in the word RAM model, while keeping the space and (asymptotic) time guarantees. This concludes the proof of \cref{thm:main}.

\subsection{Supporting Rank/Select or Predecessor/Successor Queries}

Our dictionary can also be made to answer more complicated queries such as rank/select and predecessor/successor, while keeping the time and space guarantees.

\CorMainPredecessor*

\begin{proof}
    Recall that the construction of \cref{thm:main} consists of a buffer dictionary, a distributor and small dictionaries for the chunks. Note that we can afford to perform predecessor/successor queries during updates. As a result, we no longer need the buffer dictionary (which existed because we needed to answer predecessor queries on the set of pivots \emph{in batches}). This simplifies the following exposition.

    In addition to \cref{thm:main}, we build a predecessor data structure on the set of pivots. Using dynamic van Emde Boas trees \cite{DBLP:journals/mst/BoasKZ77}, we can implement this in $O(\log\log U)$ time and $O(|P|\log U)=o(|S|)$ bits of space\footnote{The standard implementation needs $O(U)$ bits of space, but can be reduced using hashing.}. To answer a predecessor/successor query on $x$, we find the chunk containing $x$ by querying the predecessor data structure, then run predecessor/successor queries in this chunk and its neighboring chunks.
\end{proof}

\CorMainRank*

\begin{proof}
    In addition to the previous data structure, we also build the data structure of \cite{patrascu2014dynamic} on the set of pivots, which performs rank/select operations in $O(\log|P|/\log \log U)$ time and uses $O(|P|\log U)$ bits of space. Note that we actually need \emph{weighted} rank/select queries on the pivots, which can be supported by a simple extension of \cite{patrascu2014dynamic}. Here, the weights correspond to the sizes of the chunks, which are at most $\poly\log U$. We also need to augment the small dictionaries \cref{lem:polylog-dict} with the ability to answer rank/select queries, which is easy and omitted.
\end{proof}

\section{Lower Bound}

\label{sec:lb}

In this section, we prove a space lower bound on difference-encoded dictionaries. Our lower bound exactly matches the upper bound of \cref{thm:main} (up to an additive $O(|S|\log\gapbar(S))$ term), despite the fact that the upper bound constructs a dynamic dictionary and the lower bound is for \emph{static} dictionaries. Moreover, we emphasize that our lower bound holds for dictionaries that can only answer \emph{membership} queries.

\begin{theorem}
    \label{thm:lb_main}
    Let $n\ge 1, U\ge n^3$ be integers, and let $\Omega(\log\log U/\log U)\le \varepsilon<1/4$ be a parameter. Suppose that there exists a randomized static data structure in the cell-probe model with word size $w=\Theta(\log U)$ that stores a set $S\subset[U]$ of $n$ keys, and uses $(1+O(\varepsilon))\gap(S)$
    bits of space in the worst case. Then the data structure answers membership queries in $\Omega(\log\varepsilon^{-1}/\log\log\varepsilon^{-1})$ expected time in the worst case.
\end{theorem}

By ``answers queries in $\Omega(\log\varepsilon^{-1}/\log\log\varepsilon^{-1})$ expected time in the worst case'', we mean that there exists an input set $S$ and a key $x\in [U]$, such that when initialized with $S$, the data structure answers the membership query on $x$ in $\Omega(\log\varepsilon^{-1}/\log\log\varepsilon^{-1})$ expected time.

The space bound of \cref{thm:lb_main} is different from \cref{thm:main} in that it did not contain the
\begin{align*}
    O(|S|\log \gapbar(S))+O(U^{\delta})
\end{align*}
terms. It is fine to ignore the $O(|S|\log \gapbar(S))$ term, because it does not matter when $\varepsilon\cdot\gap(S)\ge \Omega(|S|\log \gapbar(S))$, and when $\varepsilon\cdot\gap(S)=o(|S|\log \gapbar(S))$, it is not possible to construct a difference-encoded dictionary, as explained in the intro, so the lower bound is vacuously true. As for the $U^{\delta}$ term, it is negligible as long as $n=U^{\Omega(1)}$.

\subsection{A Reformulation of the Dictionary Problem}

To prove the lower bound, we take a carefully designed hard distribution over possible inputs, and bound the expected space usage of any time-efficient dictionary on this distribution. Instead of describing the hard distribution directly, we formulate a new data structure problem that characterizes the behavior of the dictionary on the hard distribution.

\paragraph{The (Static) Path Verification Problem.} In the path verification problem with integer parameters $k,d,l\ge 1$, the input is a collection of $l$-bit strings $\{Y_{i,s}\}$, where $1\le i\le k$ is an integer and $s\in \{0,1\}^{\le d}$ is a non-empty bit string. We are responsible for storing the input $\{Y_{i,s}\}$, and answering the following queries:

\begin{itemize}
    \item \textsc{Verify}$(i',s',Y'_1,\dots,Y'_d)$: The input consists of an integer $1\le i'\le k$, a $d$-bit string $s'\in \{0,1\}^d$ and $d$ $l$-bit strings $Y'_1,\dots,Y'_d\in \{0,1\}^l$. We are asked to return \textsc{true} if for every $1\le j\le d$, we have $Y_{i',s'_{1\dots j}}=Y'_j$, and \textsc{false} otherwise.
\end{itemize}

Intuitively, the input consists of $k$ trees, where each tree is a full binary tree of depth $d$ with an $l$-bit string attached to each edge. For each query, we are given a leaf from one of the trees and $d$ $l$-bit strings, and are asked to verify whether the values on the path from the leaf to its root match the input bit strings (hence the name of the problem).

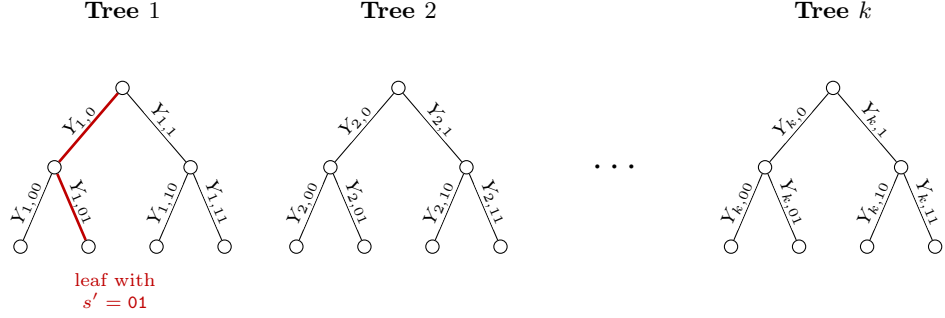
\begin{figure}[H]
\centering
\begin{tikzpicture}[
  font=\footnotesize,
  nd/.style={circle, draw, minimum size=5pt, inner sep=0pt},
  el/.style={inner sep=0.5pt, fill=white, font=\scriptsize},
  treetitle/.style={font=\small\bfseries},
]
  \begin{scope}[xshift=-3.65cm]
    \node[treetitle, anchor=south] at (0,0.8) {Tree $1$};
    \node[nd] (r1) at (0,0) {};
    \node[nd] (10) at (-0.9,-1.05) {};
    \node[nd] (11) at (0.9,-1.05) {};
    \node[nd] (100) at (-1.35,-2.1) {};
    \node[nd] (101) at (-0.45,-2.1) {};
    \node[nd] (110) at (0.45,-2.1) {};
    \node[nd] (111) at (1.35,-2.1) {};
    \draw (r1) -- (10) node[el, midway, sloped, above] {$Y_{1,0}$};
    \draw (r1) -- (11) node[el, midway, sloped, above] {$Y_{1,1}$};
    \draw (10) -- (100) node[el, midway, sloped, above] {$Y_{1,00}$};
    \draw (10) -- (101) node[el, midway, sloped, above] {$Y_{1,01}$};
    \draw (11) -- (110) node[el, midway, sloped, above] {$Y_{1,10}$};
    \draw (11) -- (111) node[el, midway, sloped, above] {$Y_{1,11}$};
    \draw[red!75!black, line width=1pt] (r1) -- (10) -- (101);
    \node[red!75!black, font=\scriptsize, align=center] at (-0.1,-2.68) {leaf with\\$s'=\texttt{01}$};
  \end{scope}

  \begin{scope}[xshift=0cm]
    \node[treetitle, anchor=south] at (0,0.8) {Tree $2$};
    \node[nd] (r2) at (0,0) {};
    \node[nd] (20) at (-0.9,-1.05) {};
    \node[nd] (21) at (0.9,-1.05) {};
    \node[nd] (200) at (-1.35,-2.1) {};
    \node[nd] (201) at (-0.45,-2.1) {};
    \node[nd] (210) at (0.45,-2.1) {};
    \node[nd] (211) at (1.35,-2.1) {};
    \draw (r2) -- (20) node[el, midway, sloped, above] {$Y_{2,0}$};
    \draw (r2) -- (21) node[el, midway, sloped, above] {$Y_{2,1}$};
    \draw (20) -- (200) node[el, midway, sloped, above] {$Y_{2,00}$};
    \draw (20) -- (201) node[el, midway, sloped, above] {$Y_{2,01}$};
    \draw (21) -- (210) node[el, midway, sloped, above] {$Y_{2,10}$};
    \draw (21) -- (211) node[el, midway, sloped, above] {$Y_{2,11}$};
  \end{scope}

  \node[font=\Large] at (2.9,-1.05) {$\cdots$};

  \begin{scope}[xshift=5.75cm]
    \node[treetitle, anchor=south] at (0,0.8) {Tree $k$};
    \node[nd] (rk) at (0,0) {};
    \node[nd] (k0) at (-0.9,-1.05) {};
    \node[nd] (k1) at (0.9,-1.05) {};
    \node[nd] (k00) at (-1.35,-2.1) {};
    \node[nd] (k01) at (-0.45,-2.1) {};
    \node[nd] (k10) at (0.45,-2.1) {};
    \node[nd] (k11) at (1.35,-2.1) {};
    \draw (rk) -- (k0) node[el, midway, sloped, above] {$Y_{k,0}$};
    \draw (rk) -- (k1) node[el, midway, sloped, above] {$Y_{k,1}$};
    \draw (k0) -- (k00) node[el, midway, sloped, above] {$Y_{k,00}$};
    \draw (k0) -- (k01) node[el, midway, sloped, above] {$Y_{k,01}$};
    \draw (k1) -- (k10) node[el, midway, sloped, above] {$Y_{k,10}$};
    \draw (k1) -- (k11) node[el, midway, sloped, above] {$Y_{k,11}$};
  \end{scope}

\end{tikzpicture}
\caption{Illustration of an input of the path verification problem for $d=2$: the input consists of $k$ full binary trees. Edges are indexed by the bitstring~$s$ of the child reached from the root; each stores an $l$-bit label~$Y_{i,s}$. A \textsc{Verify} query specifies a tree~$i'$, a leaf path $s'\in\{0,1\}^d$, and candidate strings $Y'_1,\dots,Y'_d$, and checks that they match the labels along the root-to-leaf path (e.g.\ the red path checks $Y'_1=Y_{1,0}$ and $Y'_2=Y_{1,01}$ when $s'=\texttt{01}$ and $i'=1$).}
\label{fig:path-verification}
\end{figure}

For the path verification problem, we have the following lower bound. In the following, let
\begin{align*}
    \text{OPT}_{k,d,l}\defeq 2k\cdot (2^d-1)\cdot l
\end{align*}
denote the information-theoretically optimal space for storing the input of the path verification problem with parameters $k,d,l$.

\begin{restatable}{lemma}{lbPath}
    \label{lem:lb_path_v}
    Let $k,d,l\ge 1$ be integers, such that $dl=\omega(2^d)$ and $\log k=O(dl)$. Suppose that there exists a randomized static data structure in the cell-probe model with word size $w=\Theta(dl)$ that solves the path verification problem with parameters $k,d,l$ and uses $(1+O(1/2^d))\cdot \text{OPT}_{k,d,l}$ bits of space in the worst case. Then this data structure answers queries using $\Omega(d/\log d)$ expected time in the worst case.
\end{restatable}

The meaning of ``using $\dots$ expected time in the worst case'' is the same as in \cref{thm:lb_main}.

\paragraph{Reducing the dictionary problem to the path verification problem.} Before proving \cref{lem:lb_path_v}, we show that lower bounds on difference-encoded dictionaries are indeed implied by lower bounds on the path verification problem.

\begin{lemma}
    \label{lem:lb_reduction}
    Let $n\ge 1, U\ge n^3$ be integers, and let $\Omega(\log\log U/\log U)\le\varepsilon<1/4$ be a parameter. If \cref{lem:lb_path_v} holds on parameters $d=\lfloor \log \varepsilon^{-1}\rfloor,k=\lfloor n/2^d\rfloor,l=\lfloor(\lfloor \log U\rfloor-\lceil\log k\rceil)/d\rfloor-1$, then \cref{thm:lb_main} holds on parameters $n,U,\varepsilon$.
\end{lemma}

\begin{proof}
    We first verify the conditions of \cref{lem:lb_path_v}.
    When initialized as described, we have $2^d=\Theta(\varepsilon^{-1})$, $k=\Theta(n/\varepsilon^{-1})$ and
    \begin{align*}
        l&=\lfloor(\lfloor \log U\rfloor-\lceil\log k\rceil)/d\rfloor-1 \\
        &=\Theta(\log U/d)-1\tag{$k\le n\le U^{1/3}$} \\
        &=\Theta(\log U/d).\tag{$d\le \log\log U$}
    \end{align*}
    
    Therefore, $dl=\Theta(\log U)$ and $2^d=O(\varepsilon^{-1})=O(\log U/\log\log U)$, so we have $dl=\omega(2^d)$ and $\log k=O(dl)$. Moreover, the required word size of \cref{lem:lb_path_v} is $w=\Theta(dl)=\Theta(\log U)$, which is the same as that of \cref{thm:lb_main} on parameters $n,U,\varepsilon$.
    
    In the following, we prove the contrapositive of the lemma: Let $w=\Theta(dl)=\Theta(\log U)$. In the cell-probe model with word size $w$, given a dictionary $D$ storing $n$ keys from $[U]$ that uses $(1+O(\varepsilon))\gap(S)$ bits of space and answers membership queries using $T$ probes, we can construct a data structure that solves the path verification problem with parameters $k,d,l$ using $(1+O(1/2^d))\cdot \text{OPT}_{k,d,l}$ bits of space and answers queries using $T$ probes. This proves the lemma because, if \cref{lem:lb_path_v} holds on parameters $d,k,l$, then $T=\Omega(d/\log d)=\Omega(\log \varepsilon^{-1}/\log\log \varepsilon^{-1})$, so \cref{thm:lb_main} holds on parameters $n,U,\varepsilon$.

    \paragraph{The construction.} We solve the path verification problem by building a dictionary that stores a set of $k\cdot 2^d\le n$ keys $\{x_{i,s}\}$. For each $1\le i\le k,s\in \{0,1\}^d$, letting $a$ denote the length-$\lceil\log k\rceil$ binary representation of $i$, the binary representation of $x_{i,s}$ is
    \begin{align*}
        a\circ s_1 \circ Y_{i,s_{1\dots 1}} \circ \dots \circ s_d \circ Y_{i,s_{1\dots d}}.
    \end{align*}
    Note that the length of this binary representation is
    \begin{align*}
        \lceil \log k\rceil+d\cdot (l+1)\le \lfloor\log U\rfloor,
    \end{align*}
    by definition of $l$.

    On query $\textsc{Verify}(i',s',Y'_1,\dots,Y'_d)$, letting $a$ denote the length-$\lceil\log k\rceil$ binary representation of $i'$, we answer the query by querying the dictionary on the key whose binary representation is
    \begin{align*}
        a\circ s'_1 \circ Y'_1 \circ \dots \circ s'_d \circ Y'_d,
    \end{align*}
    i.e.\ the same interleaving as in the definition of the keys $\{x_{i,s}\}$, but with $s'$ in place of $s$ and with each block $Y_{i',s'_{1\dots j}}$ replaced by $Y'_j$.
    
    \paragraph{Correctness.} When answering $\textsc{Verify}(i',s',Y'_1,\dots,Y'_d)$, by our construction, the queried key must be distinct from any $x_{i,s}$ for which $(i,s)\ne (i',s')$, so the dictionary exactly reports whether the queried key is equal to $x_{i,s}$. This happens if and only if $Y'_j=Y_{i',s'_{1\dots j}}$ for every $j\in\{1,\dots,d\}$, which matches the definition of $\textsc{Verify}$.

    \paragraph{Efficiency.} Since we answer the query by directly accessing the dictionary, the total number of probes is exactly the same as the dictionary.
    
    It remains to bound the space usage. Since the dictionary uses $(1+O(\varepsilon))\gap(S)$ bits of space and $\varepsilon=O(1/2^d)$, we only have to show that $\gap(\{x_{i,s}\})\le(1+O(1/2^d))\cdot \text{OPT}_{k,d,l}$. Recall the definition of the trie entropy \eqref{equ:trie_entropy}, in particular, recall that $\gap(S)\le \trie_U(S)+|S|$ for any $S\subset[U]$. Therefore, to upper bound $\gap(\{x_{i,s}\})$, it suffices to upper bound $\trie_U(\{x_{i,s}\})$, which is the number of edges in the trie corresponding to $\{x_{i,s}\}$. \cref{fig:lb-reduction-trie} illustrates the trie in question. Note that each edge for some $Y_{i,s}$ in the figure corresponds to a path of $l$ edges in the trie.

    \begin{figure}[H]
\centering
\begin{tikzpicture}[
  font=\footnotesize,
  nd/.style={circle, draw, minimum size=5pt, inner sep=0pt},
  el/.style={inner sep=0.5pt, fill=white, font=\scriptsize},
]
  \node[nd] (R) at (0,0) {};

  \node[nd] (A0) at (-3.35,-0.58) {};
  \node[nd] (A1) at (2.25,-0.58) {};
  \draw (R) -- (A0) node[el, midway, sloped, above] {$\texttt{0}$};
  \draw (R) -- (A1) node[el, midway, sloped, above] {$\texttt{1}$};

  \node[nd] (T1) at (-5.3,-1.32) {};
  \node[nd] (T2) at (-1.4,-1.32) {};
  \node[nd] (T3) at (1.65,-1.32) {};
  \node[nd] (T4) at (2.85,-1.32) {};
  \draw (A0) -- (T1) node[el, midway, sloped, above] {$\texttt{0}$};
  \draw (A0) -- (T2) node[el, midway, sloped, above] {$\texttt{1}$};
  \draw (A1) -- (T3) node[el, midway, sloped, above] {$\texttt{0}$};
  \draw (A1) -- (T4) node[el, midway, sloped, above] {$\texttt{1}$};

  \node[nd] (B0) at (-6.0,-1.74) {};
  \node[nd] (B1) at (-4.6,-1.74) {};
  \draw (T1) -- (B0) node[el, midway, sloped, above] {$\texttt{0}$};
  \draw (T1) -- (B1) node[el, midway, sloped, above] {$\texttt{1}$};

  \node[nd] (J0) at (-6.55,-2.79) {};
  \node[nd] (J1) at (-4.05,-2.79) {};
  \draw (B0) -- (J0) node[el, midway, sloped, above] {$Y_{1,0}$};
  \draw (B1) -- (J1) node[el, midway, sloped, above] {$Y_{1,1}$};

  \node[nd] (C00) at (-7.05,-3.24) {};
  \node[nd] (C01) at (-6.05,-3.24) {};
  \node[nd] (C10) at (-4.55,-3.24) {};
  \node[nd] (C11) at (-3.55,-3.24) {};
  \draw (J0) -- (C00) node[el, midway, sloped, above] {$\texttt{0}$};
  \draw (J0) -- (C01) node[el, midway, sloped, above] {$\texttt{1}$};
  \draw (J1) -- (C10) node[el, midway, sloped, above] {$\texttt{0}$};
  \draw (J1) -- (C11) node[el, midway, sloped, above] {$\texttt{1}$};

  \node[nd] (L00) at (-7.35,-4.29) {};
  \node[nd] (L01) at (-5.85,-4.29) {};
  \node[nd] (L10) at (-4.75,-4.29) {};
  \node[nd] (L11) at (-3.25,-4.29) {};
  \draw (C00) -- (L00) node[el, midway, sloped, above] {$Y_{1,00}$};
  \draw (C01) -- (L01) node[el, midway, sloped, above] {$Y_{1,01}$};
  \draw (C10) -- (L10) node[el, midway, sloped, above] {$Y_{1,10}$};
  \draw (C11) -- (L11) node[el, midway, sloped, above] {$Y_{1,11}$};

  \draw[red!75!black, line width=1pt] (R) -- (A0) -- (T1) -- (B1) -- (J1) -- (C10) -- (L10);
  \node[red!75!black, font=\scriptsize, align=center] at (-4.55,-5.05) {leaf\\$a{=}\texttt{00}$, $s{=}\texttt{10}$};

  \node[font=\Large] at (-1.4,-2.05) {$\cdots$};
  \node[font=\Large] at (1.65,-2.05) {$\cdots$};
  \node[font=\Large] at (2.85,-2.05) {$\cdots$};
  \node[font=\scriptsize, align=center] at (0.85,-3.85) {analogous\\subtries\\for $i{=}2,3,4$};
\end{tikzpicture}
\caption{Binary trie of the $k\cdot 2^d$ keys stored in the reduction from \cref{lem:lb_reduction}, illustrated for $d=2$ and $k=4$ (so $\lceil\log k\rceil=2$ bits for~$a$). Edge labels $\texttt{0}$ and $\texttt{1}$ denote single bits of~$a$ or of the path~$s$; each $Y_{i,s}$ marks an edge of length $l$ bits in the key (cf.\ \cref{fig:path-verification}). The red path reaches the third leaf in the leftmost drawn subtrie (tree code $a=\texttt{00}$, leaf path $s=\texttt{10}$).}
\label{fig:lb-reduction-trie}
\end{figure}

    From \cref{fig:lb-reduction-trie}, we see that the trie consists of the following parts:
    \begin{itemize}
        \item The bits that make up the codes $a$, which use $O(k)$ edges in total.
        \item For each pair of $1\le i\le k$ and non-empty $s\in\{0,1\}^{\le d}$, there is one bit followed by an $l$-bit string $Y_{i,s}$. These use $2k\cdot(2^d-1)\cdot (l+1)$ edges in total.
    \end{itemize}
    The total number of trie edges is
    \begin{align*}
        O(k)+2k\cdot(2^d-1)\cdot (l+1) = (1+O(1/l))\cdot \text{OPT}_{k,d,l}.
    \end{align*}
    Since $d=O(\log U/\log\log U)$ and $l=\Theta(\log U/d)$, we have $l=\Omega(\log U/\log\log U)=\Omega(\varepsilon^{-1})=\Omega(2^d)$. Therefore, the trie entropy is at most $(1+O(1/l))\cdot \text{OPT}_{k,d,l}$, and
    \begin{align*}
        \gap(\{x_{i,s}\})&\le (1+O(1/l))\cdot \text{OPT}_{k,d,l}+k\cdot 2^d \\
        &\le (1+O(1/l))\cdot \text{OPT}_{k,d,l}\tag{$k\cdot 2^d=O(\text{OPT}_{k,d,l}/l)$} \\
        &=(1+O(1/2^d))\cdot \text{OPT}_{k,d,l},
    \end{align*}
    as desired.
\end{proof}

\subsection{Outer Lemma}

In this section, we prove the lower bound on the path verification problem. We adopt the framework of~\cite{patrascu2010cellprobe}, and tweak the arguments to only focus on the positive queries (i.e., the queries that return $\textsc{true}$).

\lbPath*

\paragraph{Notations and definitions.} Fix some constant $T$ that is independent of $k,d,l$, and let $\eta>0$ be a sufficiently small constant that only depends on $T$. Assume for the sake of contradiction that there exists a data structure $D$ that uses $(1+T\cdot (1/2^d))\cdot \text{OPT}_{k,d,l}$ bits of space, and answers queries using $t<\eta\cdot d/\log d$ expected probes in the worst case. We will show that such $D$ cannot exist for $d\ge T+10$ (for $d<T+10$, the bound $t<\eta\cdot d/\log d$ is trivial since we can take $\eta$ to be sufficiently small). In the following, the big-O notations do not hide dependencies on $\eta$ and $T$.

For $1\le i\le k,s\in \{0,1\}^d$, let $\text{Probe}_{D}(i,s)$ denote the words probed when answering the \emph{positive} query $\textsc{Verify}(i,s,Y_{i,s_{1\dots 1}},\dots,Y_{i,s_{1\dots d}})$ on data structure $D$. For a set of tuples $S=\{(i,s):1\le i\le k,s\in \{0,1\}^d\}$, let $\text{Probe}_{D}(S)$ denote the union of the sets $\text{Probe}_{D}(i,s)$ where $(i,s)\in S$. Given a set of words $C$, let $\text{Cont}_D(C)$ be the concatenation of the contents in $C$, which is a bit string of length $w\cdot |C|$.

Following \cite{patrascu2010cellprobe}, our lower bound is going to construct data structures with \textbf{published bits}. That is, in addition to the normal memory, the data structure also has access to a string of published bits, which the query algorithm can access at no cost. For simplicity, we assume that the published bits also encode their own size, so that when reading the published bits, the query algorithm knows when to stop.

\paragraph{The round elimination process.} We prove the lower bound by showing that we can convert $D$ into a data structure that has $<\text{OPT}_{k,d,l}$ published bits and answers any \emph{positive} query using $0$ probes, which contradicts the information-theoretic lower bound for encoding the input. Similarly to \cite{patrascu2010cellprobe}, we need the following inner lemma.

\begin{restatable}{lemma}{InnerLem}
    \label{lem:inner}
    Let $k,d,l\ge 1$ be integers. Let $1\le j^*\le d$ be an integer. Suppose that there exists a data structure $D$ in the cell-probe model with word size $w=\omega(2^d)$ for the path verification problem with parameters $k,d,l$, that uses $\le \text{OPT}_{k,d,l}$ bits of normal space along with $< \text{OPT}_{k,j^*,l}/10$ published bits. Let $Q\defeq\{(i,s):1\le i\le k,s\in \{0,1\}^{j^*}\circ 0^{d-j^*}\}$. We have that
    \begin{align*}
        \Pr_{1\le i\le k,s\in \{0,1\}^d}[\text{Probe}_D(i,s)\cap\text{Probe}_D(Q)\ne \emptyset]\ge \frac{1}{10},
    \end{align*}
    where the probability is over the randomness of the positive query $(i,s)$ and the input $\{Y_{i,s}\}$.
\end{restatable}

Note that we can apply \cref{lem:inner} on our data structure because in our case, $w=\Theta(dl)=\omega(2^d)$.

The proof of \cref{lem:inner} is deferred to \cref{sec:inner}. Here, we design the round elimination process using \cref{lem:inner}.

Initially, we publish the last $p_0=T\cdot\text{OPT}_{k,d,l}/2^d$ bits of $D$, so that $D$ only uses $\text{OPT}_{k,d,l}$ bits of normal space (allowing us to apply \cref{lem:inner}). We call this new data structure $D_0$. 

Our process consists of at most $10t$ rounds. For each round $r=1\dots 10t$, we start with the data structure $D_{r-1}$ that answers \emph{positive} queries using $\le t-(r-1)/10$ probes on average (in expectation over a uniformly random input), and has $p_{r-1}$ published bits. We find the smallest level $1\le j_r\le d$ such that $p_{r-1}<\text{OPT}_{k,j_r,l}/10$ (that is, $2k\cdot(2^{j_r}-1)\cdot l=\Theta(p_{r-1})$). If there does not exist such a level, we terminate the round elimination process. Let
\begin{align*}
    C_r\defeq\text{Probe}_{D_{r-1}}(\{(i,s):1\le i\le k,s\in \{0,1\}^{j_r}\circ 0^{d-j_r}\}).
\end{align*}
We publish the words in $C_r$ (i.e., both the addresses and the contents) in $D_{r-1}$, and call the resulting data structure $D_r$. In $D_r$, when answering a query, we will not probe a word in the normal memory if it's already published. By \cref{lem:inner}, for $D_{r-1}$, at least $1/10$ of all positive queries probe some word in $C_r$. Therefore, the average number of probes needed to answer a positive query on $D_r$ is at most $t-r/10$, i.e., $1/10$ less than that of $D_{r-1}$.

\paragraph{Bounding the number of rounds.} Note that $C_r$ is the union of the probe sequences of $k\cdot 2^{j_r}$ positive queries, so $|C_r|\le t\cdot k\cdot 2^{j_r}$. The number of published bits in $D_r$ is thus
\begin{align*}
    p_r&=p_{r-1}+O(w\cdot t\cdot k\cdot 2^{j_r}) \\
    &=p_{r-1}+O(dl\cdot (d/\log d)\cdot k\cdot 2^{j_r}\cdot p_{r-1}/(2k\cdot(2^{j_r}-1)\cdot l)) \\
    &=O(d^2/\log d)\cdot p_{r-1}.
\end{align*}
Here we used the fact that $t<d/\log d$ (since $\eta$ is sufficiently small), so the big-O notation does not depend on $\eta$.

The only times that the round elimination process terminates are when $r=10t$ and when $p_{r-1}\ge \text{OPT}_{k,d,l}/10$ at some round $r$. Since $t<\eta\cdot d/\log d$, we have
\begin{align*}
    p_{10t}&=(O(d^2/\log d))^{10t}\cdot T\cdot(\text{OPT}_{k,d,l}/2^d) \\
    &\le 2^{10\eta\cdot (d/\log d)\cdot O(\log (d^2/\log d))}\cdot T\cdot(\text{OPT}_{k,d,l}/2^d) \\
    &=\text{OPT}_{k,d,l}\cdot T\cdot(2^{O(\eta d)}/2^d) \\
    &<\text{OPT}_{k,d,l}/10. \tag{$\eta$ is sufficiently small given $T$, and $d\ge T+10$}
\end{align*}
Therefore, we will only terminate when $r=10t$, at which time we have published $p_{10t}<\text{OPT}_{k,d,l}$ bits, so the final data structure $D_{10t}$ answers positive queries using $0$ probes. This is a contradiction. Therefore, we must have $t\ge \eta\cdot d/\log d$.

\subsection{Inner Lemma}
\label{sec:inner}

Finally, it remains to prove the inner lemma.

\InnerLem*

\begin{proof}
    Assume for the sake of contradiction that \cref{lem:inner} does not hold. In this case, we design an impossibly efficient one-way communication protocol that allows Alice to send the entire input $\{Y_{i,s}\}$ to Bob using fewer than $\text{OPT}_{k,d,l}$ bits.

    \paragraph{Notations.} Suppose that $D$ is a counterexample of the lemma, and define $k,d,l,j^*,Q$ as in the statement of the lemma. Let $P$ denote the published bits of $D$, and let $M\le\lfloor\text{OPT}_{k,d,l}/w\rfloor$ denote the number of words in $D$.
    
    Define
    \begin{align*}
        Q_{+}\defeq\{(i,s):1\le i\le k,s\in \{0,1\}^d,\text{Probe}_D(i,s)\cap \text{Probe}_D(Q)\ne\emptyset\}
    \end{align*}
    to be the set of positive queries whose probe sequence intersects $\text{Probe}_D(Q)$, and $Q_{-}\defeq\{(i,s):1\le i\le k,s\in \{0,1\}^d\}\setminus Q_{+}$ denote all the remaining positive queries. Note in particular that $Q_+$ contains $Q$. Since $D$ is a counterexample of the lemma, we have $\Pr_{1\le i\le k,s\in \{0,1\}^d}[(i,s)\in Q_+]< 1/10$.

    \paragraph{The communication protocol.} Initially, Alice knows the input $\{Y_{i,s}\}$. She builds the data structure $D$ on $\{Y_{i,s}\}$ and sends the following information to Bob in the listed order:
    \begin{enumerate}
        \item The published bits $P$, using $<\text{OPT}_{k,j^*,l}/10$ bits.
        \item The set of words $\text{Probe}_D(Q_-)$. Formally, we send a bitmap of $M$ bits denoting whether each word belongs to $\text{Probe}_D(Q_-)$.
        \item The contents $\text{Cont}_D(\text{Probe}_D(Q_-))$ of the words in $\text{Probe}_D(Q_-)$, using $w\cdot |\text{Probe}_D(Q_-)|$ bits.

        \item The values $Y^{\le j^*}\defeq \{Y_{i,s}:1\le i\le k,s\in \{0,1\}^{\le j^*}\}$, encoded optimally conditioned on previous messages (note the unusual range of $s$).

        \item The contents $\text{Cont}_D([M]\setminus\text{Probe}_D(Q_-))$ of the words outside of $\text{Probe}_D(Q_-)$, encoded optimally conditioned on previous messages.

        The costs of steps 4 and 5 are analyzed in more detail later.
    \end{enumerate}
    After receiving the above information, Bob can easily recover $\{Y_{i,s}\}$ because he knows the entire data structure $D$ from steps 1, 3 and 5. It remains to bound the communication cost.

    \paragraph{The communication cost of step 4.} The cost of directly sending $Y^{\le j^*}=\{Y_{i,s}:1\le i\le k,s\in \{0,1\}^{\le j^*}\}$ is exactly $\text{OPT}_{k,j^*,l}$ bits. However, we can do better than this by conditioning on the information sent in the previous steps.
    
    Note that, after step 3, Bob can privately simulate all queries $\textsc{Verify}(\cdot)$, and learn the identities of all the \emph{positive} queries that can be answered by only probing the words in $\text{Probe}_D(Q_-)$. Also note that, since $\text{Probe}_D(Q_-)$ is disjoint from $\text{Probe}_D(Q)$, the positive queries that Bob can identify are exactly the positive queries in $Q_-$. In other words, after step 3, Bob knows $Q_-$ as well as the values $Y_{i,s_{1\dots j}}$ for every $(i,s)\in Q_-$ and $1\le j\le d$.

    Given this knowledge, in step 4, Alice does not need to send the value of some $Y_{i,s}$ if there exists $(i,s')\in Q_-$ for which $s$ is a prefix of $s'$ (in which case Bob already knows $Y_{i,s}$). That is, the fraction of values $Y_{i,s}$ that Alice needs to send in step 4 is
    \begin{align*}
        &\Pr_{1\le i\le k,s\in \{0,1\}^{\le j^*}}\Bk*{\text{for any } s'\in\{0,1\}^d \text{ such that }s\text{ is a prefix of }s',(i,s')\notin Q_-} \\
        \le{}&\E_{1\le i\le k,s\in \{0,1\}^{\le j^*}}\Bk*{\Pr_{s'\in \{0,1\}^d,s\text{ is a prefix of }s'}\Bk*{(i,s')\notin Q_-}} \\
        \le{}&\Pr_{1\le i\le k,s'\in \{0,1\}^{\le d}}\Bk*{(i,s')\notin Q_-},
    \end{align*}
    which is smaller than $1/10$ since
    \begin{align*}
        \Pr_{1\le i\le k,s\in \{0,1\}^d}[(i,s)\notin Q_-]< \frac{1}{10}
    \end{align*}by our assumption that $D$ is a counterexample. Therefore, the cost of step 4 is smaller than $(1/10)\cdot \text{OPT}_{k,j^*,l}$ bits.

    \paragraph{The communication cost of step 5.} Let $C\defeq [M]\setminus\text{Probe}_D(Q_-)$ for notational simplicity. The cost of step 5 under an optimal encoding is at most
    \begin{align*}
        H(\text{Cont}_D(C)\mid P,C,Y^{\le j^*}),
    \end{align*}
    where $P$ is given by step 1, $C=[M]\setminus\text{Probe}_D(Q_-)$ is given by step 2 and $Y^{\le j^*}$ is given by step 4. We can rewrite this cost as
    \begin{align}
        \label{equ:cost_step_5}
        &H(\text{Cont}_D(C)\mid P,C,Y^{\le j^*}) \nonumber\\
        ={}&H(\text{Cont}_D(C),Y^{\le j^*}\mid P,C) -H(Y^{\le j^*}\mid P,C) \nonumber\\
        ={}&H(\text{Cont}_D(C)\mid P,C)+H(Y^{\le j^*}\mid\text{Cont}_D(C),P,C)-H(Y^{\le j^*}\mid P,C) \nonumber\\
        \le{}&w\cdot |C|+H(Y^{\le j^*}\mid\text{Cont}_D(C),P,C)-H(Y^{\le j^*})+H(C)+H(P) \nonumber\\
        <{}&w\cdot (M-|\text{Probe}_D(Q_-)|)+H(Y^{\le j^*}\mid\text{Cont}_D(C),P,C)-\text{OPT}_{k,j^*,l}+M+\text{OPT}_{k,j^*,l}/10.
    \end{align}
    
    All other terms are known, so it remains to bound $H(Y^{\le j^*}\mid\text{Cont}_D(C),P,C)$. Observe that the set of words $C=[M]\setminus\text{Probe}_D(Q_-)$ contains $\text{Probe}_D(Q)$. Indeed, this is because no positive query in $Q_-$ probes any word in $\text{Probe}_D(Q)$. Therefore, given $\text{Cont}_D(C),P,C$, we can simulate the queries as in step 4, and learn everything about the positive queries in $Q$. In particular, by definition of $Q$, for every $1\le i\le k,s\in \{0,1\}^{\le j^*}$, we always have $(i,s\circ 0^{d-\text{length}(s)})\in Q$, which means that we know the value of $Y_{i,s}$ from the simulation. Therefore, $H(Y^{\le j^*}\mid\text{Cont}_D(C),P,C)=0$ and \eqref{equ:cost_step_5} is equal to
    \begin{align*}
        &w\cdot (M-|\text{Probe}_D(Q_-)|)-\text{OPT}_{k,j^*,l}+M+\text{OPT}_{k,j^*,l}/10 \\
        ={}&w\cdot M-w\cdot |\text{Probe}_D(Q_-)|-(9/10)\cdot\text{OPT}_{k,j^*,l}+M.
    \end{align*}

    \paragraph{Putting everything together.} Summing up the cost of every step, the communication protocol costs
    \begin{align*}
        &(1/10)\cdot\text{OPT}_{k,j^*,l}+M+w\cdot |\text{Probe}_D(Q_-)| \\
        +{}&(1/10)\cdot\text{OPT}_{k,j^*,l}+w\cdot M-w\cdot |\text{Probe}_D(Q_-)|-(9/10)\cdot \text{OPT}_{k,j^*,l}+M \\
        ={}&w\cdot M+2M-(7/10)\cdot \text{OPT}_{k,j^*,l} \\
        \le{}&\text{OPT}_{k,d,l}+2(\text{OPT}_{k,d,l}/w)-(7/10)\cdot \text{OPT}_{k,j^*,l}\tag{by definition of $M$} \\
        \le{}&\text{OPT}_{k,d,l}+2(\text{OPT}_{k,d,l}/w)-(7/20)\cdot(\text{OPT}_{k,d,l}/2^{d-j^*})\tag{by definition of $\text{OPT}$} \\\le{}&\text{OPT}_{k,d,l}+2(\text{OPT}_{k,d,l}/w)-(7/10)\cdot(\text{OPT}_{k,d,l}/2^{d})\tag{$j^*\ge 1$} \\
        <{}&\text{OPT}_{k,d,l}\tag{$w=\omega(2^d)$}
    \end{align*}
    bits of communication, which is a contradiction.
\end{proof}

\section*{Acknowledgements} 

This research was funded by NSF grant CCF-2504471, NSF grant CCF-2119069, Jane Street Graduate Research Fellowship and MongoDB PhD Fellowship.

The use of AI tools in this paper is limited to checking the correctness of the results, checking the grammar and math formulas, and creating the figures.

\bibliographystyle{alpha}
\bibliography{reference.bib}

\appendix

\section{Tools for the Variable-Length Word RAM Model}
\label{app:vl_word_RAM}

In this appendix, we formally present some tools for the variable-length word RAM model that are used in the main text.

\paragraph{Relation between the word RAM model and the variable-length word RAM model.} The variable-length word RAM model is at least as powerful as the word RAM model. Indeed, any data structure in the word RAM model can also be implemented in the variable-length word RAM model (with a small space overhead).

\begin{claim}
    \label{clm:word_RAM_to_vl}
    Let $w$ be a parameter. Any data structure in the word RAM model with word size $w'=\Theta(w)$ that uses less than $2^w$ bits of space can be implemented as a legitimate data structure in the variable-length word RAM model with word size $w$. The new data structure has asymptotically the same running time as the old data structure for performing the operations, and its space usage is at most $(1+O(\log w/w))$ times the space usage of the old data structure. Furthermore, the address limit of the new data structure is at most
    \begin{align*}
        O((\text{maximum space usage})/w).
    \end{align*}
\end{claim}

On the other hand, using dynamic retrieval data structures, it is also possible to simulate the variable-length word RAM model in the classical word RAM model.

\begin{lemma}
    \label{thm:memory_allocate}
    Let $w$ be a parameter. There exists a data structure in the word RAM model with word size $w'=\Theta(w)$ that simulates the memory of the variable-length word RAM model with word size $w$. Assume the number of active words is always at least $\poly w$. Then, this data structure supports the operations listed in \cref{def:vl_word_RAM} in $O(1)$ time with high probability in the number of active words, and its space usage is at most the space usage of the memory. The data structure assumes access to lookup tables and hash functions of total size at most $2^{w\delta}$ bits, where $\delta>0$ is an arbitrary parameter.
\end{lemma}

\begin{proof}
    We simulate the variable-length word RAM memory as follows. Identify word addresses with integers in $[U]$, where $U=2^{w+1}$ (thus the number of active words is always at most $U/2$). Let $S\subseteq [U]$ denote the set of active words, and for each $x\in S$ let $s_x$ be the content stored at address~$x$ (a bit string of length at most~$w$ that encodes its own length). We maintain the resizable variable-length retrieval data structure from \cref{rem:retrieval-variable-length} on universe~$[U]$, storing the value $s_x$ (of length $|s_x|\le w=O(\log U)$) at each key $x\in S$.

    The four memory operations in \cref{def:vl_word_RAM} are implemented as follows:
    \begin{itemize}
        \item {Activate} word~$x$: insert the pair $(x,s_x)$ with $s_x$ the empty string.
        \item {Deactivate} word~$x$: delete~$x$.
        \item {Read} word~$x$: query the retrieval data structure at~$x$.
        \item {Write} content~$s$ to word~$x$: delete~$x$ and re-insert $(x,s)$.
    \end{itemize}
    Arithmetic and bitwise operations on $w$-bit integers are performed directly in the word RAM. Each simulated memory operation performs $O(1)$ retrieval operations, each of which takes $O(1)$ time with high probability in~$|S|$ by \cref{rem:retrieval-variable-length}; this yields $O(1)$ time per simulated operation.

    For space, \cref{rem:retrieval-variable-length} gives a bound of
    \begin{align*}
        \sum_{x\in S}|s_x|+O\bk*{|S|\log\log(U/|S|)}
    \end{align*}
    bits. Since $U=2^w$, we have $\log\log(U/|S|)\le \log w$, so the overhead is $O(|S|\log w)$. Choosing $C_{\mathrm{retr}}$ in \cref{def:vl_word_RAM} large enough, this is at most the space usage of the simulated memory.
\end{proof}

\paragraph{Combining multiple data structures.} One benefit of working in the variable-length word RAM model is that we can easily combine multiple variable-length data structures without having to worry about how to compactly pack their memories. We only have to make sure that these data structures never access the same words, which corresponds to limiting the total address limit of the combined data structure.

\begin{claim}
    \label{clm:virtual_address_combine}
    Given $B$ (``logical'') legitimate data structures in the variable-length word RAM model with word size $w$, such that the $i$-th data structure has address limit $M_i$ and $\sum_{i=1}^{B}M_i\le 2^w$, we can simulate them using a (``physical'') legitimate data structure in the variable-length word RAM model with word size $w$ whose space usage is equal to the total space usage of the logical data structures, and has address limit $\sum_{i=1}^{B}M_i$. The physical data structure supports operations to the memory of each logical data structure in unit time.
\end{claim}

\end{document}